\documentclass[a4paper,onecolumn,unpublished]{quantumarticle}
\pdfoutput=1

\usepackage[UKenglish]{babel}
\usepackage[T1]{fontenc}
\usepackage[utf8]{inputenc}
\usepackage{colorprofiles}
\usepackage{bookmark}   % More moder version of hyperref
\usepackage[doipre={doi:}]{uri}   % https://github.com/retorquere/zotero-better-bibtex/issues/1336

\usepackage{csquotes}

\usepackage{etoolbox}
\apptocmd{\sloppy}{\hbadness 10000\relax}{}{} % For wargings (in bibliography): Underfull \hbox (badness 10000) ...

\usepackage{mathtools}
\usepackage{physics}
\usepackage{amsfonts}
\usepackage{amsmath}
\usepackage{amssymb}
\usepackage{mathrsfs}   % cursive font
\usepackage{bigints}    % for the use of bra and ket
\usepackage{braket}
\usepackage{stmaryrd}   % doubled brackets
\SetSymbolFont{stmry}{bold}{U}{stmry}{m}{n}   % get rid of annoying warning from stmaryrd
\usepackage{extarrows} % \xLongrightarrow

\usepackage{enumerate} 

\usepackage{amsthm}
\usepackage{thmtools}
\usepackage{thm-restate}

\usepackage{placeins}

\usepackage{bm}   % for \bm

\usepackage{graphicx}
\usepackage{float}
\graphicspath{figures}
\usepackage{subcaption}
\usepackage{tikzfig}
\usepackage{pgfplots} % For style=braceedge in TikZiT
\pgfplotsset{compat=newest, compat/show suggested version=false} % Some warning
\usepackage{tikz-cd}
\usetikzlibrary{decorations.markings}
\pgfdeclaresnake{square wave qubit}{initial}
{
  \state{initial}[width=0.5pt]
  {
    \pgfpathlineto{\pgfpoint{0pt}{0.15pt}}
    \pgfpathlineto{\pgfpoint{0.25pt}{0.15pt}}
    \pgfpathlineto{\pgfpoint{0.25pt}{-0.15pt}}
    \pgfpathlineto{\pgfpoint{0.5pt}{-0.15pt}}
    \pgfpathlineto{\pgfpoint{0.5pt}{0pt}}
  }
  \state{final}
  {
    \pgfpathlineto{\pgfpoint{\pgfsnakeremainingdistance}{0pt}}
  }
}

\pgfdeclaresnake{square wave dual rail}{initial}
{
  \state{initial}[width=1pt]
  {
    \pgfpathlineto{\pgfpoint{0pt}{0.1pt}}
    \pgfpathlineto{\pgfpoint{0.5pt}{0.1pt}}
    \pgfpathlineto{\pgfpoint{0.5pt}{-0.1pt}}
    \pgfpathlineto{\pgfpoint{1pt}{-0.1pt}}
    \pgfpathlineto{\pgfpoint{1pt}{0pt}}
  }
  \state{final}
  {
    \pgfpathlineto{\pgfpoint{\pgfsnakeremainingdistance}{0pt}}
  }
}

\tikzstyle{hbox}=[draw=black, shape=rectangle, fill=yellow, minimum size=.55em, inner sep=0.15em, scale=0.85, font={\scriptsize}]
\tikzstyle{box}=[draw=black, shape=rectangle, fill=white, minimum size=1em, inner sep=0.2em, scale=0.85, font={\scriptsize}]
\tikzstyle{gn}=[draw=black, shape=circle, fill={zx_green}, draw=black, inner sep=0.7mm, minimum width=0pt, minimum height=0pt, tikzit fill={rgb,255: red,181; green,215; blue,181}]
\tikzstyle{rn}=[gn, fill={zx_red}, draw=black, tikzit fill={rgb,255: red,215; green,96; blue,96}]
\tikzstyle{wn}=[gn, fill=white, draw=black]
\tikzstyle{gn_phase}=[shape=rectangle, fill={zx_green}, draw=black, minimum size=1.2em, rounded corners=0.48em, inner sep=0.2em, outer sep=-0.2em, scale=0.8, font={\footnotesize}, tikzit shape=circle, tikzit fill={rgb,255: red,181; green,215; blue,181}]
\tikzstyle{rn_phase}=[{gn_phase}, fill={zx_red}, draw=black, tikzit fill={rgb,255: red,215; green,96; blue,96}]
\tikzstyle{phase}=[draw=black, shape=rectangle, fill=white, minimum size=.95em, inner sep=0.1em, scale=0.85, font={\scriptsize}]
\tikzstyle{phaseCircle}=[gn_phase, draw=black, fill=white]
\tikzstyle{black node}=[draw=black, shape=circle, scale=0.3, fill=black, font={\footnotesize}]
\tikzstyle{rtriang}=[shape=isosceles triangle, fill={gray!50}, draw=black, isosceles triangle stretches=true, inner sep=0.8pt, minimum width=.5mm, minimum height=.5mm]
\tikzstyle{ltriang}=[rtriang, shape=isosceles triangle, fill={gray!50}, draw=black, shape border rotate=180]
\tikzstyle{beamSplitter}=[draw=black, shape=rectangle, fill=white, minimum width=3mm, minimum height=.2mm, inner sep=.3mm]
\tikzstyle{sLabel}=[font={\scriptsize}, auto]
\tikzstyle{midArrow}=[postaction=decorate, decoration={markings, mark=at position 0.5 with {\arrow{stealth[length=6.4pt, sep=-2pt]}}}]

\tikzstyle{N}=[-, line width=0.9pt]
\tikzstyle{NStream}=[-, N, midArrow]
\tikzstyle{LOWire}=[-, very thin, tikzit draw={rgb,255: red,134; green,42; blue,43}]
\tikzstyle{NLO}=[-, line width=.6pt, tikzit draw={rgb,255: red,117; green,44; blue,83}]
\tikzstyle{LOStream}=[-, NLO, midArrow, tikzit draw={rgb,255: red,139; green,43; blue,126}]
\tikzstyle{DRWire}=[-, LOWire, double, double distance=.5pt, tikzit draw={rgb,255: red,24; green,33; blue,119}]
\tikzstyle{DRThick}=[-, DRWire, NLO, double distance=.5pt, tikzit draw={rgb,255: red,60; green,108; blue,166}]
\tikzstyle{DRStream}=[-, DRThick, midArrow, tikzit draw={rgb,255: red,0; green,175; blue,175}]
\tikzstyle{dashedE}=[-, dashed]
\tikzstyle{hadamard edge}=[-, dashed, draw=blue]
\tikzstyle{braceedge}=[-, decorate, decoration={brace, amplitude=2mm, raise=-1mm}]
\tikzstyle{arrow}=[<-]
\tikzstyle{MixedDR}=[-, tikzit draw={rgb,255: red,255; green,128; blue,0}, decorate, decoration=square wave dual rail, very thin, draw={rgb,255: red,255; green,128; blue,0}]
\tikzstyle{MixedQubit}=[-, tikzit draw={rgb,255: red,160; green,100; blue,219}, decorate, decoration=square wave qubit, very thin, draw={rgb,255: red,160; green,100; blue,219}]
\tikzstyle{ClWire}=[-, very thin, densely dashed, tikzit draw={rgb,255: red,255; green,191; blue,191}]

\definecolor{zx_grey}{RGB}{211,211,211}
\definecolor{zx_red}{RGB}{232,165,165}
\definecolor{zx_green}{RGB}{216,248,216}

\usepackage{xtab}
\usepackage{xspace}
\usepackage{tabu}
\usepackage{multicol}

\usepackage[capitalise, noabbrev]{cleveref}

\newcommand{\N}{\mathbb{N}}

\ifdef{\C}
  {\renewcommand{\C}{\mathbb{C}}}
  {\newcommand{\C}{\mathbb{C}}}
\newcommand{\minu}{\texttt{-}}
\newcommand{\plus}{\texttt{+}}

\newcommand{\interp}[1]{\left\llbracket#1\right\rrbracket}

\newcommand{\bR}{\begin{color}{red}}
\newcommand{\bB}{\begin{color}{blue}}
\newcommand{\bM}{\begin{color}{magenta}}
\newcommand{\bC}{\begin{color}{cyan}}
\newcommand{\bW}{\begin{color}{white}}
\newcommand{\bBl}{\begin{color}{black}}
\newcommand{\bG}{\begin{color}{green}}
\newcommand{\bY}{\begin{color}{yellow}}
\newcommand{\e}{\end{color}}

\newcommand{\tikzrefsize}[1]{\tiny{#1}}
\newcommand{\lemref}[1]{\tikzrefsize{\textsc{(Lem~\ref{#1})}}}

\newtheorem{theorem}{Theorem}[section]
\newtheorem{proposition}[theorem]{Proposition}

\newtheorem{lemma}[theorem]{Lemma}

\theoremstyle{definition}
\newtheorem{definition}[theorem]{Definition}
\newtheorem{example}[theorem]{Example}
\theoremstyle{remark}
\newtheorem{remark}[theorem]{Remark}

\newcommand{\sub}{\subseteq}
\renewcommand{\tt}[1]{\mathtt{#1}}
\renewcommand{\bf}[1]{\ensuremath{\mathbf{#1}}}
\renewcommand{\sc}[1]{\textsc{#1}}
\renewcommand{\cal}[1]{\mathcal{#1}}

\renewcommand{\phi}{\varphi}

\newcommand{\Odd}{\texttt{Odd}}
\newcommand{\comp}[1]{\overline{#1}}

\newcommand{\cvar}[1]{\underline{#1}}
\newcommand{\ccvar}[1]{#1}

\newcommand{\XY}{\normalfont XY\xspace}
\newcommand{\XZ}{\normalfont XZ\xspace}
\newcommand{\YZ}{\normalfont YZ\xspace}

\newcommand{\XYm}{\ensuremath\normalfont\textrm{XY}\xspace}
\newcommand{\XZm}{\normalfont\normalfont\textrm{XZ}\xspace}
\newcommand{\YZm}{\normalfont\normalfont\textrm{YZ}\xspace}
\newcommand{\Xm}{\ensuremath\normalfont\textrm{X}\xspace}
\newcommand{\Ym}{\normalfont\normalfont\textrm{Y}\xspace}
\newcommand{\Zm}{\normalfont\normalfont\textrm{Z}\xspace}

\newcommand{\Fusion}{\tiny{(\textsc{Spider})}}
\newcommand{\Bigebra}{\tiny{(\textsc{Bigebra})}}
\newcommand{\ZElim}{\tiny{(\textsc{Z-Elim})}}
\newcommand{\XElim}{\tiny{(\textsc{X-Elim})}}

\newcommand{\Colour}{\tiny{(\textsc{Color})}}
\newcommand{\Copy}{\tiny{(\textsc{Copy})}}
\newcommand{\One}{\tiny{(\textsc{One})}}
\newcommand{\Euler}{\tiny{(\textsc{Euler})}}
\newcommand{\PiCommute}{\tiny{($\pi$)}}

\graphicspath{figures}
\allowdisplaybreaks
\usepackage[numbers,sort&compress]{natbib}
\usepackage{soul}

\title{Fusion and flow: formal protocols to reliably build photonic graph states}

\author[1]{Giovanni de Felice}
\author[1]{Boldizsár Poór}
\author[1, 2]{Lia Yeh}
\author[2]{William Cashman}
\affil[1]{Quantinuum, 17 Beaumont Street, Oxford, OX1 2NA, United Kingdom}
\affil[2]{University of Oxford, United Kingdom}

\begin{document}

\maketitle

\begin{abstract}
    
Photonics offers a promising platform for implementations of measurement-based quantum computing.
Recently proposed fusion-based architectures aim to achieve universality and fault-tolerance.
In these approaches, computation is carried out by performing fusion and single-qubit measurements on a resource graph state.
The verification of these architectures requires linear algebraic, probabilistic, and control flow structures to be combined in a unified formal language.
This paper develops a framework for photonic quantum computing by bringing together linear optics, ZX calculus, and dataflow programming.
We characterize fusion measurements that induce Pauli errors and show that they are correctable using a novel flow structure for fusion networks.
We prove the correctness of new repeat-until-success protocols for the realization of arbitrary fusions 
and provide a graph-theoretic proof of universality for linear optics with entangled photon sources.
The proposed framework paves the way for the development of compilation algorithms for photonic quantum computing.

\end{abstract}

\tableofcontents 

\section{Introduction}\label{sec:introduction}

Photonics has been the platform of reference for implementation of measurement-based quantum computing (MBQC),
from the early teleportation experiments~\cite{bouwmeesterExperimentalQuantumTeleportation1997, boschiExperimentalRealizationTeleporting1998} to the realization of photonic cluster states~\cite{thomasEfficientGenerationEntangled2022, coganDeterministicGenerationIndistinguishable2023}.
Fusion-based approaches are among the most promising proposals for achieving a universal and fault-tolerant quantum computer based on photonics~\cite{bartolucciFusionbasedQuantumComputation2023, degliniastySpinOpticalQuantumComputing2024}.
In these approaches, graph states are constructed by performing fusion measurements on an underlying resource graph state.
Since fusion is an entangling two-qubit measurement~\cite{browneResourceEfficientLinearOptical2005}, the topology of the underlying graph is modified as the computation proceeds.
Fusion is also a stochastic process which induces decoherence in case of failure~\cite{bartolucciCreationEntangledPhotonic2021}, giving an additional level of complexity to the computation.

Due to the probabilistic nature of quantum measurements, MBQC requires feedforward of classical information.
Depending on previous measurement outcomes, corrections in the form of Pauli gates must be performed before proceeding with the computation~\cite{danosDeterminismOnewayModel2006}. 
Flow structure describes conditions on graph states under which these corrections result in a deterministic computation~\cite{browneGeneralizedFlowDeterminism2007}.
While these flow conditions are well understood for single-qubit measurements on a fixed graph state~\cite{browneGeneralizedFlowDeterminism2007},
they have not been studied in the photonic fusion-based setting which features multi-qubit measurements and where graph state preparation and measurement are intertwined processes.
In order to realize fusion-based quantum computing, we need a notion of flow describing how undesired measurement outcomes can be corrected to achieve a deterministic computation.

The ZX calculus~\cite{coeckeInteractingQuantumObservables2008} has proved to be an effective tool across different models of quantum computation,
including measurement-based~\cite{duncanRewritingMeasurementBasedQuantum2010, backensThereBackAgain2021, mcelvanneyFlowpreservingZXcalculusRewrite2023},
circuit-based~\cite{coeckePicturingQuantumProcesses2017, kissingerReducingNumberNonClifford2020},
and fault-tolerant~\cite{beaudrapZXCalculusLanguage2020, kissingerPhasefreeZXDiagrams2022, huangGraphicalCSSCode2023, townsend-teagueFloquetifyingColourCode2023} quantum computing.
In the context of photonic quantum computing, the ZX calculus is starting to be applied in a top-down direction to design novel error correction codes~\cite{bombinUnifyingFlavorsFault2023, pankovichFlexibleEntangledState2023}.
However, the bottom-up direction, from the hardware described in the language of linear optics~\cite{aaronsonComputationalComplexityLinear2011, clementLOvCalculusGraphicalLanguage2022, defeliceQuantumLinearOptics2023, heurtelCompleteGraphicalLanguage2024} 
to the logic of ZX calculi, has not been established formally.

The other most critical current barrier to understanding photonic MBQC is that only static pictures of computations are used.
This limitation stifles a key advantage to optical implementations of quantum computing: their ability to use both spatial and temporal degrees of freedom of photons to achieve a given computation~\cite{motesScalableBoson2014,madsenQuantumComputational2022}.
Optical circuits are dynamic processes acting on streams of photons, containing components such as routers, delays, switches, and time-delayed emitters.
A formalism to represent them should therefore have the expressivity to specify what each component does at each time step.
A time-dependent dataflow language~\cite{caretteGraphicalLanguageDelayed2021, dilavoreMonoidalStreamsDataflow2022} 
would enable rewriting and recursive reasoning for the wide range of experimental setups available in the optical setting.

This work arises from an effort to bridge the physics literature on photonic quantum computing with computer science
concepts such as control flow, formal semantics, rewriting and compositional reasoning.
A mathematical framework capturing different aspects of this exciting field of research, 
has been developed to enable the large-scale verification and compilation of photonic quantum protocols.

In \cref{sec:background}, we establish a precise relationship between the ZX calculus and linear optics.
We obtain a representation of photonic fusion measurements as a mixed sum of two ZX diagrams, 
corresponding to the success and failure outcomes of the measurement. 
The amplitudes and probabilities can then be computed by graphical rewriting or diagram evaluation.

Our first novel contribution, in \cref{sec:character}, is a characterization of all measurements locally equivalent to Type II fusion.
We classify as \enquote{green failure} all such measurements which preserve entanglement of any graph state under the fusion failure outcome.
We thereafter narrow down which of these induce correctable Pauli errors in both success and failure outcomes; see \cref{prop:characterisation}.
Interestingly, this does not have to be Clifford, and particularly includes phase gadgets, an important class of parametrized entangling gates~\cite{kissingerReducingNumberNonClifford2020}.
By additionally restricting the single-qubit measurement to be Pauli, the two most commonly used fusion measurements
--- the Type II fusion of~\cite{browneResourceEfficientLinearOptical2005,bartolucciFusionbasedQuantumComputation2023} and the CZ fusion of~\cite{limRepeatUntilSuccessLinearOptics2005,degliniastySpinOpticalQuantumComputing2024} ---
naturally arise from our characterization as $X$ and $Y$ fusions, respectively.

Several compilation frameworks for fusion-based quantum computing (FBQC) have recently been proposed~\cite{zilkCompilerUniversalPhotonic2022,zhangOneQCompilationFramework2023,leeGraphtheoreticalOptimizationFusionbased2023}.
Our second contribution, in Section 4, is the development of a theory of flow for FBQC
which supports these frameworks by enabling the correction of errors induced by fusion measurements.
Our results leverage the rich literature on flow in MBQC~\cite{browneGeneralizedFlowDeterminism2007,danosMeasurementCalculus2007,simmonsRelatingMeasurementPatterns2021,backensThereBackAgain2021},
generalizing it to a setting featuring two-qubit measurements rather than only single-qubit measurements.
This is applicable to any fusion network with $X$ fusions, $Y$ fusions, or a mix of the two.
We show that the corresponding patterns can always be factorized such that every fusion appears before single-qubit measurements; see \cref{thm:flow-pattern}.
We, moreover, show that any decomposition of an open graph as an $XY$-fusion network has $XY$-flow, provided that the original open graph has Pauli flow; see \cref{thm:flow-simplified-graph}.
This constitutes the first formal analysis of the errors induced by photonic fusion measurements, providing specific flow conditions to enable their correction.

Experimental realizations of photonic quantum computing require physical components ---
such as routers, ultra-fast optical switches, and time-delayed emitters ---
that control the interactions of photons in an optical setup.
While these components have been studied in isolation~\cite{bartolucciSwitchNetworksPhotonic2021,clementPBScalculusGraphical2020},
compilation efforts had not explored how to put together such components to design verifiable architectures.
\cref{sec:protocols} develops the first graphical language for optical protocols that gives formal semantics to each of the above-mentioned components.
Our language features both coherent and classical control and enables us to reason about the time-evolution of optical protocols
via both inductive and coinductive reasoning.
The diagrams in this calculus closely resemble their physical implementation as an experimental setup.
The combination of an \enquote{unrolling} procedure and graphical rewriting allows us to efficiently produce a
ZX diagram representing the computation that has been executed after a given number of time-steps.
The semantics of this calculus is formally underpinned by the theory of coalgebra and monoidal streams~\cite{dilavoreMonoidalStreamsDataflow2022}.
It bears a close relationship to the graphical language of~\cite{caretteGraphicalLanguageDelayed2021}
while substantially extending it to handle the optical setting.

As an application of this language, we prove the correctness of new repeat-until-success (RUS)
protocols that boost the probability of success of fusion measurements; see \cref{thm:RUS}.
This generalizes a result of~\cite{limRepeatUntilSuccessLinearOptics2005} about RUS optical CZ gates,
to arbitrary fusion measurements with green failure.
In particular, it shows that the probability of success of Type II fusion measurements can also be boosted by a RUS protocol using an $n$-ary GHZ state as resource.
The proof uses an unrolling technique which implements inductive reasoning and, to the best of our knowledge, 
constitutes the first entirely graphical proof of this kind.

We conclude \cref{sec:universality} by bringing together the results of the previous three sections to give a novel proof of universality for a simple fusion-based architecture.
The architecture, inspired by the recent spin-optical approach presented in~\cite{degliniastySpinOpticalQuantumComputing2024}, 
is based only on linear optics, active switching, classical feedforward and a single quantum dot emitter.
Rather than relying on the generation of a universal graph state~\cite{gimeno-segoviaThreePhotonGreenbergerHorneZeilingerStates2015},
our proof is constructive in the sense that it shows how any given MBQC pattern can be implemented by a sequence of instructions for a specific optical circuit.

Our results give a picture of the compilation process in fusion-based architectures as a translation between four concepts: graphs, networks, patterns and protocols.
The \emph{graph} gives a high-level specification of the quantum computation which directly corresponds to a ZX diagram.
The \emph{network} is a decomposition of the target graph indicating where to perform fusion measurements.
The \emph{pattern} specifies the computation as a sequence of operations including measurements and corrections.
And lastly, the \emph{protocol} gives its physical implementation in the form of instructions for an optical setup.
Using both recursion and rewriting, we are able to efficiently verify that a protocol implements a given quantum computation.
Our focus is on the formal description of the compilation rather than the optimization of each step, which we leave for future work.

\begin{figure}[ht]
  \noindent
    \[
      \scalebox{.93}{\input{IntroFigFin.tikz}}
    \]
    \label{fig:intro-fig}
    \caption{An overview of the main concepts studied in this paper. The arrows indicate the direction of compilation:
    an open graph is decomposed as a fusion network and a corresponding fusion pattern, which is then compiled as instructions for an optical setup.}
\end{figure}

\section{Background}\label{sec:background}

In this section, we review the background required to understand photonic implementations of MBQC.
We start by introducing the graphical notations used in this paper, allowing us to relate
linear optical circuits and ZX-diagrams via the dual-rail encoding.
We use this language to study the primitive operations of FBQC: fusion measurements and resource state generators.
We end by reviewing the literature on flow structure in MBQC\@.

\subsection{ZX calculus}

The ZX calculus is a graphical language used to represent and reason about quantum computation.
Its elementary building blocks are the green \emph{Z-spider} and the red \emph{X-spider} (therefore ZX).
\begin{align}
  \label{eq:z-spider-interp}
  \input{ZX/generators/z-spider.tikz}
  \quad &\overset{\interp{\cdot}}{\longmapsto} \quad
  \ket{0}^{\otimes n}\! \bra{0}^{\otimes m} + e^{i \alpha} \ket{1}^{\otimes n}\! \bra{1}^{\otimes m}\\
  \label{eq:x-spider-interp}
  \input{ZX/generators/x-spider.tikz}
  \quad &\overset{\interp{\cdot}}{\longmapsto} \quad
  \ket{+}^{\otimes n}\! \bra{+}^{\otimes m} + e^{i \alpha} \ket{-}^{\otimes n}\! \bra{-}^{\otimes m}
\end{align}
These spiders can have any number of input and output legs, corresponding to qubit ports, and they have a phase parameter $\alpha$.
Notice that since the $e^{i \alpha}$ function in the interpretation is $2 \pi$ periodic, we can take the parameter of spiders modulo $2 \pi$.
The last generator of the ZX calculus is the yellow \emph{Hadamard box} $\begin{tikzpicture}
	\begin{pgfonlayer}{nodelayer}
		\node [style=none] (0) at (-1, 0) {};
		\node [style=none] (1) at (1, 0) {};
		\node [style=hbox] (2) at (0, 0) {};
	\end{pgfonlayer}
	\begin{pgfonlayer}{edgelayer}
		\draw (0.center) to (2);
		\draw (2) to (1.center);
	\end{pgfonlayer}
\end{tikzpicture}
$ that corresponds to the Hadamard gate, $\ket{0}\bra{+}\ +\ \ket{1}\bra{-}$.
We define the star symbol as the diagram corresponding to $\frac{1}{\sqrt 2}$, that is, $\star \coloneqq \begin{tikzpicture}
	\begin{pgfonlayer}{nodelayer}
		\node [style=gn] (0) at (-0.625, 0) {};
		\node [style=rn] (1) at (0.625, 0) {};
	\end{pgfonlayer}
	\begin{pgfonlayer}{edgelayer}
		\draw (1) to (0);
		\draw [bend left=45, looseness=0.75] (0) to (1);
		\draw [bend left=45, looseness=0.75] (1) to (0);
	\end{pgfonlayer}
\end{tikzpicture}
$.
Using these building blocks, we are able to intuitively represent common quantum gates as well as any unitary.
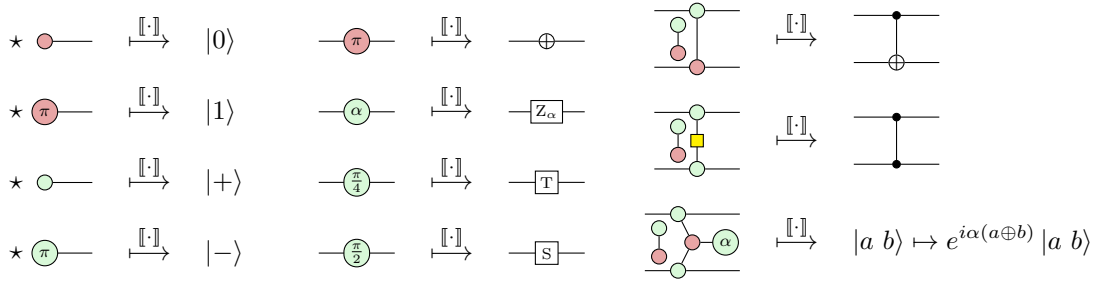
\begin{figure}[ht]
  \noindent
  \begin{minipage}{.28\textwidth}
    \begin{align*}
      \input{ZX/elements/x-state-zero.tikz}
      \quad &\overset{\interp{\cdot}}{\longmapsto} \quad
      \ket{0} \\
      \input{ZX/elements/x-state-pi.tikz}
      \quad &\overset{\interp{\cdot}}{\longmapsto} \quad
      \ket{1} \\
      \input{ZX/elements/z-state-zero.tikz}
      \quad &\overset{\interp{\cdot}}{\longmapsto} \quad
      \ket{+} \\
      \input{ZX/elements/z-state-pi.tikz}
      \quad &\overset{\interp{\cdot}}{\longmapsto} \quad
      \ket{-}
    \end{align*}
  \end{minipage}
  \begin{minipage}{.28\textwidth}
    \begin{align*}
      \input{ZX/elements/not.tikz}
      \quad &\overset{\interp{\cdot}}{\longmapsto} \quad
      \input{QuantumCircuit/not.tikz} \\
      \input{ZX/elements/z-rotate-alpha.tikz}
      \quad &\overset{\interp{\cdot}}{\longmapsto} \quad
      \input{QuantumCircuit/z-rotate-alpha.tikz} \\
      \input{ZX/elements/t.tikz}
      \quad &\overset{\interp{\cdot}}{\longmapsto} \quad
      \input{QuantumCircuit/t.tikz} \\
      \input{ZX/elements/s.tikz}
      \quad &\overset{\interp{\cdot}}{\longmapsto} \quad
      \input{QuantumCircuit/s.tikz}
    \end{align*}
  \end{minipage}
  \begin{minipage}{.44\textwidth}
    \begin{align*}
      \input{ZX/elements/cnot.tikz}
      \quad &\overset{\interp{\cdot}}{\longmapsto} \quad
      \input{QuantumCircuit/cnot.tikz} \\
      \input{ZX/elements/cz.tikz}
      \quad &\overset{\interp{\cdot}}{\longmapsto} \quad
      \input{QuantumCircuit/cz.tikz} \\
      \input{ZX/elements/PhaseGadgetUnitary.tikz}
      \quad &\overset{\interp{\cdot}}{\longmapsto} \quad
      \ket{a\ b} \mapsto e^{i \alpha (a \oplus b)}\ket{a\ b}
    \end{align*}
  \end{minipage}
\end{figure}

Further to its ability to represent any unitary, the ZX calculus also comes with a \enquote{built-in} set of graphical rewrite rules that correspond to elementary interactions between its generators.
In fact, these rewrite rules are powerful enough to derive any equation that holds for qubit maps~\cite{ngUniversalCompletionZXcalculus2017,jeandelDiagrammaticReasoningClifford2018}.

\begin{remark}
  \label{rem:scalar-zx}
  There are many different axiomatizations of the ZX calculus.
  The version we present maintains exact equality for the stabilizer fragment.
  Versions of the ZX calculus can be found in~\cite{jeandelCompleteAxiomatisationZXCalculus2018} for the Clifford+T fragment, and in~\cite{vilmartNearMinimalAxiomatisationZXCalculus2019} for the full language.
\end{remark}

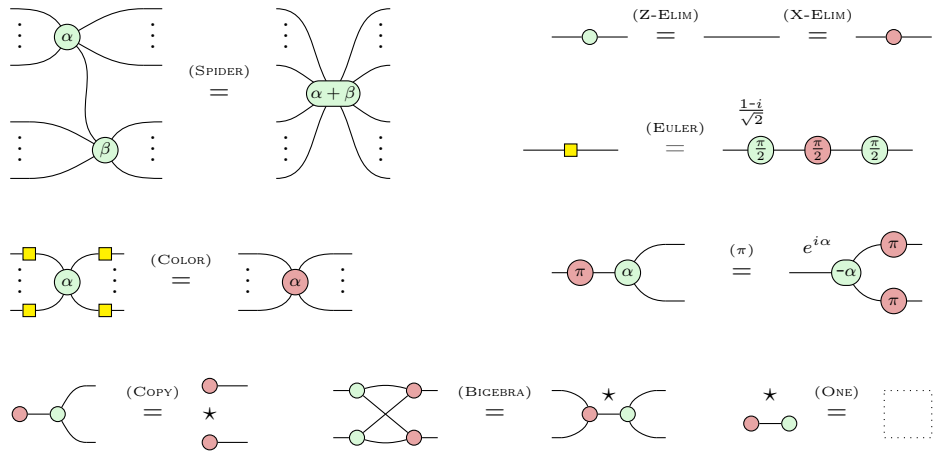
\[
  \input{ZX/axioms.tikz}
\]
A ZX-diagram with no inputs or outputs corresponds to some scalar, which we sometimes choose to omit from calculations.
We use $=$ when a rewrite preserves equality on the nose, and $\approx$ denotes equality \emph{up to some scalar}.

The above axioms can then be used to prove equalities between ZX-diagrams.
% we can show that two consecutive Hadamard gates cancel each other out and result in the identity:
% \[
%   \tikzfig{ZX/HadSquareID}
% \]
% Similarly, one can show that non-Pauli, Clifford states that can be expressed both as a red and a green spider:
% \[
%   \tikzfig{ZX/state-change}
% \]
% Another useful rule that 
For example, we can show the Hopf-rule, which allows us to disconnect spiders of different colors when connected with two wires:
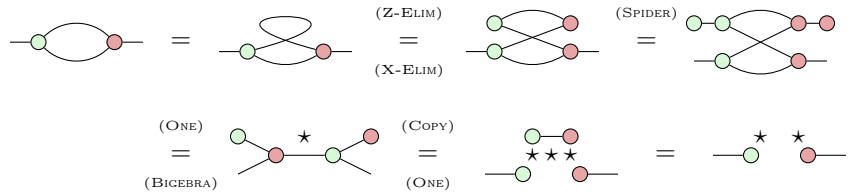
\[
  \input{ZX/Hopf.tikz}
\]

\subsection{Linear optical circuits}\label{sec:LO}

Linear optical circuits are the basic ingredients of photonic computing.
They are generated by two physical gates: beam splitters and phase shifts.
When equipped with $n$-photon state preparations and number-resolving photon detectors, they give rise to quantum statistics~\cite{aaronsonComputationalComplexityLinear2011}.
We depict the above-mentioned components, respectively, as follows:
\[
  \begin{tikzpicture}
	\begin{pgfonlayer}{nodelayer}
		\node [style=none] (0) at (-1, 0.5) {};
		\node [style=none] (1) at (-1, -0.5) {};
		\node [style=none] (2) at (1, -0.5) {};
		\node [style=none] (3) at (1, 0.5) {};
		\node [style=beamSplitter] (4) at (0, 0) {};
		\node [style=none] (5) at (-1.25, 0.5) {};
		\node [style=none] (6) at (1.25, 0.5) {};
		\node [style=none] (7) at (1.25, -0.5) {};
		\node [style=none] (8) at (-1.25, -0.5) {};
	\end{pgfonlayer}
	\begin{pgfonlayer}{edgelayer}
		\draw [style=LOWire, in=0, out=-180, looseness=0.75] (3.center) to (1.center);
		\draw [style=LOWire, in=180, out=0, looseness=0.75] (0.center) to (2.center);
		\draw [style=LOWire] (8.center) to (1.center);
		\draw [style=LOWire] (5.center) to (0.center);
		\draw [style=LOWire] (6.center) to (3.center);
		\draw [style=LOWire] (7.center) to (2.center);
	\end{pgfonlayer}
\end{tikzpicture}

  \qquad \qquad
  \begin{tikzpicture}
	\begin{pgfonlayer}{nodelayer}
		\node [style=none] (0) at (-1.25, 0) {};
		\node [style=none] (1) at (1.25, 0) {};
		\node [style=phaseCircle] (2) at (0, 0) {$\theta$};
	\end{pgfonlayer}
	\begin{pgfonlayer}{edgelayer}
		\draw [style=LOWire] (1.center) to (2);
		\draw [style=LOWire] (2) to (0.center);
	\end{pgfonlayer}
\end{tikzpicture}

  \qquad \qquad
  \begin{tikzpicture}
	\begin{pgfonlayer}{nodelayer}
		\node [style=black node] (0) at (-0.75, 0) {};
		\node [style=none] (1) at (0.75, 0) {};
		\node [style=sLabel] (2) at (-0.75, 0.375) {$n$};
	\end{pgfonlayer}
	\begin{pgfonlayer}{edgelayer}
		\draw (1.center) to (0);
	\end{pgfonlayer}
\end{tikzpicture}

  \qquad \qquad
  \begin{tikzpicture}
	\begin{pgfonlayer}{nodelayer}
		\node [style=black node] (0) at (0.75, 0) {};
		\node [style=none] (1) at (-0.75, 0) {};
		\node [style=sLabel] (2) at (0.75, 0.375) {$\cvar n$};
	\end{pgfonlayer}
	\begin{pgfonlayer}{edgelayer}
		\draw (1.center) to (0);
	\end{pgfonlayer}
\end{tikzpicture}

\]
We call the graphical language generated by the above components $\bf{LO}$.
Graphical axiomatisations of these components can be found in~\cite{clementLOvCalculusGraphicalLanguage2022, defeliceQuantumLinearOptics2023, heurtelCompleteGraphicalLanguage2024}.
We use a special notation $\begin{tikzpicture}
	\begin{pgfonlayer}{nodelayer}
		\node [style=none] (0) at (0.75, 0) {};
		\node [style=wn] (1) at (-0.25, 0) {};
	\end{pgfonlayer}
	\begin{pgfonlayer}{edgelayer}
		\draw (1) to (0.center);
	\end{pgfonlayer}
\end{tikzpicture}
 \coloneqq \begin{tikzpicture}
	\begin{pgfonlayer}{nodelayer}
		\node [style=none] (0) at (0.75, 0) {};
		\node [style=black node] (1) at (-0.25, 0) {};
		\node [style=sLabel] (2) at (-0.25, 0.375) {$0$};
	\end{pgfonlayer}
	\begin{pgfonlayer}{edgelayer}
		\draw (1) to (0.center);
	\end{pgfonlayer}
\end{tikzpicture}
$ for the \enquote{empty} state.
Diagrams in $\bf{LO}$ have an interpretation as linear maps acting on tensor products of $\mathcal{F}(\mathbb{C}) \simeq l^2(\mathbb{N})$
--- the bosonic Fock space with a single degree of freedom, also called \emph{bosonic mode}.
Formally, the bosonic Fock space over a Hilbert space $\mathcal{H}$ is given by $\mathcal{F}(\mathcal{H}) = \bigoplus_n \mathcal{H}^{\tilde{\otimes} n}$
where $\tilde{\otimes}$ denotes the quotient of the tensor product by $x \otimes y \sim y \otimes x$.
For a finite number of modes $m$, we moreover have $\mathcal{F}(\mathbb{C}^m) \simeq \mathcal{F}(\mathbb{C})^{\otimes m}$~\cite{defeliceQuantumLinearOptics2023}.
Given an $\bf{LO}$ circuit on $m$ modes, the amplitudes of different input-output pairs can be computed by taking permanents of an underlying $m \times m$ unitary matrix.
We fix the interpretation of the phase shift of angle $\theta$ to be the matrix
$\begin{pmatrix}
   e^{i \theta}
\end{pmatrix}$,
and of the beam splitter to be the Hadamard matrix,
$\frac{1}{\sqrt{2}}\begin{pmatrix}
                     1 & 1  \\
                     1 & -1
\end{pmatrix}$.
The underlying matrix of arbitrary circuits is obtained by a simple block-diagonal matrix multiplication.
Throughout the paper, we use two different types of classical variables.
Underlined variables are \emph{outcomes} of a quantum measurement, they can take different values probabilistically.
Other variables are \emph{controlled}, they have a fixed value, set before the computation is executed.
For example, in the picture above, the number of prepared photons $\ccvar{n}$ is controlled while the number of detected
photons $\cvar{n}$ is a probabilistic outcome.

\subsection{Dual-rail qubits}\label{sec:bgd-dual-rail}

In photonic quantum computing, qubits are usually encoded by a photon in a pair of bosonic modes, a method known as \emph{dual-rail encoding}~\cite{knillSchemeEfficientQuantum2001}.
These could be two possible positions of the photon (spatial modes), or any other binary degree of freedom of the photon such as polarisation.
We use a \enquote{double wire} \begin{tikzpicture}
	\begin{pgfonlayer}{nodelayer}
		\node [style=none] (0) at (-0.75, 0) {};
		\node [style=none] (1) at (0.75, 0) {};
	\end{pgfonlayer}
	\begin{pgfonlayer}{edgelayer}
		\draw [style=DRWire] (0.center) to (1.center);
	\end{pgfonlayer}
\end{tikzpicture}
 to denote dual-rail modes.
These wires are interpreted as $\mathcal{F}(\mathbb{C}^2)$ --- the bosonic Fock space over a qubit ---
meaning that there can be any number of qubits in the same dual-rail mode.
Linear optical operations on these modes are defined from $\bf{LO}$ by using the following maps:
\[
  \input{dualrail/dual_rail_split.tikz}
  \qquad \qquad \qquad \qquad
  \input{dualrail/dual_rail_merge.tikz}
\]
These two maps are inverses of each other, corresponding to the natural isomorphism $\mathcal{F}(\mathbb{C}^2) \simeq \mathcal{F}(\mathbb{C}) \otimes \mathcal{F}(\mathbb{C})$.
We wish to represent processes acting on dual-rail qubits using the ZX calculus~\cite{coeckeInteractingQuantumObservables2008}.
However, ZX diagrams act on qubit spaces of the form $(\mathbb{C}^2)^{\otimes m}$ and there is no standard way of extending this action to $\mathcal{F}(\mathbb{C}^2)^{\otimes m}$.
There is, however, a natural isometry $\mathbb{C}^2 \to \mathcal{F}(\mathbb{C}^2)$, encoding a qubit state into its dual-rail representation.
We call it \enquote{triangle} and represent it as follows:
\begin{equation}
  \scalebox{1.5}{\begin{tikzpicture}
	\begin{pgfonlayer}{nodelayer}
		\node [style=ltriang] (0) at (0, 0) {};
		\node [style=none] (1) at (-1, 0) {};
		\node [style=none] (2) at (1, 0) {};
	\end{pgfonlayer}
	\begin{pgfonlayer}{edgelayer}
		\draw (1.center) to (0);
		\draw [style=DRWire] (2.center) to (0);
	\end{pgfonlayer}
\end{tikzpicture}
}
  \label{eq:triangle}
\end{equation}
Note that the adjoint of the triangle is a projector onto the qubit subspace, and we never use it in this paper.
We can now translate between dual-rail circuits and ZX diagrams using this graphical component.
For example, the qubit computational basis states are given by the dual-rail states:
\begin{equation}
  \label{eq:dual-rail-encoding}
  \input{dualrail/encode_ket_0.tikz}
  \qquad \qquad \qquad
  \input{dualrail/encode_ket_1.tikz}
\end{equation}
Going further, one may show that any single-qubit unitary can be realised on dual-rail qubits using only linear optical devices
by proving the following equations:
\[
  \input{dualrail/EncodeHadamard.tikz}
  \qquad \qquad
  \input{dualrail/EncodeZSpider11.tikz}
\]
Similarly, we may perform any single-qubit measurement on dual-rail qubits using photon detectors:
\[
  \input{dualrail/EncodeZAplhaMeasurement.tikz}
  \qquad \qquad
  \input{dualrail/EncodeXAplhaMeasurement.tikz}
\]
By pushing triangles from left to right in a dual-rail diagram, we compute the action of this diagram assuming
that a single photon is input in each dual-rail mode.
\begin{remark}
  Our results in this paper focus on dual-rail encoded qubits. 
  However, the only properties we use of the dual-rail encoding are the equations given in this subsection.
  In order to generalise our results to qubit encodings in $SU(n)$ (rather than $SU(2)$), one may use $n$ optical 
  modes and replace the beam splitter and phase shift with $n \times n$ unitaries satisfying the same equations.
\end{remark}

\subsection{Measurements and mixed channels}

The graphical language described so far has a standard interpretation in \emph{pure} quantum mechanics.
Any diagram $D$ has an associated linear map $\interp{D}$ that acts on $\mathcal{H} = (\mathbb{C}^2)^{\otimes n}$ when 
$D \in \bf{ZX}$ and on $\mathcal{H} = \mathcal{F}(\mathbb{C})^{\otimes n}$ when $D \in \bf{LO}$.
The same diagrams can be interpreted as \emph{mixed} quantum channels modeled by completely positive maps, using the CP construction~\cite{selingerIdempotentsDaggerCategories2006, coeckePicturesCompletePositivity2016}.
We interpret a diagram $D$ as the completely positive linear operator $\interp{D}_{CP} : \rho \mapsto \interp{D}\rho \interp{D}^\dagger$ 
acting on the space of bounded operators on $\mathcal{H}$, denoted $B(\mathcal{H})$.
We call $\interp{D}_{CP}$ the \emph{CP interpretation} of $D$.
The calculus can then be extended with a discarding effect, interpreted as the trace operator $\mathcal{B}(\mathcal{H}) \to I$.
The discarding maps for each space together with the relations between them are as follows:
\[
  \input{discards.tikz}
\]
We are usually interested in \emph{causal maps}, also called \emph{channels}.
These are defined as completely positive maps that preserve the trace operator, in the following sense:
\[
  \input{causal-map.tikz}
\]
When $D$ is a pure map the equation above corresponds to $\interp{D}^\dagger \interp{D} = I$, i.e. $\interp{D}$ is an isometry.
Completely positive maps are closed under addition and multiplication by real positive scalars, allowing us to interpret any positive linear combination of diagrams.
Causal maps are closed under taking \emph{probability distributions}: for a (discrete) probability distribution $p_i \in [0, 1]$ with $i \in X$, 
and causal maps $\set{f_i}_{i \in X}$, the map defined by $\sum_{i \in X} p_i f_i$ is also causal.

A destructive orthonormal basis measurement can be seen as a completely positive operator $\mathcal{B}(\mathcal{H}) \to \mathcal{H}$
returning a distribution over a basis of $\mathcal{H}$ for any input state.
In categorical quantum mechanics, the standard approach is to represent these maps using Frobenius algebras~\cite{coeckeCategoriesQuantumClassical2016,abramskyAlgebrasNonunitalFrobenius2012}.
Here, we use \emph{variables} to represent classical data, which is also standard in the literature~\cite{duncanGraphicalApproachMeasurementbased2013, backensThereBackAgain2021}.
We describe arbitrary measurements in terms of their Kraus decomposition given by
$\rho \mapsto \sum_{k \in O} \ket{k} \interp{D_k} \rho \interp{D_k}$ where $k$ ranges over the possible outcomes of the measurement.
The diagrams $D_k$ correspond to the different \emph{branches}, or Kraus maps, of the measurement.
In each of these branches, the classical variable $k$ has a fixed value.
For example, the $X$ and $Z$ single-qubit measurements can be written as:
\[ 
 \input{kraus-Zmeasurement.tikz} \quad \overset{\interp{\cdot}_{CP}}{\longmapsto} \quad \rho \mapsto \ket{0} \bra{0}\rho\ket{0} + \ket{1} \bra{1}\rho\ket{1}
\]
\[ 
 \input{kraus-Xmeasurement.tikz} \quad \overset{\interp{\cdot}_{CP}}{\longmapsto} \quad \rho \mapsto \ket{0} \bra{\plus}\rho\ket{\plus} + \ket{1} \bra{\minu}\rho\ket{\minu}
\]
Note that the sum above indicates mixing, rather than superposition.
As an example of a non-destructive measurement, 
consider a process $F$ which produces a classical bit $\cvar{k}$, making it act as one of two maps $F_A, F_B$.
In our language, we can represent this process as follows:
\[
  \input{outcome-var-box.tikz}
\]
In other words, $F = F_A$ if $\cvar{k} = 1$ and $F = F_B$ is $\cvar{k} = 0$.
The output $\cvar{k}$ can then be used later in the optical circuit.
If instead, we wish to represent a classically controlled process, which acts as $F_A$ if $\ccvar{x} = 1$ and as $F_B$ if $\ccvar{x} = 0$
(for some control parameter $\ccvar{x}$) then we depict the process as follows.
\[
  \input{control-var-box.tikz}
\]
To recover the standard notation with input and output classical wires it is sufficient to remember that underlined 
variables are classical outputs and other variables are classical inputs.

We can thus represent a quantum channel with classical output as a diagram $D$ labelled by an output variable $\cvar{k}$.
Then, the probability of an outcome $e$, given an input state $\rho$, is obtained by setting $\cvar{k} = e$ in $D$
and tracing out the remaining outputs:
\[ 
P_D(\cvar{k} = e \, \vert \, \rho) = \quad \input{probability.tikz}
\]
For example, the probability of observing the $\cvar{k} = 0$ outcome of an X measurement on input $\ket{0}\bra{0}$ is:
\[ 
  \input{probability-Xmeasurement.tikz}
\]
Note that the scalar $\star$ corresponds to $\frac{1}{2}$ in the CP interpretation.
This is an instance of the Born rule $\interp{\star}_{CP} = \norm{\interp{\star}}^2 = \frac{1}{2}$.
The causality condition for quantum channels given above ensures that normalised states are mapped to normalised states by the channel.

\subsection{Fusion measurements}\label{subsec:fusion-measurements}

Fusion measurements are at the heart of recent proposals for performing fault-tolerant quantum computation with photons~\cite{gimeno-segoviaThreePhotonGreenbergerHorneZeilingerStates2015, bartolucciFusionbasedQuantumComputation2023},
with different variations on this idea available in literature~\cite{limRepeatUntilSuccessLinearOptics2005, liHeraldedNondestructiveQuantum2021, pankovichFlexibleEntangledState2023,degliniastySpinOpticalQuantumComputing2024}.
Browne and Rudolf introduced two types of entangling measurements that they call \enquote{Type I} and \enquote{Type II fusion}~\cite{browneResourceEfficientLinearOptical2005}.
These fusions were expressed with quarter-wave plates and polarizing beam splitters, respectively, as follows:
\[
  \begin{tikzpicture}
	\begin{pgfonlayer}{nodelayer}
		\node [style=none] (0) at (-1.25, 0.75) {};
		\node [style=none] (1) at (1.25, 0.75) {};
		\node [style=none] (2) at (1.25, -0.75) {};
		\node [style=none] (3) at (-1.25, -0.75) {};
		\node [style=none] (4) at (-1.75, -0.75) {};
		\node [style=none] (7) at (-1.25, 0.75) {};
		\node [style=none] (8) at (-1.75, 0.75) {};
		\node [style=none] (11) at (2.5, -0.75) {};
		\node [style=none] (12) at (1.25, -0.75) {};
		\node [style=beamSplitter, rotate=90] (13) at (1.75, -0.75) {};
		\node [style=sLabel] (14) at (1.75, 0) {$45^{\circ}$};
		\node [style=none] (15) at (3.5, 0.75) {};
		\node [style=none] (16) at (1.25, 0.75) {};
		\node [style=none] (19) at (0, -0.375) {};
		\node [style=none] (20) at (-0.375, 0) {};
		\node [style=none] (21) at (0, 0.375) {};
		\node [style=none] (22) at (0.375, 0) {};
		\node [style=none] (23) at (2.5, -0.375) {};
		\node [style=none] (24) at (2.5, -1.125) {};
		\node [style=sLabel] (27) at (2.75, -0.75) {$\cvar a$};
	\end{pgfonlayer}
	\begin{pgfonlayer}{edgelayer}
		\draw [style=DRWire, in=0, out=180, looseness=0.75] (2.center) to (0.center);
		\draw [style=DRWire] (3.center) to (4.center);
		\draw [style=DRWire, in=0, out=-180, looseness=0.75] (1.center) to (3.center);
		\draw [style=DRWire] (7.center) to (8.center);
		\draw [style=DRWire] (11.center) to (12.center);
		\draw [style=DRWire] (15.center) to (16.center);
		\draw [fill=white, style=LOWire] (22.center)
			 to (20.center)
			 to (19.center)
			 to cycle
			 to (21.center)
			 to (20.center);
		\draw [style=LOWire] (23.center) to (24.center);
		\draw [style=LOWire, bend left=90, looseness=2.50] (23.center) to (24.center);
	\end{pgfonlayer}
\end{tikzpicture}

  \qquad \qquad \qquad \qquad
  \begin{tikzpicture}
	\begin{pgfonlayer}{nodelayer}
		\node [style=none] (1) at (-1.25, -0.75) {};
		\node [style=none] (2) at (1.25, -0.75) {};
		\node [style=none] (3) at (1.25, 0.75) {};
		\node [style=none] (12) at (-1.25, 0.75) {};
		\node [style=none] (13) at (-2.75, 0.75) {};
		\node [style=beamSplitter, rotate=90] (14) at (-2, 0.75) {};
		\node [style=sLabel] (15) at (-2, 1.5) {$45^{\circ}$};
		\node [style=none] (16) at (-1.25, -0.75) {};
		\node [style=none] (17) at (-2.75, -0.75) {};
		\node [style=beamSplitter, rotate=90] (18) at (-2, -0.75) {};
		\node [style=sLabel] (19) at (-2, -1.5) {$45^{\circ}$};
		\node [style=none] (20) at (2.75, 0.75) {};
		\node [style=none] (21) at (1.25, 0.75) {};
		\node [style=beamSplitter, rotate=90] (22) at (2, 0.75) {};
		\node [style=sLabel] (23) at (2, 1.5) {$45^{\circ}$};
		\node [style=none] (24) at (2.75, -0.75) {};
		\node [style=none] (25) at (1.25, -0.75) {};
		\node [style=beamSplitter, rotate=90] (26) at (2, -0.75) {};
		\node [style=sLabel] (27) at (2, -1.5) {$45^{\circ}$};
		\node [style=none] (8) at (0, 0.375) {};
		\node [style=none] (9) at (-0.375, 0) {};
		\node [style=none] (10) at (0, -0.375) {};
		\node [style=none] (11) at (0.375, 0) {};
		\node [style=none] (28) at (2.75, 0.375) {};
		\node [style=none] (29) at (2.75, 1.125) {};
		\node [style=none] (30) at (2.75, -1.125) {};
		\node [style=none] (31) at (2.75, -0.375) {};
		\node [style=sLabel] (32) at (3, 0.75) {$\cvar c$};
		\node [style=sLabel] (33) at (3, -0.75) {$\cvar a$};
	\end{pgfonlayer}
	\begin{pgfonlayer}{edgelayer}
		\draw [style=DRWire, in=0, out=-180, looseness=0.75] (3.center) to (1.center);
		\draw [style=DRWire] (12.center) to (13.center);
		\draw [style=DRWire, in=0, out=180, looseness=0.75] (2.center) to (12.center);
		\draw [style=DRWire] (16.center) to (17.center);
		\draw [style=DRWire] (20.center) to (21.center);
		\draw [style=DRWire] (24.center) to (25.center);
		\draw [fill=white, style=LOWire] (11.center)
			 to (9.center)
			 to (8.center)
			 to cycle
			 to (10.center)
			 to (9.center);
		\draw [style=LOWire] (28.center) to (29.center);
		\draw [style=LOWire, bend right=90, looseness=2.50] (28.center) to (29.center);
		\draw [style=LOWire] (30.center) to (31.center);
		\draw [style=LOWire, bend right=90, looseness=2.50] (30.center) to (31.center);
	\end{pgfonlayer}
\end{tikzpicture}

\]
Note that Type I fusion is a partial measurement having two inputs and one output dual-rail mode,
while Type II is destructive and measures both qubits.
We can translate from polarization primitives to $\bf{LO}$ circuits using the following equalities:
\[
  \begin{tikzpicture}
	\begin{pgfonlayer}{nodelayer}
		\node [style=none] (4) at (-1.25, 0.75) {};
		\node [style=none] (5) at (-1.25, -0.75) {};
		\node [style=none] (6) at (1.25, -0.75) {};
		\node [style=none] (7) at (1.25, 0.75) {};
		\node [style=none] (9) at (-1.75, 0.75) {};
		\node [style=none] (10) at (1.75, 0.75) {};
		\node [style=none] (11) at (1.75, -0.75) {};
		\node [style=none] (12) at (-1.75, -0.75) {};
		\node [style=none] (34) at (0, 0.375) {};
		\node [style=none] (35) at (-0.375, 0) {};
		\node [style=none] (36) at (0, -0.375) {};
		\node [style=none] (37) at (0.375, 0) {};
	\end{pgfonlayer}
	\begin{pgfonlayer}{edgelayer}
		\draw [style=DRWire] (10.center)
			 to (7.center)
			 to [in=0, out=-180, looseness=0.75] (5.center)
			 to (12.center);
		\draw [style=DRWire] (11.center)
			 to (6.center)
			 to [in=0, out=180, looseness=0.75] (4.center)
			 to (9.center);
		\draw [fill=white, style=LOWire] (37.center)
			 to (36.center)
			 to (35.center)
			 to cycle
			 to (34.center)
			 to (35.center);
		\draw [fill=white, style=LOWire] (37.center) to (35.center);
		\draw [fill=white, style=LOWire] (35.center) to (34.center);
		\draw [fill=white, style=LOWire] (34.center) to (37.center);
		\draw [fill=white, style=LOWire] (37.center) to (36.center);
		\draw [fill=white, style=LOWire] (36.center) to (35.center);
	\end{pgfonlayer}
\end{tikzpicture}

  \quad=\quad
  \begin{tikzpicture}
	\begin{pgfonlayer}{nodelayer}
		\node [style=none] (4) at (-1.25, 0.375) {};
		\node [style=none] (5) at (-1.25, -1.125) {};
		\node [style=none] (6) at (1.25, -1.125) {};
		\node [style=none] (7) at (1.25, 0.375) {};
		\node [style=none] (12) at (-2.75, 0.775) {};
		\node [style=none] (13) at (-2.75, 0.725) {};
		\node [style=none] (14) at (-2.75, 0.75) {};
		\node [style=none] (15) at (-1.5, 1.125) {};
		\node [style=none] (16) at (-1.5, 0.375) {};
		\node [style=none] (17) at (-2.75, -0.725) {};
		\node [style=none] (18) at (-2.75, -0.775) {};
		\node [style=none] (19) at (-2.75, -0.75) {};
		\node [style=none] (20) at (-1.5, -0.375) {};
		\node [style=none] (21) at (-1.5, -1.125) {};
		\node [style=none] (22) at (2.75, 0.775) {};
		\node [style=none] (23) at (2.75, 0.725) {};
		\node [style=none] (24) at (2.75, 0.75) {};
		\node [style=none] (25) at (1.5, 1.125) {};
		\node [style=none] (26) at (1.5, 0.375) {};
		\node [style=none] (27) at (2.75, -0.725) {};
		\node [style=none] (28) at (2.75, -0.775) {};
		\node [style=none] (29) at (2.75, -0.75) {};
		\node [style=none] (30) at (1.5, -0.375) {};
		\node [style=none] (31) at (1.5, -1.125) {};
		\node [style=none] (32) at (-3.5, 0.75) {};
		\node [style=none] (33) at (-3.5, -0.75) {};
		\node [style=none] (34) at (3.5, -0.75) {};
		\node [style=none] (35) at (3.5, 0.75) {};
	\end{pgfonlayer}
	\begin{pgfonlayer}{edgelayer}
		\draw [style=LOWire, in=0, out=-180, looseness=0.75] (7.center) to (5.center);
		\draw [style=LOWire, in=0, out=180, looseness=0.75] (6.center) to (4.center);
		\draw [style=LOWire, in=0, out=180, looseness=0.75] (15.center) to (12.center);
		\draw [style=LOWire, in=180, out=0, looseness=0.75] (13.center) to (16.center);
		\draw [style=LOWire] (4.center) to (16.center);
		\draw [style=LOWire, in=0, out=180, looseness=0.75] (20.center) to (17.center);
		\draw [style=LOWire, in=180, out=0, looseness=0.75] (18.center) to (21.center);
		\draw [style=LOWire, in=180, out=0, looseness=0.75] (25.center) to (22.center);
		\draw [style=LOWire, in=0, out=180, looseness=0.75] (23.center) to (26.center);
		\draw [style=LOWire, in=180, out=0, looseness=0.75] (30.center) to (27.center);
		\draw [style=LOWire, in=0, out=180, looseness=0.75] (28.center) to (31.center);
		\draw [style=LOWire] (20.center) to (30.center);
		\draw [style=LOWire] (26.center) to (7.center);
		\draw [style=LOWire] (5.center) to (21.center);
		\draw [style=LOWire] (31.center) to (6.center);
		\draw [style=LOWire] (15.center) to (25.center);
		\draw [style=DRWire] (35.center) to (24.center);
		\draw [style=DRWire] (34.center) to (29.center);
		\draw [style=DRWire] (32.center) to (14.center);
		\draw [style=DRWire] (33.center) to (19.center);
	\end{pgfonlayer}
\end{tikzpicture}

  \qquad \qquad \qquad
  \begin{tikzpicture}
	\begin{pgfonlayer}{nodelayer}
		\node [style=none] (0) at (1.25, 0) {};
		\node [style=none] (4) at (-1.25, 0) {};
		\node [style=beamSplitter, rotate=90] (25) at (0, 0) {};
		\node [style=sLabel] (35) at (0, 0.75) {$45^{\circ}$};
	\end{pgfonlayer}
	\begin{pgfonlayer}{edgelayer}
		\draw [style=DRWire] (0.center) to (4.center);
	\end{pgfonlayer}
\end{tikzpicture}

  \quad=\quad
  \begin{tikzpicture}
	\begin{pgfonlayer}{nodelayer}
		\node [style=none] (0) at (-1, 0.5) {};
		\node [style=none] (1) at (-1, -0.5) {};
		\node [style=none] (2) at (1, -0.5) {};
		\node [style=none] (3) at (1, 0.5) {};
		\node [style=beamSplitter] (4) at (0, 0) {};
		\node [style=none] (9) at (-2.5, 0.025) {};
		\node [style=none] (10) at (-2.5, -0.025) {};
		\node [style=none] (11) at (-2.5, 0) {};
		\node [style=none] (12) at (-1.25, 0.5) {};
		\node [style=none] (13) at (-1.25, -0.5) {};
		\node [style=none] (14) at (2.25, 0.025) {};
		\node [style=none] (15) at (2.25, -0.025) {};
		\node [style=none] (16) at (2.25, 0) {};
		\node [style=none] (17) at (1.25, 0.5) {};
		\node [style=none] (18) at (1.25, -0.5) {};
		\node [style=none] (19) at (-3.25, 0) {};
		\node [style=none] (20) at (3, 0) {};
	\end{pgfonlayer}
	\begin{pgfonlayer}{edgelayer}
		\draw [style=LOWire, in=0, out=-180, looseness=0.75] (3.center) to (1.center);
		\draw [style=LOWire, in=180, out=0, looseness=0.75] (0.center) to (2.center);
		\draw [style=LOWire, in=0, out=180, looseness=0.75] (12.center) to (9.center);
		\draw [style=LOWire, in=180, out=0, looseness=0.75] (10.center) to (13.center);
		\draw [style=LOWire, in=180, out=0, looseness=0.75] (17.center) to (14.center);
		\draw [style=LOWire, in=0, out=180, looseness=0.75] (15.center) to (18.center);
		\draw [style=DRWire] (20.center) to (16.center);
		\draw [style=DRWire] (19.center) to (11.center);
		\draw [style=LOWire] (0.center) to (12.center);
		\draw [style=LOWire] (1.center) to (13.center);
		\draw [style=LOWire] (17.center) to (3.center);
		\draw [style=LOWire] (18.center) to (2.center);
	\end{pgfonlayer}
\end{tikzpicture}

\]
Expressing the Type I and Type II fusions using this translation, we obtain the following diagrams, respectively:
\[
  \begin{tikzpicture}
	\begin{pgfonlayer}{nodelayer}
		\node [style=none] (14) at (0.5, 0.375) {};
		\node [style=none] (15) at (0.5, -0.375) {};
		\node [style=beamSplitter] (16) at (-0.5, 0) {};
		\node [style=black node] (20) at (1, -0.375) {};
		\node [style=black node] (21) at (1, 0.375) {};
		\node [style=sLabel] (22) at (1.375, 0.375) {$\cvar a$};
		\node [style=sLabel] (23) at (1.375, -0.375) {$\cvar b$};
		\node [style=none] (34) at (-3.5, 0.775) {};
		\node [style=none] (35) at (-3.5, 0.725) {};
		\node [style=none] (36) at (-3.5, 0.75) {};
		\node [style=none] (37) at (-2.25, 1.125) {};
		\node [style=none] (38) at (-1.5, -0.375) {};
		\node [style=none] (39) at (3, 0.025) {};
		\node [style=none] (40) at (3, -0.025) {};
		\node [style=none] (41) at (3, 0) {};
		\node [style=none] (42) at (1.25, 1.125) {};
		\node [style=none] (43) at (1.25, -1.125) {};
		\node [style=none] (44) at (-4.25, 0.75) {};
		\node [style=none] (45) at (3.75, 0) {};
		\node [style=none] (46) at (-3.5, -0.725) {};
		\node [style=none] (47) at (-3.5, -0.775) {};
		\node [style=none] (48) at (-3.5, -0.75) {};
		\node [style=none] (49) at (-1.5, 0.375) {};
		\node [style=none] (50) at (-2.25, -1.125) {};
		\node [style=none] (51) at (-4.25, -0.75) {};
	\end{pgfonlayer}
	\begin{pgfonlayer}{edgelayer}
		\draw [style=LOWire] (15.center) to (20);
		\draw [style=LOWire] (14.center) to (21);
		\draw [style=LOWire, in=0, out=180, looseness=0.75] (37.center) to (34.center);
		\draw [style=LOWire, in=180, out=0, looseness=0.75] (35.center) to (38.center);
		\draw [style=LOWire, in=180, out=0, looseness=0.75] (42.center) to (39.center);
		\draw [style=LOWire, in=0, out=180, looseness=0.75] (40.center) to (43.center);
		\draw [style=DRWire] (45.center) to (41.center);
		\draw [style=DRWire] (44.center) to (36.center);
		\draw [style=LOWire, in=0, out=180, looseness=0.75] (49.center) to (46.center);
		\draw [style=LOWire, in=180, out=0, looseness=0.75] (47.center) to (50.center);
		\draw [style=DRWire] (51.center) to (48.center);
		\draw [style=LOWire, in=0, out=180, looseness=0.75] (15.center) to (49.center);
		\draw [style=LOWire, in=-180, out=0, looseness=0.75] (38.center) to (14.center);
		\draw [style=LOWire] (42.center) to (37.center);
		\draw [style=LOWire] (43.center) to (50.center);
	\end{pgfonlayer}
\end{tikzpicture}

  \qquad \qquad \qquad \qquad
  \begin{tikzpicture}
	\begin{pgfonlayer}{nodelayer}
		\node [style=none] (58) at (-0.5, -0.375) {};
		\node [style=none] (59) at (-0.5, 0.375) {};
		\node [style=none] (60) at (1.5, 0.375) {};
		\node [style=none] (61) at (1.5, -0.375) {};
		\node [style=beamSplitter] (62) at (0.5, 0) {};
		\node [style=black node] (63) at (1.75, -0.375) {};
		\node [style=black node] (64) at (1.75, 0.375) {};
		\node [style=sLabel] (65) at (2.125, 0.375) {$\cvar a$};
		\node [style=sLabel] (66) at (2.125, -0.375) {$\cvar b$};
		\node [style=none] (67) at (2, 1.125) {};
		\node [style=none] (68) at (2, -1.125) {};
		\node [style=none] (71) at (-3.25, 0.375) {};
		\node [style=none] (72) at (-3.25, 1.125) {};
		\node [style=none] (73) at (-1.25, 1.125) {};
		\node [style=none] (74) at (-2.25, 0.75) {};
		\node [style=beamSplitter] (75) at (-2.25, 0.75) {};
		\node [style=none] (76) at (-3.25, -1.125) {};
		\node [style=none] (77) at (-3.25, -0.375) {};
		\node [style=none] (78) at (-2.25, -0.75) {};
		\node [style=none] (79) at (-1.25, -1.125) {};
		\node [style=beamSplitter] (80) at (-2.25, -0.75) {};
		\node [style=beamSplitter] (83) at (3.5, 0) {};
		\node [style=none] (84) at (4.25, 0.375) {};
		\node [style=none] (85) at (4.25, -0.375) {};
		\node [style=black node] (86) at (4.75, -0.375) {};
		\node [style=black node] (87) at (4.75, 0.375) {};
		\node [style=sLabel] (88) at (5.125, -0.375) {$\cvar d$};
		\node [style=sLabel] (89) at (5.125, 0.375) {$\cvar c$};
		\node [style=none] (90) at (-4.875, 0.775) {};
		\node [style=none] (91) at (-4.875, 0.725) {};
		\node [style=none] (92) at (-4.875, 0.75) {};
		\node [style=none] (93) at (-3.625, 1.125) {};
		\node [style=none] (94) at (-3.625, 0.375) {};
		\node [style=none] (95) at (-5.625, 0.75) {};
		\node [style=none] (96) at (-4.875, -0.725) {};
		\node [style=none] (97) at (-4.875, -0.775) {};
		\node [style=none] (98) at (-4.875, -0.75) {};
		\node [style=none] (99) at (-3.625, -0.375) {};
		\node [style=none] (100) at (-3.625, -1.125) {};
		\node [style=none] (101) at (-5.625, -0.75) {};
	\end{pgfonlayer}
	\begin{pgfonlayer}{edgelayer}
		\draw [style=LOWire, in=0, out=180, looseness=0.75] (61.center) to (59.center);
		\draw [style=LOWire, in=-180, out=0, looseness=0.75] (58.center) to (60.center);
		\draw [style=LOWire] (61.center) to (63);
		\draw [style=LOWire] (60.center) to (64);
		\draw [style=LOWire, in=0, out=150, looseness=0.75] (74.center) to (72.center);
		\draw [style=LOWire, in=-180, out=0, looseness=0.75] (71.center) to (73.center);
		\draw [style=LOWire, in=0, out=180, looseness=0.75] (79.center) to (77.center);
		\draw [style=LOWire, in=-150, out=0, looseness=0.75] (76.center) to (78.center);
		\draw [style=LOWire] (67.center) to (73.center);
		\draw [style=LOWire] (68.center) to (79.center);
		\draw [style=LOWire] (85.center) to (86);
		\draw [style=LOWire] (84.center) to (87);
		\draw [style=LOWire, in=0, out=180, looseness=0.75] (93.center) to (90.center);
		\draw [style=LOWire, in=180, out=0, looseness=0.75] (91.center) to (94.center);
		\draw [style=DRWire] (95.center) to (92.center);
		\draw [style=LOWire, in=0, out=180, looseness=0.75] (99.center) to (96.center);
		\draw [style=LOWire, in=180, out=0, looseness=0.75] (97.center) to (100.center);
		\draw [style=DRWire] (101.center) to (98.center);
		\draw [style=LOWire] (72.center) to (93.center);
		\draw [style=LOWire] (71.center) to (94.center);
		\draw [style=LOWire] (77.center) to (99.center);
		\draw [style=LOWire] (76.center) to (100.center);
		\draw [style=LOWire, in=-30, out=180, looseness=0.75] (58.center) to (74.center);
		\draw [style=LOWire, in=180, out=30, looseness=0.75] (78.center) to (59.center);
		\draw [style=LOWire, in=0, out=180, looseness=0.75] (85.center) to (67.center);
		\draw [style=LOWire, in=-180, out=0, looseness=0.75] (68.center) to (84.center);
	\end{pgfonlayer}
\end{tikzpicture}

\]
This representation enables us to diagrammatically calculate the action of these circuits by representing them as a mixture of ZX diagrams.
Starting with Type I fusion, we get the following Kraus decomposition, proved in \cref{sec:mixedfusion}.
\begin{restatable}{proposition}{typeOneFusionProp} The following equation holds in the CP interpretation:
  
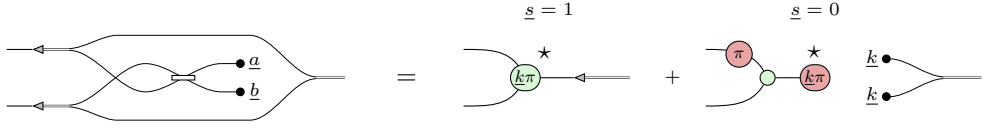
\label{prop:type-I}
  \begin{gather*}
    \input{dualrail/Type1FusionFull.tikz}
  \end{gather*}
  after coarse-graining of the measurement operator by the equations $\cvar s = \cvar a \oplus \cvar b$ and $\cvar k = \cvar s \cvar b + \neg \cvar s (1 - \frac{\cvar a + \cvar b}{2})$.
  Here, $s$ is the Boolean value of success and $k$ is the Pauli measurement error.
\end{restatable}
This means that the error is $\cvar b$ in case of success and $1 - \frac{\cvar a + \cvar b}{2}$ in the failure case.
Note that, in case of failure, the pair of output modes is no longer in the qubit subspace defined in \cref{eq:dual-rail-encoding}.
Now, considering the Type II fusion, we see that this is just a Type I fusion
preceded by beam splitters and followed by a single qubit measurement in the Z-basis.
We can thus compute its action on the qubit subspace:
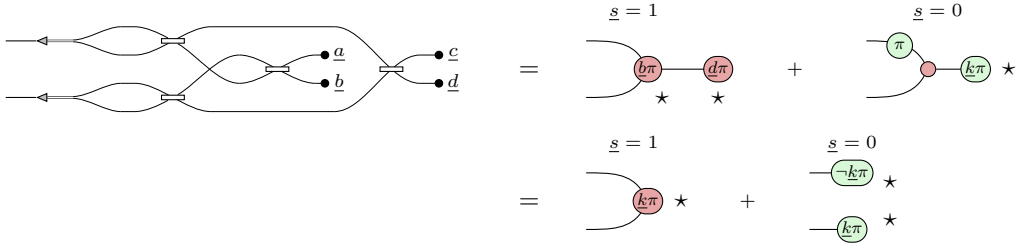
\[
  \input{dualrail/Type2FusionFull.tikz}
\]
where $\cvar s = \cvar a \oplus \cvar b$ is the Boolean value of success and $\cvar k = \cvar s (\cvar b + \cvar d) + \neg \cvar s (1 - \frac{\cvar a + \cvar b}{2})$ is the error.
This means that the error is $\cvar b + \cvar d$ in case of success and $1 - \frac{\cvar a + \cvar b}{2}$ in case of failure.
The scalars in the diagrams above are crucial for computing probabilities, 
but they can be disregarded in many cases. 
We use them only in \cref{sec:universality}, to compute the probability of success of fusion measurements for different input states.
We see that, in their success outcomes, Type I and Type II fusion correspond to Z and X spiders in the ZX calculus.
However, nothing prevents us from defining other types of fusion measurements.
In \cref{sec:character}, we give a complete characterisation of the fusion measurements that induce correctable Pauli errors
in both the success and failure branches.

\begin{remark}
  \label{remark:nType1}
  Both Type I and Type II fusions can be generalized to arbitrary number of inputs.
  For example, the Type I fusion with $3$ inputs is given by the following $\bf{LO}$ circuit:
  \[
    \input{dualrail/FusionOfType1Fusions.tikz}
  \]
  Similarly, one may construct the linear optical diagram for a Type I fusion with $n$ inputs.
  Measuring the output modes of a $n$-input Type I fusion, we obtain the $n$-GHZ state analyzer studied in~\cite{panGreenbergerHorneZeilingerstateAnalyzer1998, boseMultiparticleGeneralizationEntanglement1998, pankovichFlexibleEntangledState2023}.
  Assuming the input is in dual-rail encoding, we can then compute the Kraus decomposition of this measurement.
  For the $3$-GHZ analyzer we obtain:
  \[
    \input{dualrail/Type2Fusion-ternary.tikz}
  \]
  where $\cvar{j} = r_2 + r_4 + r_6 \mod 2$, $\cvar{s} = r_1 + r_2 \mod 2$, $\cvar{s'} = r_3 + r_4 \mod 2$, 
  $\cvar{k} = 1 - \frac{r_{3} + r_{4}}{2}$ if $\cvar{s} = 1$ and $\cvar{k} = 1 - \frac{r_1 + r_2}{2}$ otherwise. 
  Similarly, one may compute the Kraus decompositions of $n$-input fusions altough some work is required to handle all the failure branches.
\end{remark}

\subsection{Resource state generation}\label{subsubsec:resource-states}

In fusion-based architectures,
cluster states are constructed by gluing smaller \enquote{resource states}, that are provided at every time step.
Here, we discuss different methods for generating photonic graph states.
These can be broadly assigned to two classes --- (i) linear optical and (ii) matter-based methods ---
although the two procedures can in principle be used in conjunction.

Linear optical methods typically begin with single photons, which are usually generated by spontaneous parametric down-conversion~\cite{harrisObservationTunableOptical1967}.
These photons are then entangled using linear optical Bell measurements or other heralded linear optical circuits~\cite{bartolucciCreationEntangledPhotonic2021}.
%Below is an example of how fusion measurements can be used to create bell states.
The advantage of this approach is that the constructed resource states can in principle have arbitrary connectivity as photons
are not spatially or temporally restricted.
Moreover, photons of different resource states can be prepared with low distinguishability by active alignment of sources~\cite{carolanScalableFeedbackControl2019}.
Examples of photonic graph states used in the literature include the star graph~\cite{gimeno-segoviaThreePhotonGreenbergerHorneZeilingerStates2015, leeGraphtheoreticalOptimizationFusionbased2023},
rings~\cite{bartolucciFusionbasedQuantumComputation2023} and complete-like graphs~\cite{azumaAllphotonicQuantumRepeaters2015}.
The main drawback of these approaches is that, because of the fundamental limits of linear optics~\cite{stanisicGeneratingEntanglementLinear2017},
entanglement can only be generated probabilistically.
%This introduces an additional level of non-determinism to the overall computation, which can be restrictive.
This drawback can be mitigated by using ancillary photons~\cite{bartolucciCreationEntangledPhotonic2021}
and \enquote{switch networks}~\cite{bartolucciSwitchNetworksPhotonic2021} to boost probabilities of success~\cite{ewertEfficientBellMeasurement2014}.

Matter-based approaches rely on the generation of photons by excitation of a trapped ion~\cite{blinovObservationEntanglementSingle2004} or an artificial atom~\cite{ekimovQuantumSizeEffect1981, rossettiQuantumSizeEffects1983, murraySynthesisCharacterizationNearly1993, kastnerArtificialAtoms1993}.
The emitted photons are entangled with the state of the atom, and thus, if the atom is kept in a coherent superposition,
the emitted photons will also be entangled to each other~\cite{thomasEfficientGenerationEntangled2022}.
This technology has proved particularly effective for the generation of entangled photonic graph states~\cite{schwartzDeterministicGenerationCluster2016, costeHighrateEntanglementSemiconductor2023, coganDeterministicGenerationIndistinguishable2023}, and it has the advantage that resource states can be generated deterministically~\cite{schwartzDeterministicGenerationCluster2016, coganDeterministicGenerationIndistinguishable2023}.
Nevertheless, as photons are emitted one at a time, the resulting entanglement is restricted to \emph{linear} structure,
and photons need to be demultiplexed making them more susceptible to loss.
Photons emitted by non-identical atoms also suffer from distinguishability, although methods for mitigating this are being developed~\cite{yardOnchipQuantumInformation2022}.
A great advantage of matter-based emitters is that they can be used in repeat-until-success protocols, as shown in \cref{subsec:RUS},
that enables the near-deterministic implementation of entangling gates by fusion measurements~\cite{limRepeatUntilSuccessLinearOptics2005,degliniastySpinOpticalQuantumComputing2024}.

\begin{figure}[ht]
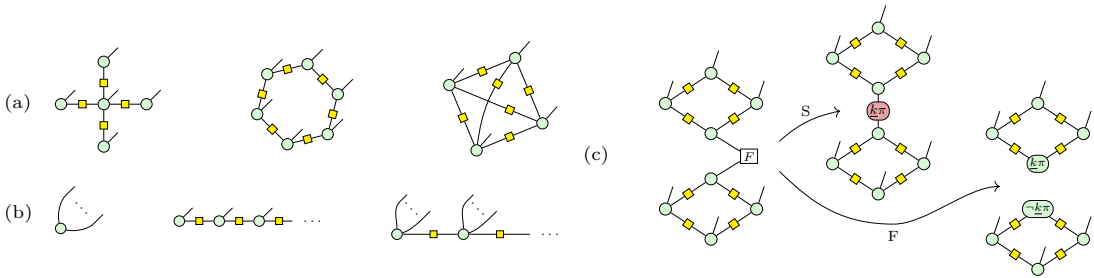

  \label{fig:fbqc}
  \centering
  \begin{subfigure}[c]{.5\textwidth}
    \scriptsize{(a)}
    \scalebox{0.75}{\input{ResourceStates/5Star.tikz}}
    \hfill
    \scalebox{0.75}{\input{ResourceStates/6Cycle.tikz}}
    \hfill
    \scalebox{0.75}{\input{ResourceStates/4Complete.tikz}}
    \\[\jot]
    \scriptsize{(b)}
    \scalebox{0.7}{\input{ResourceStates/GHZ.tikz}}
    \hfill
    \scalebox{0.7}{\input{ResourceStates/linear.tikz}}
    \hfill
    \scalebox{0.7}{\input{ResourceStates/VariableLinearGHZ.tikz}}
  \end{subfigure}
  \begin{subfigure}[c]{.49\textwidth}
    \scriptsize{(c)}
    \hfill
    \scalebox{0.8}{\input{fusion-example.tikz}}
    \hfill\null
%  \scalebox{0.8}{\tikzfig{ResourceStates/BellStatePreparationAlt}}
  \end{subfigure}
  \caption{
    Examples of resource graphs generated using
    (a) linear-optical
    and (b) matter-based methods.
    (c) Action of a Type-II fusion measurement on a graph state depending on measurement outcomes.
  }
  \label{fig:example-resuorce-states}
\end{figure}

\subsection{Measurement-based quantum computing}\label{subsec:MBQC}

This section reviews the basics of Measurement-Based Quantum Computing~\cite{raussendorfOneWayQuantumComputer2001} and the related \emph{flow conditions}~\cite{danosExtendedMeasurementCalculus2009} that ensure the determinism of the model.
We suggest~\cite{kissingerUniversalMBQCGeneralised2019, backensThereBackAgain2021} for a more thorough introduction,
and~\cite{simmonsRelatingMeasurementPatterns2021, mcelvanneyCompleteFlowpreservingRewrite2023, mcelvanneyFlowpreservingZXcalculusRewrite2023} for the latest developments in the field.
The literature on MBQC defines two distinct languages to specify measurement-based computations: open graphs and measurement patterns.
We introduce each of these notions and recall the determinism theorems~\cite{browneGeneralizedFlowDeterminism2007,danosExtendedMeasurementCalculus2009}
that are used to relate them.

In MBQC, computation is performed in two stages:
(i) a \emph{graph state} is prepared and
(ii) it is processed by a sequence of \emph{single-qubit measurements}.
A graph state associated with the graph $G = (V, E)$ is an entangled quantum state constructed by preparing a qubit for each vertex in the $\ket{+}$ state and applying $CZ$ gates for each edge.
We may depict graph states equivalently as ZX diagrams or qubit circuits.
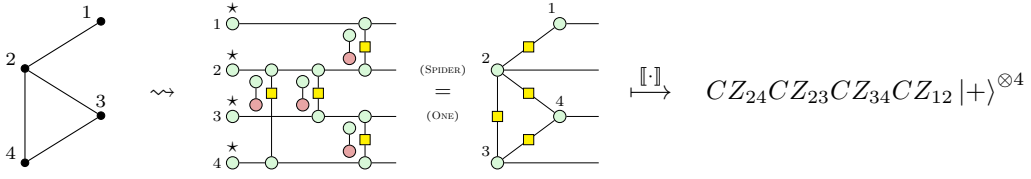
\begin{example}
  \label{example:graph-sate}
  \[
    \begin{tikzpicture}
	\begin{pgfonlayer}{nodelayer}
		\node [style=black node] (0) at (-1, 0.625) {};
		\node [style=black node] (1) at (-1, -1.875) {};
		\node [style=black node] (2) at (1, -0.625) {};
		\node [style=black node] (3) at (1, 1.875) {};
		\node [style=sLabel] (4) at (0.625, 2.125) {$1$};
		\node [style=sLabel] (5) at (-1.375, 0.875) {$2$};
		\node [style=sLabel] (6) at (-1.375, -1.625) {$4$};
		\node [style=sLabel] (7) at (1, -0.25) {$3$};
	\end{pgfonlayer}
	\begin{pgfonlayer}{edgelayer}
		\draw (2) to (1);
		\draw (1) to (0);
		\draw (0) to (2);
		\draw (3) to (0);
	\end{pgfonlayer}
\end{tikzpicture}

    \quad\leadsto
    \scalebox{.82}{\begin{tikzpicture}
	\begin{pgfonlayer}{nodelayer}
		\node [style=gn] (0) at (-6.75, 0.75) {};
		\node [style=gn] (1) at (-6.75, -0.75) {};
		\node [style=gn] (2) at (-6.75, -2.25) {};
		\node [style=gn] (3) at (-4, 0.75) {};
		\node [style=gn] (4) at (-4, -0.75) {};
		\node [style=gn] (5) at (-5.5, 0.75) {};
		\node [style=gn] (6) at (-5.5, -2.25) {};
		\node [style=gn] (7) at (-2.5, -0.75) {};
		\node [style=gn] (8) at (-2.5, -2.25) {};
		\node [style=none] (9) at (-1.5, 0.75) {};
		\node [style=none] (10) at (-1.5, -0.75) {};
		\node [style=none] (11) at (-1.5, -2.25) {};
		\node [style=hbox] (12) at (-4, 0) {};
		\node [style=hbox] (13) at (-5.5, 0) {};
		\node [style=hbox] (14) at (-2.5, -1.5) {};
		\node [style=none] (15) at (0, 0) {$=$};
		\node [style=gn] (16) at (1.75, 0.75) {};
		\node [style=gn] (17) at (1.75, -2.25) {};
		\node [style=gn] (18) at (3.75, -0.75) {};
		\node [style=hbox] (19) at (2.75, 0) {};
		\node [style=hbox] (20) at (1.75, -0.75) {};
		\node [style=hbox] (21) at (2.75, -1.5) {};
		\node [style=none] (22) at (5, 0.75) {};
		\node [style=none] (23) at (5, -0.75) {};
		\node [style=none] (24) at (5, -2.25) {};
		\node [style=gn] (25) at (3.75, 2.25) {};
		\node [style=hbox] (26) at (2.75, 1.5) {};
		\node [style=none] (27) at (5, 2.25) {};
		\node [style=gn] (28) at (-6.75, 2.25) {};
		\node [style=gn] (29) at (-2.5, 0.75) {};
		\node [style=gn] (30) at (-2.5, 2.25) {};
		\node [style=none] (31) at (-1.5, 2.25) {};
		\node [style=hbox] (32) at (-2.5, 1.5) {};
		\node [style=sLabel] (33) at (-7.25, 2.25) {$1$};
		\node [style=sLabel] (34) at (-7.25, 0.75) {$2$};
		\node [style=sLabel] (35) at (-7.25, -2.25) {$4$};
		\node [style=sLabel] (36) at (-7.25, -0.75) {$3$};
		\node [style=sLabel] (37) at (3.5, 2.75) {$1$};
		\node [style=sLabel] (38) at (1.375, 1.125) {$2$};
		\node [style=sLabel] (39) at (3.75, -0.25) {$4$};
		\node [style=sLabel] (40) at (1.375, -1.875) {$3$};
		\node [style=sLabel] (41) at (0, 0.75) {\Fusion};
		\node [style=none] (42) at (-6.75, 2.75) {$\star$};
		\node [style=none] (43) at (-6.75, 1.25) {$\star$};
		\node [style=none] (44) at (-6.75, -0.25) {$\star$};
		\node [style=none] (45) at (-6.75, -1.75) {$\star$};
		\node [style=gn] (46) at (-3, 1.875) {};
		\node [style=rn] (47) at (-3, 1.125) {};
		\node [style=gn] (48) at (-4.5, 0.375) {};
		\node [style=rn] (49) at (-4.5, -0.375) {};
		\node [style=gn] (50) at (-3, -1.125) {};
		\node [style=rn] (51) at (-3, -1.875) {};
		\node [style=gn] (52) at (-6, 0.375) {};
		\node [style=rn] (53) at (-6, -0.375) {};
		\node [style=sLabel] (54) at (0, -0.75) {\One};
	\end{pgfonlayer}
	\begin{pgfonlayer}{edgelayer}
		\draw (3) to (5);
		\draw (1) to (4);
		\draw (4) to (7);
		\draw (7) to (10.center);
		\draw (2) to (6);
		\draw (6) to (8);
		\draw (8) to (11.center);
		\draw (5) to (13);
		\draw (13) to (6);
		\draw (7) to (14);
		\draw (14) to (8);
		\draw (3) to (12);
		\draw (12) to (4);
		\draw (18) to (17);
		\draw (17) to (16);
		\draw (16) to (18);
		\draw (17) to (24.center);
		\draw (18) to (23.center);
		\draw (16) to (22.center);
		\draw (25) to (16);
		\draw (27.center) to (25);
		\draw (30) to (31.center);
		\draw (30) to (28);
		\draw (30) to (32);
		\draw (32) to (29);
		\draw (29) to (9.center);
		\draw (5) to (0);
		\draw (3) to (29);
		\draw (47) to (46);
		\draw (49) to (48);
		\draw (51) to (50);
		\draw (53) to (52);
	\end{pgfonlayer}
\end{tikzpicture}
}
    \overset{\interp{\cdot}}{\longmapsto} \quad
    CZ_{24} CZ_{23} CZ_{34} CZ_{12} \ket{+}^{\otimes 4}
  \]
\end{example}
\noindent Qubits can be inputs or outputs, which we depict by connecting wires to the left or right boundaries of the diagram.
\begin{definition}[Open graph]
  An open graph is a tuple $(G, I, O)$, where $G = (V, E)$ is an undirected graph, and $I, O \subseteq V$ are (possibly overlapping) subsets representing inputs and outputs.
  We use the notations $\comp{O} \coloneqq V \backslash O$ for the non-output and $\comp{I} \coloneqq V \backslash I$ for the non-input vertices.
\end{definition}
\noindent During computation, every non-output vertex of the graph is measured in a certain basis specified by a measurement plane ($\lambda$) and angle ($\alpha$).
\begin{definition}[Labelled open graph]
  A labelled open graph is a tuple $\mathcal{M} = (G, I, O, \lambda, \alpha)$,
  where $(G, I, O)$ is an open graph,
  $\lambda: \comp{O} \to \set{\XYm, \XZm, \YZm}$ is an assignment of measurement planes,
  and $\alpha: \comp{O} \to [0, 2\pi)$ assigns measurement angles to each non-output qubit.
\end{definition}
\noindent We use the following notation to denote an arbitrary pure single qubit state (its corresponding effect is defined analogously).
\[
\ket{\pm_{\lambda,\alpha}} = \begin{dcases}
\frac{1}{\sqrt{2}}(\ket{0} \pm e^{i \alpha}\ket{1}) & \text{if } \lambda = \XYm \hspace{10mm} \begin{tikzpicture}
	\begin{pgfonlayer}{nodelayer}
		\node [style={gn_phase}] (0) at (1, 0) {$\alpha$};
		\node [style=none] (1) at (0, 0) {};
		\node [style=none] (2) at (1.75, 0) {$\star$};
	\end{pgfonlayer}
	\begin{pgfonlayer}{edgelayer}
		\draw (1.center) to (0);
	\end{pgfonlayer}
\end{tikzpicture}
 \\
\frac{1}{\sqrt{2}}(\ket{+} \pm e^{i \alpha}\ket{-}) & \text{if } \lambda = \YZm \hspace{10mm} \begin{tikzpicture}
	\begin{pgfonlayer}{nodelayer}
		\node [style={rn_phase}] (0) at (1, 0) {$\alpha$};
		\node [style=none] (1) at (0, 0) {};
		\node [style=none] (2) at (1.75, 0) {$\star$};
	\end{pgfonlayer}
	\begin{pgfonlayer}{edgelayer}
		\draw (1.center) to (0);
	\end{pgfonlayer}
\end{tikzpicture}
 \\
\frac{1}{\sqrt{2}}(\ket{i} \pm e^{i \alpha}\ket{-i}) & \text{if } \lambda = \XZm  \hspace{10mm} \begin{tikzpicture}
	\begin{pgfonlayer}{nodelayer}
		\node [style={gn_phase}] (0) at (1, 0) {$\frac{\pi}{2}$};
		\node [style=none] (1) at (0, 0) {};
		\node [style={rn_phase}] (2) at (2, 0) {$\alpha$};
		\node [style=none] (3) at (2.75, 0) {$\star$};
	\end{pgfonlayer}
	\begin{pgfonlayer}{edgelayer}
		\draw (1.center) to (0);
		\draw (2) to (0);
	\end{pgfonlayer}
\end{tikzpicture}

\end{dcases}
\]
\noindent Any labelled open graph defines a target linear map which corresponds to the quantum computation that we want to execute.
To ensure such a map is well defined, we provide a measurement pattern which contains a concrete sequence of instructions to generate the graph.
%\emph{Measurement patterns}~\cite{danosMeasurementCalculus2007} are a declarative language for MBQC, that describe how qubits are prepared, entangled, corrected, and measured.
\begin{definition}[\cite{danosMeasurementCalculus2007}]\label{def:meas_pattern}
  A \emph{measurement pattern} consists of an $n$-qubit register $V$ with distinguished sets $I, O \subseteq V$ of input and output qubits and a sequence of commands consisting of the following operations:
  \begin{itemize}
    \item Preparations $N_i$, which initialise a qubit $i \in \comp{I}$ in the state $\ket{+}$.
    \item Entangling operators $E_{ij}$, which apply a $CZ$-gate to two distinct qubits $i$ and $j$.
    \item Destructive measurements $M_i^{\lambda,\alpha, \cvar{s}}$, which project a qubit $i\in \comp{O}$ onto the orthonormal basis $\{\ket{+_{\lambda,\alpha}},\ket{-_{\lambda,\alpha}}\}$, where $\lambda \in \{ \XYm, \XZm, \YZm \}$ is the measurement plane, $\alpha$ is the non-corrected measurement angle .
    The projector $\ket{+_{\lambda,\alpha}}\bra{+_{\lambda,\alpha}}$ corresponds to outcome $\cvar s = 0$ and $\ket{-_{\lambda,\alpha}}\bra{-_{\lambda,\alpha}}$ corresponds to outcome $\cvar s = 1$.
    \item Corrections $[X_i]^t$, which depend on a measurement outcome (or a linear combination of measurement outcomes) $t\in\{0,1\}$ and act as the Pauli-$X$ operator on qubit $i$ if $t$ is $1$ and as the identity otherwise,
    \item Corrections $[Z_j]^s$, which depend on a measurement outcome (or a linear combination of measurement outcomes) $s\in\{0,1\}$ and act as the Pauli-$Z$ operator on qubit $j$ if $s$ is $1$ and as the identity otherwise.
  \end{itemize}
\end{definition}
A measurement pattern is \emph{runnable} if no command acts on a qubit already measured or not yet prepared (except preparation commands) and no correction depends on a qubit not yet measured.
Any runnable measurement pattern with $m$ measurement commands defines $2^m$ \emph{branches}, corresponding to the linear maps obtained by replacing the measurement 
commands with the Kraus map associated to a particular outcome and accordingly setting the corrections that depend on that outcome.
We say that a measurement pattern is \emph{deterministic} if all the branches of the pattern are proportional to each other.
In other words, all branches implement the same linear map (possibly with a different probability for each).
A pattern is \emph{strongly deterministic} if it is deterministic and all branches are equal up to a global phase,
i.e.\@ all branches have the same probability.
The pattern is \emph{uniformly deterministic} if it is deterministic for all choices of measurement angles $\alpha_i$.
It is \emph{step-wise deterministic} if the $m$ sub-patterns, obtained by truncating the sequence after a measurement command $M_i^\lambda$
and adding back all the corrections depending on qubit $i$, are deterministic.

\begin{definition}\label{def:ogs-to-linear-map}
 Suppose $\mathcal{M} = (G, I, O, \lambda, \alpha)$ is a labelled open graph.
 The \emph{target linear map of $\mathcal{M}$} is given by
 \[
  T(\mathcal{M}) \coloneqq \left( \prod_{i\in\comp{O}} \bra{+_{\lambda(i),\alpha(i)}}_i \right) E_G N_{\comp{I}},
 \]
  where $E_G \coloneqq \prod_{i\sim j} E_{ij}$ and $N_{\comp{I}} \coloneqq \prod_{i\in\comp{I}} N_i$.
\end{definition}
We may represent $T(\mathcal{M})$ in the ZX calculus by attaching the appropriate effects to the dangling qubits in the graph state.
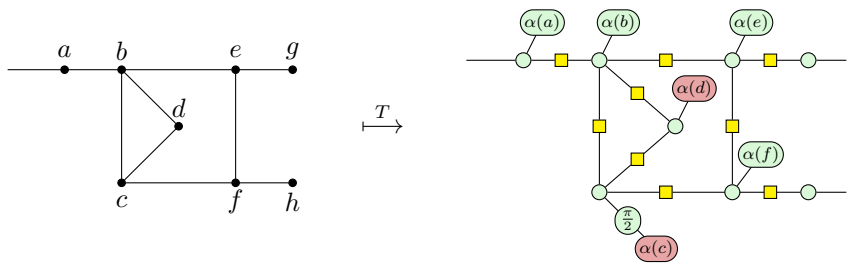
\begin{example}
  \[
    \input{figures/mbqc-example-pattern-graph.tikz}
    \qquad\overset{T}{\longmapsto}\qquad
    \input{figures/mbqc-example-pattern.tikz}
  \]
  where the input set is $I = \set{a}$, the set of outputs is $O = \set{g, h}$, and the measurement planes are $\lambda(v) = \XYm$ for all $v \in \set{a, b, e, f}$, $\lambda(d) = \YZm$, and $\lambda(c) = \XZm$.
\end{example}
The above description only discussed MBQC with post-selected measurement outcomes, i.e.\@ assuming determinism of measurements.
However, quantum measurements are fundamentally probabilistic processes: they may or may not induce Pauli errors upon observation.
Measurements with potential errors are projections onto the orthonormal basis $\{\ket{+_{\lambda,\alpha}},\ket{-_{\lambda,\alpha}}\}$, where $\lambda \in \{ \XYm, \XZm, \YZm \}$ is the measurement plane and $\alpha$ is the measurement angle.
In ZX calculus, these are given by an additional $k \pi$ phase in the measurement, $\begin{tikzpicture}
	\begin{pgfonlayer}{nodelayer}
		\node [style={rn_phase}] (0) at (1.75, 0) {$\alpha + \cvar{k} \pi$};
		\node [style=box] (1) at (0, 0) {$\lambda$};
		\node [style=none] (2) at (-1.25, 0) {};
	\end{pgfonlayer}
	\begin{pgfonlayer}{edgelayer}
		\draw (2.center) to (1);
		\draw (1) to (0);
	\end{pgfonlayer}
\end{tikzpicture}
$.
Here, $\lambda$ is given by the Hadamard gate, the S gate, and the identity for $\XYm, \XZm$, and $\YZm$, respectively.
In order to describe how measurement outcomes can be corrected, we introduce a lower-level language to specify MBQC programs.

\subsection{Flow structure}
Flow structure gives sufficient (and sometimes necessary) conditions for a labelled open graph to be implementable by a deterministic measurement pattern.
It incorporates a time-ordering of the measurements and a function that indicates where to correct undesired measurement outcomes.
Gflow (or generalised flow) is a specific type of flow structure that ensures that the target linear map is an isometry for all choices of measurement angles.
\begin{definition}
  For a graph $G = (V, E)$ and a subset of its vertices $K \subseteq G$,
  let $\Odd(K) \coloneqq \set{u \in V : \abs{N(u) \cap K} \equiv 1  \!\mod 2}$ be the \emph{odd neighbourhood} of $K$ in $G$,
  where $N(u)$ is the set of neighbours of $u$.
\end{definition}
\begin{definition}[Generalized flow~\cite{browneGeneralizedFlowDeterminism2007}]
  \label{def:gflow}
  An open graph $(G, I, O)$ labelled with measurement planes $\lambda: \comp{O} \to \set{\XYm, \XZm, \YZm}$ has generalized flow
  if there exists a map $g: \comp{O} \to \cal{P}(\comp{I})$, where $\cal{P}$ is the power set function, and a strict partial order $<$ over $V$ such that for all $v \in \comp{O}$:
  \begin{enumerate}
    \item for all $w \in g(v)$ if $v \neq w$ then $v < w$
    \item for all $w \in \Odd(g(v))$ if $v \neq w$ then $v < w$
    \item $\lambda(v) = \XYm \, \implies \, v \notin g(v) \land v\in \Odd(g(v))$
    \item $\lambda(v) = \XZm \, \implies \, v \in g(v) \land v\in \Odd(g(v))$
    \item $\lambda(v) = \YZm \, \implies \, v \in g(v) \land v\notin \Odd(g(v))$
  \end{enumerate}
  The set $g(v)$ is called the \emph{correction set of $v$}.
\end{definition}

Extending the notion of gflow, \emph{Pauli flow} allows vertices to be measured in a Pauli basis.
In this setting, the function $\lambda$ defining measurement planes is of type $\lambda: \comp{O} \to \set{\XYm, \XZm, \YZm, \Xm, \Ym, \Zm}$,
while the function $\alpha$ is only defined for nodes $v \in G$ when $\lambda(v) \in \set{\XYm, \XZm, \YZm}$.
In other words, the pattern specifies vertices that are measured in the $X$, $Y$, or $Z$ basis.
For these specific measurements, the correction set is less restricted, and we obtain the conditions below.
\begin{definition}[Pauli flow~\cite{browneGeneralizedFlowDeterminism2007, simmonsRelatingMeasurementPatterns2021}]
  \label{def:pauli-flow}
  An open graph $(G, I, O)$ labelled with measurement planes $\lambda: \comp{O} \to \set{\XYm, \XZm, \YZm, \Xm, \Ym, \Zm}$ has Pauli flow
  if there exists a map $p: \comp{O} \to \cal{P}(\comp{I})$ and a strict partial order $<$ over $V$ such that:
  \begin{enumerate}
    \item for all $w \in p(v)$ if $\lambda(w) \notin \set{\Xm, \Ym} \land v \neq w$ then $v < w$
    \item for all $w \in \Odd(p(v))$ if $\lambda(w) \notin \set{\Ym, \Zm} \land v \neq w$ then $v < w$
    \item for all $w \leq v$ if $\lambda(w) = \Ym \land v \neq w$ then $(w \in p(v) \Longleftrightarrow w \in \Odd(p(v)))$
    \item $\lambda(v) = \XYm \, \implies \, v \notin p(v) \land v\in \Odd(p(v))$
    \item $\lambda(v) = \XZm \, \implies \, v \in p(v) \land v\in \Odd(p(v))$
    \item $\lambda(v) = \YZm \, \implies \, v \in p(v) \land v\notin \Odd(p(v))$
    \item $\lambda(v) = \Xm \, \implies \, v \in \Odd(p(v))$
    \item $\lambda(v) = \Zm \, \implies \, v \in p(v)$
    \item $\lambda(v) = \Ym \, \implies \, (v \notin p(v) \land v \in \Odd(p(v))) \lor (v \in p(v) \land v\notin \Odd(p(v)))$
  \end{enumerate}
\end{definition}

To understand the definition above, first note that for measurements in the planes $\set{\XYm, \XZm, \YZm}$, the conditions 
are the same as for gflow.
The above conditions $7-9$ are obtained by taking the pairwise disjunctions \enquote{$\lor$}
of conditions $4-6$, using the fact that each Pauli measurement belongs to a pair of planes.   
To obtain condition $1$, note that a Pauli $X$ error on a qubit measured in the $X$ basis only induces a global phase on the state.
Therefore we must not correct $X$ errors on $X$ measurements.
Condition $2$ is the equivalent condition for $Z$ errors
and condition $3$ ensures that $Y$ measurements need only carry $Y = XZ$ corrections. 
A consequence is that $Y$ measurements in a graph with Pauli flow need not carry corrections, justifying conditions $1-2$.

We can now state the main result of~\cite{browneGeneralizedFlowDeterminism2007} which ensures that labelled open graphs with flow are implementable by deterministic patterns.
\begin{theorem}
  If a labelled open graph $\mathcal{M}$ has flow, then the pattern defined by:
  \[\prod^<_i (X^{\cvar{s}_i}_{g(i) \cap \{j \vert i < j\}} Z^{\cvar{s}_i}_{\Odd(g(i))\cap \{j \vert i < j\}} M_i^{\lambda_i, \alpha_i, \cvar{s}_i}) E_G N_{\comp{I}}\]
  where $\prod^{<}$ denotes concatenation in the order $<$,
 is runnable, uniformly, strongly and step-wise deterministic and realises the target linear map $T(\mathcal{M})$, which is guaranteed to be an isometry.
\end{theorem}
The theorem indicates that $X$ corrections will be performed in $g(v) - \set{v}$ and $Z$ corrections in $\Odd(g(v)) - \set{v}$, for any qubit $v \in \comp{O}$.
A converse version of this theorem also holds for gflow~\cite{browneGeneralizedFlowDeterminism2007}.
Moreover, any qubit circuit can be turned into labelled open graph satisfying the gflow conditions,
which ensures that MBQC can perform universal quantum computation~\cite{backensThereBackAgain2021}.

\section{Characterisation of correctable fusion measurements}\label{sec:character}

In the previous section, we showed how the action of linear optical circuits on dual-rail qubits can be translated into ZX diagrams.
Beyond Type I and Type II fusions, this also enables the description of different entangling measurements that can be implemented by linear optical circuits.
This section delves into the classification and characterisation of such fusion measurements using the ZX calculus.

\begin{remark}
  Ref.~\cite{loblTransformingGraphStates2024} also analyzes the action of generalizations of fusion measurements on graph states.
  However, the generalizations of fusion they consider are local Clifford equivalent to Bell measurements and they consider only the success case.
  Here instead, we consider all measurements local unitarily equivalent to Type I fusion followed by arbitrary single-qubit measurement.
\end{remark}

\subsection{General fusion measurements}

In \cref{subsec:fusion-measurements}, we saw that Type II fusion differs from Type I fusion only by single-qubit unitaries applied before and after the fusion, as well as an additional measurement.
Similarly, we can describe the success and failure outcomes of all possible entangling measurements unitarily equivalent to Type II fusion as follows:
\begin{equation}
  \label{eq:GF}
  \input{ZX/GeneralFusion.tikz}
\end{equation}
Here, $U_1$, $U_2$, and $U_3$ are arbitrary single-qubit unitaries which, up to some global phase, can be expressed by three alternating rotations around the $Z$ and $X$ axes:
\begin{equation}
  \input{ZX/Euler.tikz}
  \label{eq:euler}
\end{equation}

Using this decomposition, we get a total of $9$ parameters to characterise a general fusion;
however, we are able to reduce this number using different observations that we support with calculations in ZX calculus.
First, one parameter can be eliminated from $U_3$ as it is followed by a single-qubit measurement that only contributes an irrelevant global phase:
\begin{equation*}
  \input{ZX/MeasureQubit.tikz}
\end{equation*}
Second, we observe that $Z(\gamma_1)$, $Z(\gamma_2)$, $Z(\alpha_3)$, and the fusion error $\cvar k \pi$ itself are simultaneously diagonalizable in the Z basis.
In other words, we can apply the spider fusion rule of the ZX calculus as follows:
\begin{equation}
  \label{eq:Fusion6}
  \input{ZX/Fusion6.tikz}
\end{equation}
where $\phi = \gamma_1 + \gamma_2 + \alpha_3$.

With this, we have reduced the number of parameters to describe an entangling measurements to $6$.
However, by only considering fusion with certain desirable properties, we can reduce this number even further.

\subsection{Green failure}

%We start by focusing on the errors induced by fusion failure.
Let us consider the fusion with $U_1$, $U_2$, and $U_3$ all being the identity, like Type I fusion composed with a Z measurement.
An unsuccessful fusion in this case acts as a projector in the $Z$-basis.
This means that in addition to failing to fuse the two nodes, it also \emph{disconnects} them from their neighbours:
\begin{equation}
  \input{ZX/FusionFailDisconnect2.tikz}
\end{equation}
In order to preserve the entanglement of the graph state, we want failures to be \enquote{green}.
\begin{definition}[Green failure]
  We say that a fusion measurement has green failure if its failure outcome satisfies:
  \begin{equation}
    \input{ZX/GreenFailure.tikz}
  \end{equation}
  for some $\theta_1, \theta_2$.
\end{definition}
\noindent This means that, upon failure, the underlying resource graph preserves its connectivity:
\begin{equation}
  \input{ZX/FusionFailGreen2.tikz}
\end{equation}

We can characterise all types of fusions up to Z rotations on the nodes by the choices of $\beta_i$ in the Euler decomposition shown in \cref{eq:euler}.
We have the following measurement outcomes for different choices of $\beta_i$, with $i \in \{0,1\}$,
\begin{equation}
  \input{ZX/FusionFailCases.tikz}
\end{equation}
In other words, the failure is either green and keeps the connection of the graph, red and disconnects the graph, or it induces a non-unitary (and thus not correctable) error on the graph.
Asking for green failure reduces $U_1$ and $U_2$ to be of the following shape:
\begin{equation}
  \label{eq:GreenU1U2}
  \input{ZX/GreenU1U2.tikz}
\end{equation}
This reduces the $6$ parameters of a general fusion to $4$ if it has green failure:
\begin{proposition}
  \label{prop:green-failure}
  Any fusion measurement with green failure has the following form:
  \begin{equation}
    \input{ZX/GreenCharacterisation.tikz}
  \end{equation}
  for some choice of angles $\alpha_1, \alpha_2, \beta, \phi \in [0, 2\pi)$ and measurement outcomes $\cvar{j}, \cvar{k} \in \set{0, 1}$.
\end{proposition}

\subsection{Pauli errors} % Fusion with correctable Pauli error

Since measurements in quantum computing are probabilistic processes, a fusion induces random errors in both its success and failure cases.
In MBQC, such errors can be propagated when the measurement pattern has flow.
To similarly propagate fusion measurement errors in our framework, they must be equivalent to local Pauli gates on the input qubits.
Diagrammatically, this means that the measurement errors $\cvar{k}\pi$ and $\cvar{j}\pi$ can be pulled out to the input wires while keeping them in X, Y or, Z basis.
\begin{definition}[Pauli error]
  \label{def-paulifusion}
  A fusion measurement has Pauli error when the success outcome satisfies:
  \begin{equation}
    \input{ZX/PauliFusion.tikz}
  \end{equation}
  for some bits $\cvar{w},\cvar{x},\cvar{y},\cvar{z} \in \set{0, 1}$.
\end{definition}
From the equation above, we deduce that either $U_1$ or $U_2$ must be Clifford, and that $U_3$ must be a gate locally equivalent to $H$, $S$, or $Id$ so that the measurement is in the \YZ, \XZ, or \XY plane, respectively.
Further requiring that failure is green gives us the following characterisation.
\begin{restatable}{proposition}{characterization}
  \label{thm-characterisation}
  Any fusion measurement with green failure and Pauli error has the following form:
  \begin{equation}
    \input{ZX/PauliCharacterisation.tikz}
  \end{equation}
  for a measurement plane $\lambda \in \set{\YZm, \XZm, \XYm}$, angles $\alpha, \omega \in [0, 2\pi)$, and a choice
  of Clifford parameter $d \in \set{0, 1, 2, 3}$.
\end{restatable}
\noindent The detailed proof is given in \cref{app:fusion}.

\begin{definition}
  We call YZ, XZ, and XY fusion the three classes of fusion obtained by the choice of $\lambda$:
  \begin{equation}
    \input{ZX/PauliLambdaXYZ.tikz}
  \end{equation}
\end{definition}

In practical applications, it is desirable that the action of a fusion measurement on its target qubits is symmetric,
so that errors can be propagated on either qubit at will.
\begin{definition}[Symmetric fusion]
  We say that a fusion measurement is symmetric if it is invariant under swap in the success case, that is,
  \begin{equation}
    \input{fusion/SymmetricFusion.tikz}
  \end{equation}
\end{definition}
\begin{theorem}
  \label{prop:characterisation}
  Any symmetric fusion measurement with green failure and Pauli error has the following form:
  \begin{equation}
    \input{ZX/Characterisation.tikz}
  \end{equation}
  where $c \in \set{0, 1}$, $\lambda \in \set{\YZm, \XZm, \XYm}$, and $\alpha \in [0, 2\pi)$.
\end{theorem}
\begin{proof}
  This directly follows from \cref{thm-characterisation} proved in \cref{app:fusion}.
\end{proof}

\noindent
This gives us the following three success cases depending on the choice of $\lambda$:\\
\tabulinesep=3mm
\begin{tabu}
  to \textwidth { X[$$c] | X[$$c] | X[$$c] }
  \lambda = \YZm             & \lambda = \XZm             & \lambda = \XYm             \\
  \input{fusion/Fusion-YZ.tikz} & \input{fusion/Fusion-XZ.tikz} & \input{fusion/Fusion-XY.tikz} \\
\end{tabu}

\noindent These fusions can be used to implement a large family of entangling operations,
such as phase gadgets~\cite{kissingerReducingNumberNonClifford2020}.

\begin{example}[Phase gadgets]
  A \emph{phase gadget} is an entangling gate that plays an important role in quantum circuit optimization~\cite{debeaudrapFastEffectiveTechniques2020, vandeweteringOptimalCompilationParametrised2024}
and quantum machine learning where they allow tuning the amount of entanglement between their inputs.
  A phase gadget is an instance of a \YZ-fusion with $c = 0$ and $\alpha \in [0, 2\pi)$:
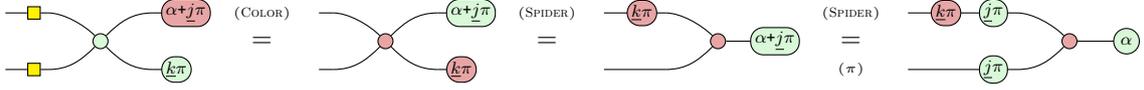
  \[
    \input{fusion/PhaseGadget.tikz}
  \]
\end{example}

\subsection{X and Y fusions}\label{sec:three-fusions}

The characterisation that we obtained for fusions with green failure and Pauli error is three-fold, corresponding to
the three planes on the Bloch sphere.
We now consider measurements with the additional property of being Pauli measurements, in the X, Y or Z basis.

\begin{definition}[Pauli green fusion]
  We say that a fusion measurement is Pauli green if it is of the form given in \cref{prop:characterisation} with $\alpha$ being a multiple of $\frac{\pi}{2}$.
\end{definition}
\noindent This gives us the following three fusions:\\
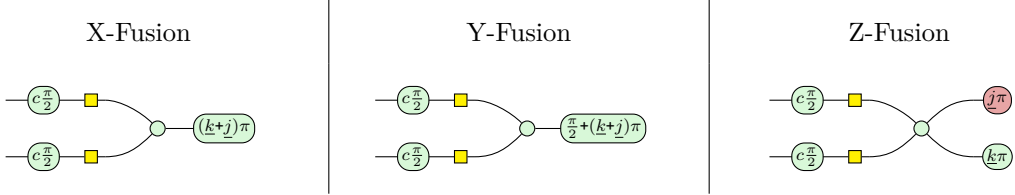
\begin{tabu}
  to \textwidth { X[c] | X[c] | X[c] }
  X-Fusion                  & Y-Fusion                  & Z-Fusion                  \\
  \input{fusion/Fusion-X.tikz} & \input{fusion/Fusion-Y.tikz} & \input{fusion/Fusion-Z.tikz} \\
\end{tabu}

\noindent Suppose we write $\alpha = a \frac{\pi}{2}$ where $a \in \{0, 1, 2, 3\}$. Then from YZ fusion we obtain $Z$-fusion when $a$ is even and $Y$-fusion when odd, from XY fusion we obtain $X$-fusion when $a$ is odd and $Z$-fusion when even, and from $XY$ fusion we obtain $X$-fusion when $a$ is even and $Y$-fusion when odd.

Note first that the Z-fusion is trivial: it is a separable two-qubit measurement and leaves the connectivity
of the graph state unchanged.
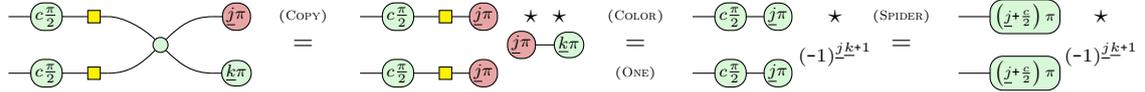
\[
  \input{fusion/Fusion-Z-disconnects.tikz}
\]
X and Y fusions instead are entangling measurements that qualitatively change the connectivity of
the graph: they either fuse two nodes into one (X-fusion) or add a hadamard edge between them (Y-fusion).

\begin{example}[Type II as X-fusion]
  The Type II fusion~\cite{browneResourceEfficientLinearOptical2005} is an instance of $X$-fusion with $c = 0$:
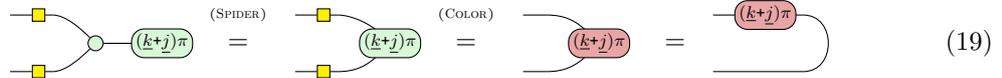
\begin{equation}
  \input{fusion/XFusion.tikz}
\end{equation}
Note that setting $c = 1$ is undesirable in this case as it only changes the errors from the X to the Y basis.
\end{example}

\begin{example}[CZ with Y-fusion]
  \label{ex:Y-fusion}
  The fusion measurement for performing CZ gates with linear optics, studied in~\cite{limRepeatUntilSuccessLinearOptics2005,degliniastySpinOpticalQuantumComputing2024},
  is an instance of $Y$-fusion.
  Indeed, up to Pauli errors, $Y$-fusion with $c = 1$ adds a Hadamard edge in the success case:
  \[
    \input{fusion/HFusion.tikz}
  \]
\end{example}

\begin{remark}
  Recall from \cref{remark:nType1} that Type II fusions can be generalized to an arbitrary number of input legs.
  Similarly, $Y$-fusion can also be generalized to any number of inputs.
  Its action corresponds to applying a CZ gate between each pair of qubits:
  \[
    \input{fusion/nHFusion.tikz}
  \]
  where the connections of the spiders form a complete graph on the right-hand side.
  This rewrite rule corresponds to toggling the CZ edges between all of the nodes being fused in the fusion of the underlying graph state; a formal proof can be derived from~\cite[Lemma 5.2.]{duncanGraphtheoreticSimplificationQuantum2020}.
\end{remark}

\begin{proposition}
  Up to local Clifford rotation on the target qubits, entangling Pauli green fusions are either X or Y fusions.\\
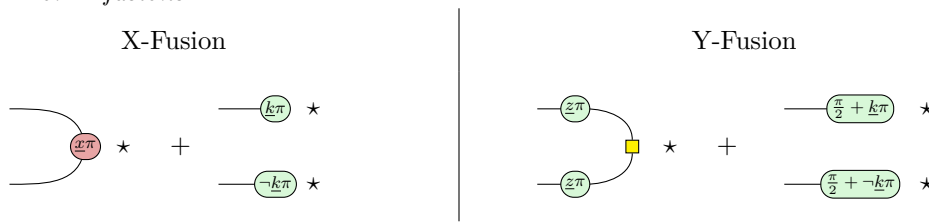
  \begin{tabu}
    to \textwidth { X[c] | X[c] }
    \textnormal{X-Fusion}     & \textnormal{Y-Fusion}     \\
    \input{fusion/X-Fusion.tikz} & \input{fusion/Y-Fusion.tikz} \\
  \end{tabu}
\end{proposition}

\section{Flow structure for fusion networks}\label{sec:flow}

In the previous section we characterised fusion measurements that induce Pauli errors on their input qubits.
The aim of this section is to describe the flow structure that enables correction of these Pauli errors.
We give a general definition of fusion network, but we then focus on developing a notion of flow --- called XY-flow ---
for a subclass of fusion networks that use exclusively X and Y fusions.
Following~\cite{browneGeneralizedFlowDeterminism2007}, we introduce the notion of an XY-fusion pattern describing an FBQC computation as a sequence of instructions.
Assuming that all fusions are successful, we show that any XY-fusion network with XY-flow can be implemented deterministically by an XY-fusion pattern.
The resulting pattern can moreover be factorised such that all fusions appear before single-qubit measurements.
Finally, we show how that any decomposition of a labelled open graph as an XY-fusion network has XY-flow
provided that the original open graph has Pauli flow.

\begin{remark}
  A related prior work~\cite{debeaudrapPauliFusionComputational2020} introduces a notion of flow for a similar model of computation.
  Their model is defined on a two-qubit fusion measurement arising from lattice surgery operations.
  However, these induce different errors than those appearing in photonic Type I and Type II fusion, which are outside the Pauli Fusion model.
\end{remark}

\subsection{Fusion networks}\label{subsec:fusion-network}

In FBQC, a resource state of photons is prepared and it is probed by a sequence of
destructive fusions and single-qubit measurements.
A \emph{fusion network} specifies a configuration of fusions and single-qubit measurement to be performed on the resource state.
In our definition, we assume that each node in the resource graph state is implemented by multiple photons ---
one for each fusion and one for the single-qubit measurement --- entangled as a GHZ state to each other.
This is equivalent to treating fusion as a non-destructive measurement, acting as follows in the success case:
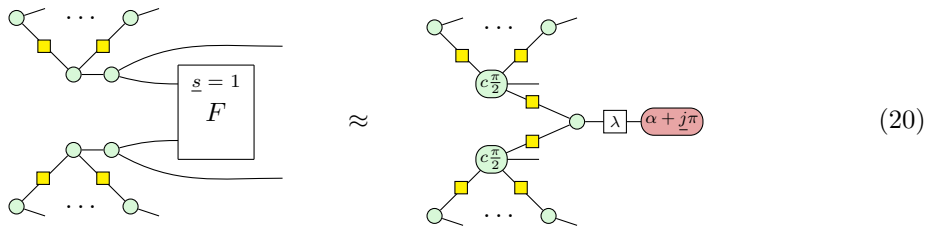
\begin{equation}
  \input{fnetwork1.tikz}
\end{equation}
This approach is convenient as it allows us to directly relate fusion networks and the standard MBQC notion of labelled open graph.
It also allows us to obtain a general form for the probability of success, as shown in \cref{subsec:fusion-prob}.
Moreover, we may recover the action of destructive fusions by measuring the remaining photons in the $X$ basis and,
as shown in \cref{subsec:destructive-networks}, any fusion network in our sense gives rise to an equivalent destructive one.
We can thus give a general definition of fusion networks, allowing any symmetric fusion with green failure and Pauli error, as
described in \cref{prop:characterisation}.

\begin{definition}[Fusion network]
  A fusion network, denoted by $\mathcal{F} = (G, I, O, F, \lambda, \alpha, c)$, is given by the following:
  \begin{enumerate}
    \item an open graph $(G, I, O)$ (called \enquote{resource graph}),
    \item a set of fusions $F \sub M(\comp{O} \times \comp{O})$,
    \item an assignment of measurement planes $\lambda: \comp{O} + F \to \set{\XYm, \XZm, \YZm, X, Y, Z}$,
    \item an assignment of measurement angles $\alpha: \comp{O} + F \to [0, 2\pi)$ ($\alpha(v)$ is set to zero if $\lambda(v) \in \set{X, Y, Z}$), and
    \item a Clifford parameter for each qubit $c: \comp{O} \to \set{0, 1, 2, 3}$.
  \end{enumerate}
  where $M$ denotes the multi-set construction and $+$ denotes the disjoint union.
\end{definition}

\begin{remark}
  The definition of fusion network given here references a single resource graph.
  In practice, the graph $G$ is the disjoint union of multiple copies of the same basic resource state, such as the ones depicted in \cref{fig:example-resuorce-states}.
\end{remark}

Following \cref{prop:characterisation}, a successful fusion has the effect of introducing an additional
node in the graph, measured in an arbitrary plane and angle.
Thus, any fusion network $\mathcal{F}$ defines a \emph{target open graph}, denoted $\mathcal{M_F}$, capturing the computation performed when we post select on fusion successes.

\begin{definition}[Target open graph]
  Given a fusion network $\mathcal{F} = (G, I, O, F, \lambda, \alpha, c)$ with $G = (V, E)$.
  Any fusion $f \in F$ contributes an extra vertex to the graph, labelled $v_f$.
  The \emph{target open graph} of $\mathcal{F}$ is $\mathcal{M_F} \coloneqq (G_\mathcal{F} = (V_\mathcal{F}, E_\mathcal{F}), I, O, \lambda_\mathcal{F}, \alpha_\mathcal{F})$,
  where:
  \begin{gather*}
      V_\mathcal{F} = V \cup \set{v_f | f \in F} \qquad
      E_\mathcal{F} = E \cup \set{ (v_f, w) | w \text{ belongs to } f \in F }\\
      \lambda_\mathcal{F}(u) =
    \begin{dcases}
      \lambda(f), & \text{if } u = v_f \text{ for some } f \in F \\
      \YZm, & \lambda(u) = \XZm \land c(u) \bmod 2 \equiv 1 \\
      \XZm, & \lambda(u) = \YZm \land c(u) \bmod 2 \equiv 1 \\
      \lambda(u), & \text{otherwise}
    \end{dcases}
    \qquad
    \alpha_\mathcal{F}(u) =
    \begin{dcases}
      \alpha(f), & \hspace*{-1.2cm}\text{if } u = v_f \text{ for some } f \in F \\
      \alpha(u) + \frac{c(u)\pi}{2}, & \lambda(u) = \XYm\\
      (-1)^{\left\lceil\! \frac{c(u)}{2}\! \right\rceil} \alpha(u), & \lambda(u) = XZ \\
      (-1)^{\left\lfloor\! \frac{c(u)}{2}\! \right\rfloor} \alpha(u), & \lambda(u) = YZ
    \end{dcases}
  \end{gather*}
  The \emph{target linear map} of the fusion network $T(\mathcal{F})$ is the target linear map of $\mathcal{M_F}$.
\end{definition}

In other words, nodes and edges of $G$ are extended with those coming from the set of fusions $F$.
The fusion measurement planes and angles are part of the new single-qubit measurement parameters.
Furthermore, some of the original single-qubit measurements are modified if their Clifford parameters are non-zero.
Clifford parameters on measured qubits correspond to changes in measurement planes.
We can capture these by the following equations:
\[
  \input{figures/fusion/clifford-parameter-change.tikz}
\]

\begin{example}\label{ex:network}
  \label{ex:fusion-network}
  Consider a fusion network with a pair of lines as the resource graph and two fusions.
  The target measurement graph is obtained by adding a new node in the graph for each fusion.
  \[
    \input{figures/network-graph.tikz}
    \qquad \longrightarrow \qquad
    \input{figures/network-target1.tikz}
  \]
\end{example}

We are interested in a particular subclass of fusion networks that use $X$ and $Y$ fusions exclusively.

\begin{definition}[XY-fusion network]
  An \emph{XY-fusion network} is a tuple $\mathcal{F} = (G, I, O, F, \lambda, \alpha)$, where
  \begin{itemize}
    \item $(G, I, O)$ is an open graph,
    \item $F \sub M(\comp{O} \times \comp{O})$ is a set of fusions,
    \item $\lambda = (\lambda_V, \lambda_F)$ with $\lambda_V: \comp{O} \to \set{\XYm, \XZm, \YZm, X, Y, Z}$ and $\lambda_F: F \to \set{\Xm, \Ym}$ assigns measurement planes to single-qubit measurement and fusions, respectively, and
    \item $\alpha: \comp{O} \to [0, 2\pi)$ assigns measurement angles to each non-output qubit.
  \end{itemize}
\end{definition}

In this case, the target open graph can be further simplified: $X$ fusions merge the nodes they are applied to into one,
and $Y$ fusions add a hadamard edge between them.

\begin{definition}[Simplified target graph]
  The simplified target graph of $\mathcal{F}$ is $\mathcal{M_F} \coloneqq (G_\mathcal{F} = (V_\mathcal{F}, E_\mathcal{F}), I, O, \lambda_F, \alpha)$, where
\begin{itemize}
  \item $V_\mathcal{F} = (V \setminus V_X) \cup \set{v_f | f \in F, \lambda_F(f) = \Xm }$,
  \item $E_\mathcal{F} = E_{unchanged} \cup E_X \cup E_Y$ where
  \begin{itemize}
    \item $E_{unchanged} = \set{e | (v, w) = e \in E \text{ where } v, w \notin V_X}$ are unchanged edges,
    \item $E_X = \set{ (v_f, w) | v \in V_X | w \in N(v) | f \in F \text{ where } v \text{ belongs to } f }$ are edges connected to new vertices that are added by X-fusions, and
    \item $E_Y = \set{ f \in F \text{ where } \lambda_F(f) = Y }$ are the extra edges added by Y-fusions.
  \end{itemize}
\end{itemize}
  where $V_X = \set{v | f \in F \text{ where } v \text{ belongs to } f \text{ and } \lambda(f) = \Xm}$.
Note that the target linear map of $\mathcal{G_F}$ is $T(\mathcal{M_F})$.
\end{definition}

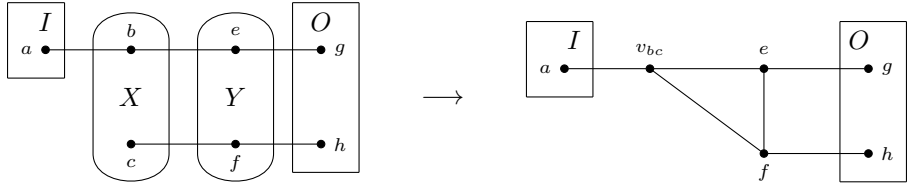
\begin{example}\label{ex:XYnetwork}
  \[
    \input{figures/network-fusion.tikz}
    \qquad \longrightarrow \qquad
    \input{figures/network-target2.tikz}
  \]
\end{example}

\subsection{XY-fusion patterns}

Following the literature on flow~\cite{browneGeneralizedFlowDeterminism2007,danosExtendedMeasurementCalculus2009,backensThereBackAgain2021},
we can now define a notion of pattern that specifies the linear map implemented by a fusion-based computation as a sequence of operations.

\begin{definition}{XY-fusion pattern}
    A \emph{XY-fusion pattern} consists of an $n$-qubit register $V$ with distinguished sets $I, O \subseteq V$ of input and output qubits and a sequence of commands consisting of the following operations:
  \begin{itemize}
    \item Preparations $N_i$, which initialise a qubit $i \in \comp{I}$ in the state $\ket{+}$.
    \item Entangling operators $E_{ij}$, which apply a $CZ$-gate to two distinct qubits $i$ and $j$.
    \item Destructive fusions $F^{\lambda, \cvar{s}, \cvar{k}}_{ij}$ where $\lambda \in \set{\Xm, \Ym}$: $F^{X}$ is an X-fusion and $F^Y$ is a Y-fusion, and $s, k \in \{0,1\}$ are the success outcome and measurement outcome, respectively.
    \item Destructive measurements $M_i^{\lambda, \alpha, \cvar{k}}$, which project a qubit $i\notin O$ onto the orthonormal basis $\{\ket{+_{\lambda, \alpha}},\ket{-_{\lambda, \alpha}}\}$, where $\lambda$ is the measurement plane, $\alpha$ is the measurement angle and $k \in\{0,1\}$ indicates the measurement outcome.
    The projector $\ket{+_{\lambda, \alpha}}\bra{+_{\lambda, \alpha}}$ corresponds to outcome $\cvar{k} = 0$ and $\ket{-_{\lambda, \alpha}}\bra{-_{\lambda, \alpha}}$ corresponds to outcome $\cvar k = 1$.
    \item Corrections $[X_i]^k$, which depend on a measurement outcome (or a linear combination of measurement outcomes) $k\in\{0,1\}$ and act as the Pauli-$X$ operator on qubit $i$ if $k$ is $1$ and as the identity otherwise,
    \item Corrections $[Z_j]^l$, which depend on a measurement outcome (or a linear combination of measurement outcomes) $l\in\{0,1\}$ and act as the Pauli-$Z$ operator on qubit $j$ if $l$ is $1$ and as the identity otherwise.
  \end{itemize}
\end{definition}
\begin{remark}
  Note that this definition assumes that all the correction commands are performed on the qubits in the register.
  They can therefore not be applied on fusion nodes, i.e. on the qubits resulting from a Type I fusion.
  This simplified model is sufficient in the setting of XY-fusions. For fusions with arbitrary measurement planes and angles,
  a more refined definition of pattern may be required.
\end{remark}

An XY-fusion pattern is runnable if no command acts on a qubit already measured or not yet prepared (except preparation commands) and no correction depends on a qubit not yet measured.
Any runnable XY-fusion pattern has an underlying XY-fusion network given by forgetting the correction commands and adding inputs for the qubits in $I$ and outputs for the qubits in $O$.
The resource graph is given by the entangling commands $E_{ij}$,
the fusion pairs are given by the commands $F_{ij}^\lambda$ and
the measurement labels by the commands $M_i^{\lambda, \alpha}$.
Any XY-fusion pattern with $m$ single-qubit measurements and $f$ fusions defines $2^{m + 2f}$ branches given by post-selecting on the two possible outcomes of $M_i^{\lambda, \alpha}$ commands
and the $4$ possible outcomes of fusion commands $F_{ij}^\lambda$.
The \emph{success branches} are the $2^{m + f}$ branches where every fusion is successful.
We say that an XY-fusion pattern is \emph{deterministic on success} if all the success branches are proportional,
i.e. they implement the same linear map.
Similarly as for measurement patterns, we may define strong, uniform and step-wise determinism for the success branches of
XY-fusion patterns; see \cref{subsec:MBQC}.

\begin{example}\label{ex:XYpattern}
  As an example, consider the pattern defined by the following sequence:
  \[[X_4]^{l} [Z_2]^{j} [X_2]^{k} M^{\XYm, \alpha, \cvar{l}}_{3} M^{X, \cvar{k}}_{1} F^{X, \cvar{s}, \cvar{j}}_{13} E_{34} E_{12} N_4 N_3 N_2 N_1 \, .\]
  This pattern has $8$ success branches obtained by setting the different values of $k, l, j \in \{0, 1 \}$ in the following ZX diagram:
  \[
    \input{fpattern.tikz}
  \]
  where two stars have been cancelled by the scalars from the two entangling gates.
  By rewriting the ZX diagram above, we can show that these $8$ branches are proportional to each other:
  \[
    \scalebox{.9}{\input{fpattern-rewrite.tikz}}
  \]
  Therefore, this specific pattern is deterministic on success.
  Moreover, each of the branches carries the same scalar, making the pattern strongly deterministic.
  Since the rewrite above holds for any angle $\alpha$, the pattern is also uniformly deterministic.
  By considering the pattern truncated at single-qubit measurement commands, a similar rewrite shows that it is also step-wise deterministic on success.
\end{example}
\begin{remark}
  In this section we are only interested in \emph{proportionality} between linear maps, and we do not consider the probabilities of individual outcomes.
  We will analyse these probabilities in \cref{sec:universality}.
\end{remark}

\subsection{XY-flow and determinism on success}

We now define a notion of flow for fusion networks that makes them deterministically implementable by a fusion pattern.
Following~\cite{browneGeneralizedFlowDeterminism2007}, we prove that our notion of flow is both necessary and sufficient
for an XY-fusion pattern to be uniformly, strongly and stepwise deterministic on success.
Moreover, every such pattern can be factorized such that all fusions appear before single-qubit measurements.

\begin{definition}[XY-flow]
  \label{def:flow}
  An XY-flow for an XY-fusion network $\mathcal{F} = (G, I, O, F, \lambda, \alpha)$ is a Pauli flow $(p, \leq)$ on the target open graph $\mathcal{M_F}$,
  such that no corrections need to be applied on fusion nodes.
  Concretely, for any fusion node $f \in F$:
  \begin{center}
    if $\lambda(f) = X$ then for any $v \in \mathcal{M_F}$, $ f\notin \Odd(p(v))$.
  \end{center}
\end{definition}

\begin{remark}
  Note that the condition above is precisely what is necessary to define a flow on the target open graph which does not require corrections on fusion nodes.
  For $Y$-measured nodes this is already the case by conditions $1-3$ in \cref{def:pauli-flow}, so we do not need to impose additional conditions.
  To extend this definition to general fusion networks, one may use the condition that for any fusion node $f \in F$ and
  vertex $v \in M_\mathcal{F}$, $f \notin p(v)$ and $f \notin \Odd(p(v))$.
\end{remark}

\begin{theorem}\label{thm:flow-pattern}
  Given an XY-fusion network $\mathcal{F}$ with XY-flow $(p, <)$ the XY-fusion pattern defined by:
  \[
    \left( \prod^<_i X^{k_i}_{g(i)} Z^{k_i}_{\Odd(g(i))} M_i^{\lambda_i, \alpha_i, \cvar{k_i}}\right)
    \left(\prod_{f = (i, j) \in F} X^{k_f}_{g(f)} Z^{k_f}_{\Odd(g(f))} F_{ij}^{\lambda(f), \cvar{s_f}, \cvar{k_f}}\right)
    E_G N_{\comp{I}}
  \]
  where $g(i) = p(i)\, \cap \{ j \, \vert \, i < j \}$, $\prod^{<}$ denotes concatenation in the order $<$ and
  $\prod$ denotes concatenation in any order,
  is uniformly, strongly, and stepwise deterministic on success and implements the target linear map $T(\mathcal{F})$ when all fusions are successful.
\end{theorem}
\begin{proof}
  Since $\mathcal{M_F}$ has Pauli flow, we have a correction function $p: \comp{O} + F \to \mathcal{P}(\comp{O} + F)$
  satisfying the Pauli flow conditions.
  Since the errors in the success branches of $\mathcal{F}$ correspond exactly to the errors in $\mathcal{M_F}$,
  by~\cite[Theorem 4]{browneGeneralizedFlowDeterminism2007}, $\mathcal{F}$ is uniformly, strongly and step-wise deterministic on success
  and implements the target linear map $T(\mathcal{F}) = T(\mathcal{M_F})$.
  It remains to show that we can factorise the pattern as above.
  Therefore, by conditions $1-2$ in \cref{def:pauli-flow}, if $\lambda(f) = Y$, we can set $f < v$ for any vertex $v \in \comp{O} + F$.
  If $\lambda(f) = X$, by \cref{def:flow} we have $f \notin \Odd(p(v))$ for any node $v$ in $\mathcal{M_F}$, and by condition $1$ in \cref{def:pauli-flow} we can set $f < v$ for any $v \in \comp{O} + F$.
  This gives us the factorisation required, where every fusion appears before single qubit measurements.
\end{proof}

\begin{theorem}
  If a runnable XY-fusion pattern is uniformly strongly, and stepwise deterministic on success, then the underlying XY fusion network has XY-flow.
\end{theorem}
\begin{proof}
  For any runnable XY-fusion pattern, we may construct a corresponding measurement pattern where for each fusion command $F_{ij}^\lambda$,
  a new vertex $f$ is added to the qubit register, and the command is replaced by $M_{f}^\lambda(f) E_{if} E_{fj}$.
  Then the success branches of the XY-fusion pattern are proportional to the $2^{m + \absolutevalue{F}}$ branches of the resulting measurement pattern.
  The underlying geometry of this measurement pattern is precisely the target open graph of the underlying network $\mathcal{F}$ of the fusion pattern.
  Since the measurement pattern is uniform, strongly and stepwise deterministic, its underlying geometry must have gflow, by~\cite[Theorem 3]{browneGeneralizedFlowDeterminism2007}.
  Moreover, since the XY-fusion pattern only contains corrections on qubits in the register, there is a flow for $\mathcal{M_F}$ with no corrections on fusion nodes.
  Therefore $\mathcal{F}$ has XY-flow.
\end{proof}

\subsection{Decomposing open graphs as XY-fusion networks}\label{subsec:flow-preserving}

Suppose we wish to implement a given labelled open graph $\mathcal{G}$ as an XY-fusion network.
Several different fusion networks may exist that have $\mathcal{G}$ as their simplified target graph.
We now show that any such decomposition of $\mathcal{G}$ as a fusion network $\mathcal{F}$ is guaranteed to have XY-flow,
provided that $\mathcal{G}$ has Pauli flow.

\begin{restatable}[X-fusion]{proposition}{xFusionFlowPreserving}\label{prop:Xflow-preserving}
  The following open graph rewrite preserves the existence of Pauli flow:
  \[
    \input{fusion/XFusionFlowPreservingGraph.tikz}
  \]
  where $\lambda(f) = \lambda(b) = X$, $\lambda(a) = \lambda(v_f)$, and $\alpha(a) = \alpha(v_f)$.
\end{restatable}

\begin{restatable}[Y-fusion]{proposition}{yFusionFlowPreserving}
  The following open graph rewrite preserves the existence of Pauli flow:
  \[
    \input{fusion/YFusionFlowPreservingGraph.tikz}
  \]
  where $\lambda(f) = Y$, $c(a) = c(b) = 0$ on the left and $c(a) = c(b) = 1$ on the right-hand side.
\end{restatable}
The proofs are in Appendix~\ref{app:flow}.

We can now show that Pauli flow on the simplified target open graph $\mathcal{G_F}$ is both necessary and sufficient for $\mathcal{F}$ to have XY-flow.

\begin{theorem}\label{thm:flow-simplified-graph}
  An XY-fusion network $\mathcal{F}$ has XY-flow if and only if the simplified target graph $\mathcal{G_F}$ has Pauli flow.
\end{theorem}
\begin{proof}
  This follows from the two propositions above.
  For $X$-fusion we moreover need to show that when rewriting from $\mathcal{G_F}$ to $\mathcal{M_F}$ the newly introduced fusion node is not in the odd neighbourhood of some correction set.
  Using the notation of \cref{prop:Xflow-preserving},
  suppose that $v \in \Odd(g(u))$ for some node $u$ in $\mathcal{G_F}$ where $g$ is the Pauli flow on $\mathcal{G_F}$,
  then if $u$ is a neighbour of $v$, it is a neighbour of either $a$ or $b$ in $\mathcal{M_F}$ (and not a neighbour of both).
  Therefore, we can set $p(u) = g(u) - \{v\} + \{a, b\}$ as the correction function in $\mathcal{M_F}$ without changing the connectivity of
  $u$ to its correction set, and thus wihtout violating the Pauli flow conditions for $u$.
  Then we have $f \notin \Odd(p(u))$, as required.
\end{proof}

\subsection{Destructive vs non-destructive fusion networks}\label{subsec:destructive-networks}

In this section thus far, we treated fusion as a non-destructive measurement that may be applied multiple times on the same node of the resource graph state.
As discussed in \cref{subsec:fusion-network}, this may be achieved by emitting an additional photon for each node belonging to a fusion.
Equivalently, we may consider an \enquote{inflated} resource graph where each node is used in at most one fusion and such that $\lambda(u) = X$ whenever $u$ belongs to a fusion.
This inflated graph is obtained by unfusing a pair of $X$ measured nodes for each fusion, as in the following example:
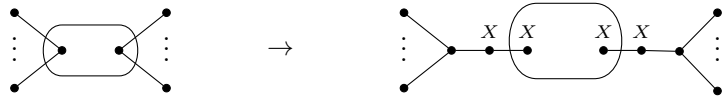
\[
  \input{fnetwork-decomp.tikz}
\]
We now show that the resulting inflated fusion network is equivalent to the original one and that it has flow if and only if the original network has flow.
\begin{proposition}
  The following open graph rewrite preserves the existence of Pauli flow:
  \[
    \input{fnetwork-destructive-rewrite.tikz}
  \]
  where $\lambda(a), \lambda(b), \lambda(f)$ are arbitrary and $\lambda(u_i) = \lambda(v_i) = X$ for $i \in \{0, 1\}$.
  Moreover, no corrections are required on $u_i$ or $v_i$.
\end{proposition}
\begin{proof}
  It is sufficient to show that the following rewrite preserves the existence of Pauli flow:
  \[
    \input{fnetwork-destructive-rewrite2.tikz}
  \]
  with $\lambda(u) = \lambda(v) = X$ and such that both $u$ and $v$ are not in the odd neighbourhood of some correction set.
  The fact that this rewrite preserves the existence of Pauli flow is precisely \cref{prop:Xflow-preserving}.
  And the fact that $u$ does not need to hold corrections is proved in \cref{thm:flow-simplified-graph}.
  It remains to show that $v$ does not need to hold corrections.
  Let $g$ be the Pauli flow on the left-hand side and $p$ the Pauli flow on the right-hand side.
  Suppose $a \in \Odd(g(w))$ for some node $w$.
  After the rewrite, on the RHS, we either have that $a \in \Odd(p(w))$ or $a \notin \Odd(p(w))$.
  In the first case, we must have an even number of neighbours of $v$ which belong to $g(w)$, so $v \notin \Odd(p(w))$ and the result follows.
  In the second case, there is an odd number of neighbours of $v$ which is in $p(w)$, but we can then define a new correction set $p'(w) = p(w) + \set{u}$
  to ensure that $v \notin \Odd(p'(w))$.
  Since $u$ is not a neighbour of $w$, this does not change the connectivity of $w$ to its correction set,
  and thus gives a valid Pauli flow for the RHS\@.
  Therefore, no corrections need to be applied on $v$ and the result follows.
\end{proof}

As a consequence, any fusion network in our sense can be implemented by an inflated fusion network with destructive fusions.
Note that, in this formulation, we are correcting additional errors induced by the $X$ measured nodes, which would not arise in a physical implementation.
It is thus possible that a more refined notion of flow exists for destructive fusion networks.

\section{Optical protocols}\label{sec:protocols}

Currently available optical setups are built from linear optical circuits, photon sources, optical routers, delay lines, and photon detectors.
An \emph{optical protocol} is a sequence of instructions for these setups to perform a given computation.
In this section, we introduce a \emph{dataflow programming language} for optical protocols.
This language has a category-theoretic interpretation in terms of \emph{monoidal streams}~\cite{dilavoreMonoidalStreamsDataflow2022},
which define the unrolling of the time evolution of an experimental setup.
Technically, we instantiate the $\bf{Stream}$ construction~\cite{dilavoreMonoidalStreamsDataflow2022} on
the monoidal category $\bf{C}$ obtained by combining $\bf{LO}$ circuits and $\bf{ZX}$ diagrams using the triangle node~\eqref{eq:triangle}.
The graphical language that we obtain is a formal language which allows for both graphical rewriting and recursive reasoning. 
 
\subsection{Stream processes}

We define a stream process recursively by what it does at time step zero, together with a stream describing what it does at future time steps.
We use letters $X, Y$ to denote infinite sequences $X_0, X_1, X_2, \dots$ of objects in $\bf{C}$, 
and use $X^\plus$ to denote the sequence obtained from $X$ by removing the head and by $\partial X$ the sequence $(I, X_0, X_1, \dots)$
obtained by adding the monoidal unit $I$ to $X$ as the head.

\begin{definition}[Stream]\label{def-stream}
  A stream process $\bf{f}: X \to Y$ in $\bf{Stream}(\bf{C})$ is a stream of types $M$ (called \enquote{memory}),
  a process $f_0 : M_0 \otimes X_0 \to M_1 \otimes Y_0$ in $\bf{C}$ (called \enquote{now}) and a stream $\bf{f}^\plus : X^\plus \to Y^\plus$ (called \enquote{later}).
  \[
    \input{figures/stream-def.tikz}
  \]
  The wire labelled $M_0$ carries the initial state of the memory, $X_0$ and $Y_0$ are the input and output at time-step $0$,
  and $M_1$ is the memory created at time-step $0$ which serves as the initial memory for the stream $\bf{f}^\plus$.
\end{definition}

Diagrams in $\bf{Stream}(\bf{C})$ contain thin wires corresponding to a specific time-step and thick wires for streams.

\begin{remark}
  Note that the definition above requires the memory type to be specified for each time step (including time step zero). 
  These are called \emph{intensional} streams. It is sometimes convenient to quotient the set of streams 
  by \enquote{sliding} along the memory, giving extensional equality. Further quotienting by coinduction gives the correct notion of 
  \emph{observational} equality: two stream processes are observationally equal when we cannot distinguish them from their input/output behaviour. 
  See \cite{dilavoreMonoidalStreamsDataflow2022} for details. The intentional definition allows us to easily initialise stream processes but we will be interested in observational equality 
  when proving results about these processes.
\end{remark}

The simplest stream to initialize is the \emph{constant stream}.
Given a process $f: x \to y$ in $\bf{C}$ we obtain a stream $\bf{f}: X \to Y$ where $X = (x, x, \dots)$ and similarly for $Y$, 
with empty memory $M = I$, with $f_0 = f$ and $\bf{f}^\plus = \bf{f}$.
We denote the constant stream induced by a diagram $f$ simply by thickening its wires.
For example, the following constant stream defines the swap between two constant objects $X = (x, x, \dots)$ and $Y=(y, y, \dots)$ 
in $\bf{Stream}(\bf{C})$:
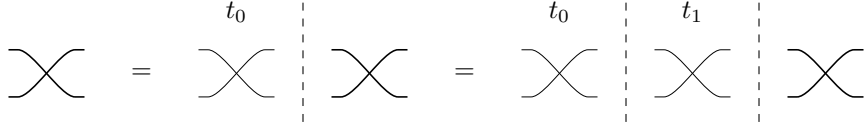
\[
  \input{figures/stream-constant.tikz}
\]
The equation above is read as a recursive definition: the swap stream is the \enquote{swap} now and itself later.
We may define the more general class of \emph{memoryless streams} to be those such that the memory type is the unit of the tensor $M = I$. 
These correspond to sequences $\set{f_t: x_t \to y_t}_{t \in \mathbb{N}}$ in the base category.
For example, we can now define the swap operation between any two objects $X= (X_0, X_1, \dots)$ and $Y = (Y_0, Y_1, \dots)$
as the sequence $\set{\mathtt{swap}_t: X_t \otimes Y_t \to Y_t \otimes X_t}_{t \in \mathbb{N}}$.

In order to link different time steps and model \emph{feedback} of information, we need streams with a memory.
We can obtain these from memoryless streams by taking the feedback $\tt{fbk}_S(\bf{f}): X \to Y$ defined for any stream 
$\bf{f}: \partial S \otimes X \to S \otimes Y$. This corresponds to adding $S$ to the memory of the stream by feeding back its output values to the inputs,
as shown below.
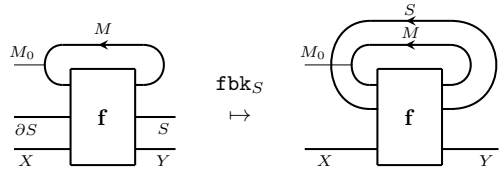
\[
  \scalebox{0.85}{\input{figures/stream-feedback.tikz}}
\]
We can use this to model \emph{delay}.
For example, the delay of length $1$ is the feedback of the swap $\partial X \otimes X \to X \otimes \partial X$:
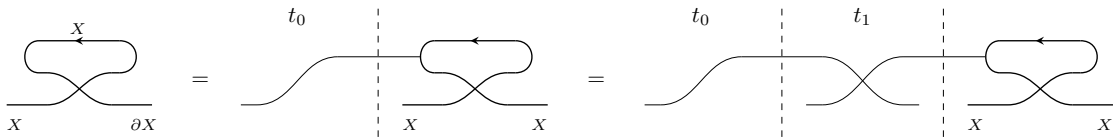
\[
  \scalebox{0.85}{\input{figures/stream-delay.tikz}}
\]
The delay of length $d$ is given by the following composition:
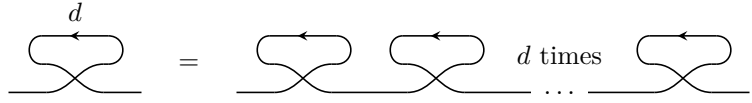
\[
  \input{figures/stream-delay-d.tikz}
\]
Streams represent infinite processes but it is often useful to consider their action for a fixed number of time-steps.
This is done by \emph{unrolling} the stream.

\begin{definition}[Unrolling]\label{def:unrolling}
  Given a stream $\bf{f} : X \to Y$ with memory $M$, the unrolling for $n$ time-steps of $\bf{f}$ is a process in $\bf{C}$ of the form:
  \[
    \mathtt{unroll}_n(\bf{f}) \colon M_0 \otimes X_0 \otimes X_1 \otimes \dots X_{n} \to Y_0 \otimes Y_1 \otimes \dots Y_{n} \otimes M_{n + 1}
  \]
  defined by induction as follows:
  \begin{align*}
    \mathtt{unroll}_0(\mathbf{f}) = \; & \mathtt{swap}_{M_1, X_0} \circ f_0 \\
    \mathtt{unroll}_n(\bf{f}) = \; & (\mathtt{id}_{X_0} \otimes \mathtt{unroll}_{n - 1}(\bf{f}^\plus) ) \circ (\mathtt{unroll}_0(\mathbf{f}) \otimes \mathtt{id}_{Z})
  \end{align*}
  where $Z = X_1 \otimes X_2 \otimes \dots X_{n - 1}$.
\end{definition}

Below are the three first unrollings of a generic stream $\bf{f} : X \to Y$ with memory $M$.

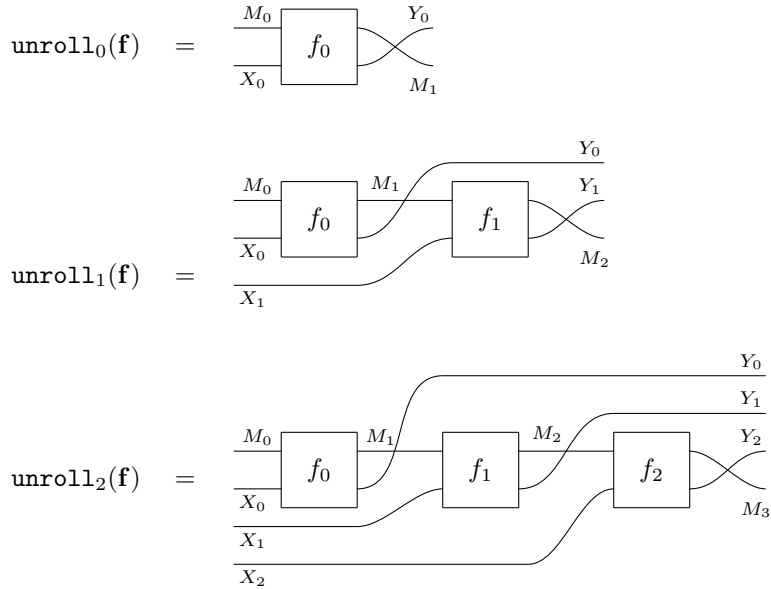
\begin{align*}
  \mathtt{unroll}_0(\bf{f}) \quad&= \quad \input{stream-unroll.tikz}\\
  \mathtt{unroll}_1(\bf{f}) \quad&= \quad \input{stream-unroll1.tikz}\\
  \mathtt{unroll}_2(\bf{f}) \quad&= \quad \input{stream-unroll2.tikz}
\end{align*}

In order to reason diagrammatically about the unrolling it is useful to express the inductive definition above in terms of string diagrams as follows.
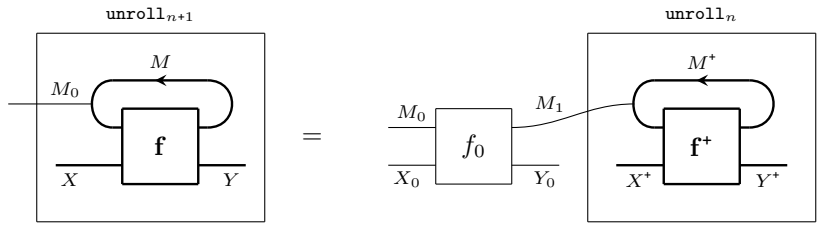
\[
  \input{unrolling.tikz}
\]
Note that this graphical definition avoids the bureaucracy of ordering input and output wires. 
This notation, and the one used in \cref{def-stream}, could be formalised in terms of \enquote{open diagrams}~\cite{romanOpenDiagramsCoend2020}.

\subsection{Streams of linear optics}

The language of stream processes introduced above can be built over any base symmetric monoidal category $\bf{C}$.
Taking linear optical circuits as the base, we are further able to discard a given mode and prepare the empty state.
This gives rise to a useful class of memoryless processes in $\bf{Stream}(\bf{LO})$, called \emph{routers}. 
For example, the binary oscillating router is defined by:
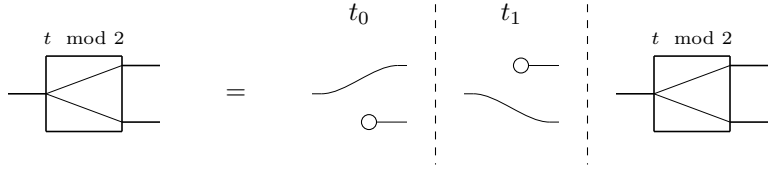
\[
  \input{figures/router-binary.tikz}
\]
And a similar 2 to 1 router is obtained using the discard map:
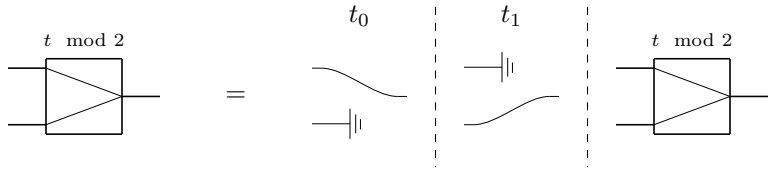
\[
  \input{figures/router-2-1.tikz}
\]
We assume that we have access to routers that can be controlled by a stream of classical variables $x_t$.
We may construct arbitrary routers from the binary router.
For example, the following setup implements an identity if $x_t = 0$ and a swap if $x_t = 1$.
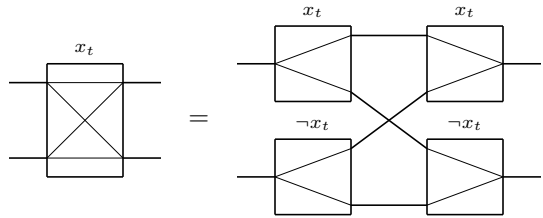
\[
  \input{figures/router-general.tikz}
\]
By composing delays and routers, we can model some important components in photonic computing.
For example, the $d$-ary multiplexer going from time encoding to spatial encoding, and its inverse, the demultiplexer, are defined as follows:
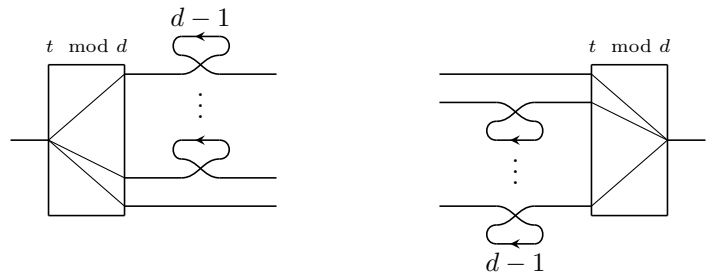
\[
  \input{figures/stream-mulitplexer.tikz} \qquad \qquad \qquad \input{figures/stream-demultiplexer.tikz}
\]
In fact, it is possible to use routers and delays to reorder the time ordering of any inputs arbitrarily:
\begin{lemma}\label{lemma:permutations}
  The following setup implements any permutation $\sigma$ of length $d$ in time encoding:
  \[
    \input{figures/router-permutation.tikz}
  \]
\end{lemma}
\begin{proof}
  Given a permutation $\sigma$, we set $x_t = t + \sigma(t) \mod 2d$ and $y_t = d + t \mod 2d$.
\end{proof}
There are more efficient ways of routing permutations. 
For example, the following setup~\cite{bartolucciSwitchNetworksPhotonic2021}
implements arbitrary delay of size $2^d$ using $d$ binary routers.
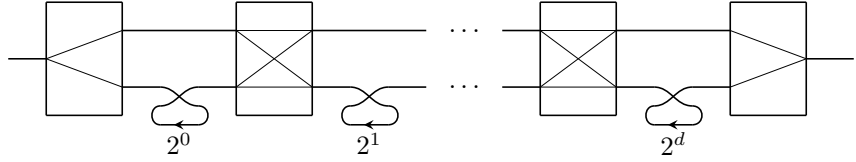
\[
  \input{figures/router-bitstring-permutation.tikz}
\]
where the instructions for routers are obtained from the binary encoding of the delay.

We distinguish between routers and \emph{switches}.
The control parameters of a router for every time step are set before executing the program.
With switches, the routing can be actively controlled by a measurement outcome or a classical variable computed at run-time.
Such actively controlled switches are denoted as routers but with the control parameter drawn inside the box.
For example, the switch with two inputs and two outputs is defined by the constant stream induced by the following controlled process:
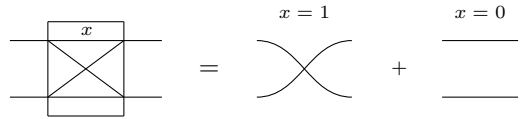
\[
  \input{switch-def.tikz}
\]

We are now ready to introduce the basic optical components required to perform MBQC: a correction module and a measurement module.
%In order to perform MBQC, we need to actively apply Pauli corrections in the $X$, $Y$ or $Z$ bases.
Any sequence of Pauli correction can be implemented using switches with control parameters $x_t, z_t \in \set{0, 1}$, as follows:
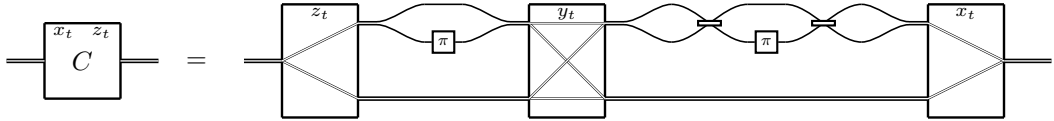
\[
  \input{RUS/CorrectionModule.tikz}
\]
where $y_t = x_t \oplus z_t$.
Then, for the measurement module, we need to be able to measure in the $X$, $Y$, or $Z$ bases at different time-steps, typically fixed before running the experiment.
Given a choice of measurement plane $\lambda_t \in \set{X, Y, Z}$ and angle $\alpha_t \in [0, 2\pi)$ for each time-step $t$,
the measurement module $M_{\lambda, \alpha}$ is given by the following setup:
\[
  \input{RUS/MeasurementModule.tikz}
\]

\subsection{Streams of ZX diagrams}

In \cref{subsubsec:resource-states}, we described two methods for resource state generation,
these have different representations as streams of ZX diagrams.
Photonic resource state generators~\cite{bartolucciCreationEntangledPhotonic2021} typically do not have memory 
and provide constant-size resource states at every time-step.
They can thus be modeled by the constant stream induced by their underlying graph.
For example, emitters of the $4$-star and square resource states are defined by the following streams:
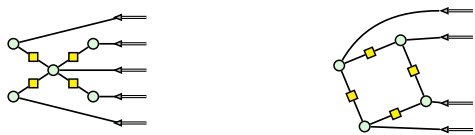
\[
  \scalebox{0.7}{\input{figures/stream-star.tikz}} \qquad \qquad \qquad \scalebox{0.7}{\input{figures/stream-square.tikz}}
\]
where we use the triangle generator to indicate that the qubits are encoded as photons in dual-rail.

Matter-based methods for resource state generation have a particularly simple representation in our language.
Spin-based emitters, such as quantum dots, produce a stream of photons entangled with the atom.
Generally, they may be considered as a \enquote{photonic machine gun}~\cite{lindnerProposalPulsedOnDemand2009}, 
and represented accordingly:
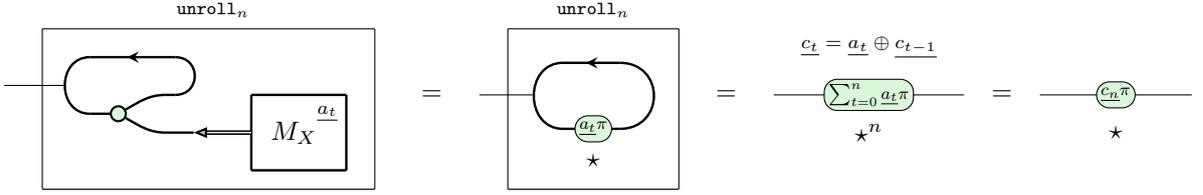
\[
  \input{figures/stream-emitter.tikz}
\]
The atom is the memory of the stream, entangled via a Z spider to the dual-rail states of the emitted photons.
At each time-step $t$ we may perform a single qubit unitary $x_t$ on the atom.
Special cases of interest are GHZ state generators when the white box is the identity,
linear clusters when it is a Hadamard gate and variable GHZ-linear clusters when it can be programmed
arbitrarily, as shown in \cref{fig:fbqc}.

Combining streams of linear optics and ZX diagrams, we can reason about optical protocols used in photonic quantum computing.
For example, the measurement and correction modules defined above with linear optical components, give rise to the expected streams of ZX diagrams.
\begin{lemma}\label{lemma-correct-measure}
  \[
    \input{RUS/MBQCModule.tikz}
  \]
\end{lemma}
\begin{proof}
  Since both the correction and measurement modules are memoryless streams,
  it is sufficient to prove that the equation above holds for any given time step $t$.
  The proof is then obtained by enumerating the different cases for the binary choices of $x_t$ and $y_t$ and the ternary choice of $\lambda_t$.
  In each case, the routers and switches define a specific path and the result follows from the equations of \cref{sec:bgd-dual-rail}.
\end{proof}
\begin{remark}
  Note that routers and switches with more than one output do not exist in $\bf{Stream}(\bf{ZX})$.
  This is because the qubit space $\mathbb{C}^2$ does not allow for the empty state.
\end{remark}

The language of streams allows for different forms of graphical recursive reasoning.
To illustrate this, we now consider two simple properties of the photonic machine gun. 
The first holds for the infinite process and is proved by \emph{coinduction}~\cite{kozenPracticalCoinduction2017}. 
The second holds for any finite unrolling of the stream, and we prove it by \emph{induction} instead.
\begin{lemma}\label{lemma:coinduction}
  \[
    \input{figures/stream-emitter-phase.tikz}
  \]
\end{lemma}
\begin{proof}
  We prove this by coinduction, first using \cref{def-stream}, then by assuming the hypothesis in the future of the stream.
  \begin{align*}
      \input{figures/stream-emitter-proof-0.tikz}\quad
      &\input{figures/stream-emitter-proof.tikz}\\
      &\input{figures/stream-emitter-proof-2.tikz}
  \end{align*}
\end{proof}

\begin{lemma}
  \[
    \input{lemma-measure-emitter.tikz}
  \]
\end{lemma}
\begin{proof}
  The first equality follows from \cref{lemma-correct-measure} and the spider fusion rule as shown below.
  \[
    \scalebox{0.8}{\input{lemma-measure-proof-1.tikz}}
  \]
  We prove the second equality by induction. The statement for $n=0$ is easy to show. Then using \cref{def:unrolling}, we have
  \[
    \scalebox{0.8}{\input{lemma-measure-proof-2.tikz}}
  \]
  Finally, the last equality holds after discarding the measurement outcomes $a_t$. Indeed there are $2^n$ possible measurement outcomes for the list $a_t$, 
  but they together induce one random bitflip $c_n \pi$ after $n$ time-steps. The statement then follows from:
  \[ 
    \input{lemma-measure-proof-3.tikz}
  \]
\end{proof}

\section{Universality in linear optics}\label{sec:universality}

Any claim of universality in linear optics has to deal with the probabilistic nature of linear optical entanglement.
In the case of fusion measurements, the probability of successfully entangling the qubits is only $\frac{1}{2}$.
The language of streams allows us to reason about fusion as an iterated stochastic process. 
We use it to prove correctness of repeat-until-success protocols that can be used to boost the probability of success of any fusion with green failure.
Finally, we show universality of a simple FBQC architecture based on a single emitter.

\subsection{Fusion as a probabilistic process}\label{subsec:fusion-prob}

Let us consider the circuit of a non-destructive fusion measurement with green failure, parametrised by three phases $\theta_1$, $\theta_2$ and $\theta_3$.
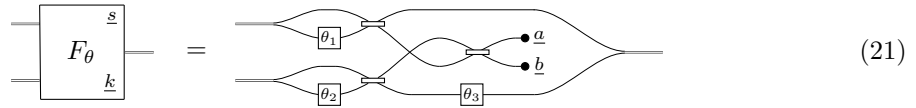
\begin{equation}\label{fusion-green}
  \input{fusion-green.tikz}
\end{equation}
Recall from \cref{prop:type-I}, that the behaviour of Type I fusion measurements on dual-rail encoded qubits can be described by a sum of ZX diagrams.
A similar equation holds for the circuit defined above.
\begin{equation}\label{fusion-green-decomp}
  \input{fusion-prob-1.tikz}
\end{equation}
where $\cvar s = \cvar a \oplus \cvar b$ and $\cvar k = \cvar s \cvar b + \neg \cvar s (1 - \frac{\cvar a + \cvar b}{2})$.
The two diagrams above represent the action of fusion in case of success and failure, respectively.
The \emph{probability of success} for an input state $\Psi$ is obtained by taking the \emph{trace} of the success diagram, 
and discarding the classical output $\cvar{k}$, which corresponds to summing over its possible values.
\begin{equation}
  \Pr(\cvar{s}=1\,\vert\Psi) = \; \input{fusion-prob-2.tikz} \; = \; \input{fusion-prob-2-1.tikz} \; = \; \input{fusion-prob-2-2.tikz}
\end{equation}
Note that the probability will usually depend on the input state $\Psi$.
For example, we may engineer input states for which the fusion \enquote{always succeeds}.
\[
  \input{fusion-prob-fail.tikz}
\]
The calculation above uses dotted wires to represent ZX diagrams in the standard (rather than the CP) interpretation.
Replacing the first input with a green $\theta_1 + \pi$ spider we obtain an example where the fusion \enquote{always fails}.
If the input state is the completely mixed state on two qubits, the probability is found to be $\frac{1}{2}$ by the following derivation.
\[
  \input{fusion-prob-fully-mixed.tikz}
\]

By the purification theorem, a general mixed state $\Psi$ can be expressed in terms of a pure state $\psi$ on a larger space:
\[
  \input{purification.tikz}
\]
The dependence of the probability of success on the input can be avoided if we assume that the inputs are unmeasured qubits in a graph state.
Then the state $\psi$ has the following form:
\[
  \input{fusion-prob-3.tikz}
\]
where $\ket{G}$ is a graph state.
\begin{proposition}\label{prop:green2qhalf}
  When the inputs are unmeasured qubits in a graph state $\ket{G}$, 
  the success probability of any fusion with green failure is $\frac{1}{2}$.
\end{proposition}
\begin{proof}
  \begin{align*}
    \Pr(\cvar{s}=1\,|\Psi)\quad
    &\input{fusion-prob-graph-state.tikz}\\
    &\input{RUS/fusion-prob-4.tikz}
  \end{align*}
\end{proof}

The action of fusion on unmeasured qubits in a graph state may be seen as a non-demolition measurement, 
obtained by composing the fusion operation with Z spiders as above.
It is useful to write this operation as a classical probability distribution over causal maps, as follows.

\begin{proposition}\label{prop:fusion-stochastic}
  Let $F_\theta$ be the optical circuit defined above and $F^\alpha_\theta$ the same optical circuit followed by a 
  measurement in the $\XZm$ plane of angle $\alpha$. Then we have:
  \[
    \input{fusion-stochastic.tikz}
  \]
  \[
    \input{fusion-stochastic2.tikz}
  \]
\end{proposition}
\begin{proof}
  This follows from \cref{fusion-green-decomp} and $\interp{\star}_{CP} = \frac{1}{2}$.
\end{proof}

\subsection{Repeat-until-success}\label{subsec:RUS}

Boosting the probability of success of fusion measurements is an essential requirement for scaling FBQC.
Linear optical approaches include the use of ancillary photons~\cite{griceArbitrarilyCompleteBellstate2011, ewertEfficientBellMeasurement2014} and switch networks~\cite{bartolucciSwitchNetworksPhotonic2021}.
In matter-based approaches, fusion can be boosted by assuming photons are input from the atom in an entangled state.
We are particularly interested in the \emph{repeat-until-success} protocol of~\cite{limRepeatUntilSuccessLinearOptics2005,degliniastySpinOpticalQuantumComputing2024},
that allows for near-deterministic CZ gates assuming photons are generated as a GHZ resource state.
They use a version of the $H$-fusion measurement given in \cref{ex:Y-fusion} in order to obtain CZ entanglement.
We now generalise the protocol to arbitrary fusions with green failure and give a formal proof of its correctness.

Starting from the circuit in (\ref{fusion-green}), we define the destructive fusion module, implementing arbitrary fusions with green failure, as the following stream:
\[
  \input{RUS/FusionModule.tikz}
\]
This module produces two streams of classical outputs: the success values $\cvar{s_t}$ and the errors $\cvar{k_t}, \cvar{j_t}$.
We are now ready to state our results about repeat-until-success protocols. The proofs are given in \cref{app:RUS}.

\begin{definition}[RUS protocol]
  The repeat-until-success fusion protocol is defined by the following setup:
  \[
    \input{repeat-until-success.tikz}
  \]
  Note that the control parameter of the switch takes the value $s_{t-1}$ of the previous success outcome.
\end{definition}

\begin{restatable}{theorem}{RUS}\label{thm:RUS}
  Any fusion with green failure can be boosted with a repeat-until-success protocol. More precisely, the following holds for $n \geq 1$:
  \[
  \scalebox{.8}{\input{repeat-boosted.tikz}}
  \]
  where $T$ is the time of the first successful fusion (if it exists) and:
  \[
    c_{t} = c_{t - 1} \oplus (\neg s_t) k_t \oplus s_t a_t \quad \qquad d_{t} = d_{t - 1} \oplus (\neg s_t) (\neg k_t) \oplus s_t b_t
  \]
  with $s_{\minu 1} = 0$, $c_{\minu 1} = d_{\minu 1} = 1$.
\end{restatable}

As a consequence, the probability of success of a repeat-until-success fusion protocol after $n + 1$ time-steps is $1 - \frac{1}{2^{n + 1}}$.
Note that, even though the first two terms are success cases, we cannot unify them as one diagram because the first has $4$ output variables
while the second only has $3$. For the particular cases of $X$ and $Y$ fusions, we can simplify the equations to obtain the following corollaries.

\begin{restatable}[X fusion RUS]{corollary}{XRUS}
  For $n \geq 1$ we have:
  \[
    \input{repeat-X.tikz}
  \]
  where $x_T = k_T \oplus j_T$ and $z_t = s_t(c_t \oplus d_t) \oplus (\neg s_t)c_t$.
\end{restatable}

\begin{restatable}[Y fusion RUS]{corollary}{YRUS}
  For $n \geq 1$ we have:
  \[
    \input{repeat-Y.tikz}
  \]
  where $z_t = (k_T \oplus j_T) \oplus c_t$ and $y_t = (k_T \oplus j_T) \oplus d_t$ if $T < t$ and $y_t = (k_T \oplus j_T) \oplus \neg c_t$ if $T=t$.
\end{restatable}

\subsection{A universal architecture}\label{subsec:architecture}

We now propose a simple architecture based on a single quantum emitter, linear optics, active switching and classical feedforward.
We then show that this architecture can be used to implement arbitrary MBQC patterns and thus achieve universal quantum computation.
This proof has the following assumptions:
\begin{enumerate}
    \item resource state generation is deterministic,
    \item photons are indistinguishable,
    \item all components have perfect efficiency.
\end{enumerate}
Relaxing any of these assumptions defines a landscape for optimisation depending on the error model.
We leave these considerations for future work and focus on showing universality in this idealised setting.

The architecture studied here is built from a spin-based emitter, the delay, measurement and correction modules defined in \cref{sec:protocols}, 
and the repeat-until-success fusion module described in \cref{subsec:RUS}. We define it as the following diagram.
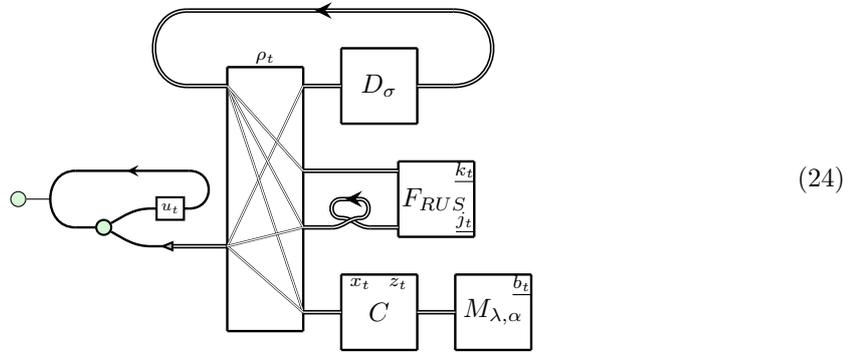
\begin{equation}\label{architecture-1}
  \input{figures/architecture-1.tikz}
\end{equation}
Let us consider what happens when we unroll the stream defined by the diagram, step by step, for some particular choice of the parameters.
By recursively applying \cref{def:unrolling}, we produce a diagram in $\bf{C}$ which is structured as follows:
\begin{itemize}
  \item the top part of the diagram consists of a variable GHZ-linear cluster produced by the emitter,
  \item the middle part is obtained by unrolling the delay module and equates to a permutation of the qubits,
  \item the bottom part is a sequence of fusions and single qubit measurements of arbitrary types.
\end{itemize}
As an example, the following is the success term of a possible unrolling of the architecture:
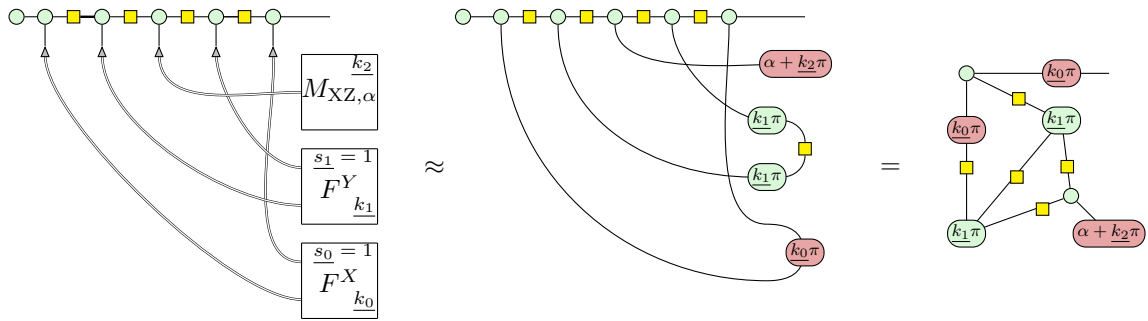
\[
  \input{figures/architecture-unrolling.tikz}
\]
By unrolling time steps and pushing triangles from left to right, we turn the diagram progressively 
from an optical circuit into a ZX diagram and a corresponding MBQC pattern.
This efficient rewriting process produces a ZX diagram capturing the quantum computation that has been executed. 
We can thus efficiently \emph{verify} that the protocol in (\ref{architecture-1}) performs a given quantum computation.

To prove universality of the architecture, we show that it can be used to implement any XY-fusion network where the resource graph is a line.
\begin{definition}
  A linear XY-fusion pattern is an XY-fusion pattern where the $n$-qubit register is totally ordered $V = \set{1, \dots, n}$ and the entangling commands are restricted to be of the form $E_{i, i \plus 1}$.
  A linear XY-fusion network is an XY-fusion network where the resource graph is a disjoint union of lines.
\end{definition}
\begin{proposition}\label{prop:hamiltonian-path}
  For any labeled open graph $\mathcal{G}$ with flow, there is a linear XY-fusion network $\mathcal{F}$ with flow with the same target linear map $T(\mathcal{G}) = T(\mathcal{F})$.
\end{proposition}
\begin{proof}
  We use exclusively the $Y$-fusion measurement which adds a hadamard edge between nodes.
  Given any labeled open graph $\mathcal{G} = (G, I, O, \lambda, \alpha)$ with flow, we may extend it to an equivalent labeled open graph $(G', I, O, \lambda, \alpha)$
  such that $G'$ has a Hamiltonian path by finding a Hamiltonian completion of $G$. The additional edges in the completion are constructed by introducing nodes measured in the Z-basis, an operation which preserves the existence of Pauli flow \cite{mcelvanneyCompleteFlowpreservingRewrite2023}.
  Then we may construct a linear XY fusion network $(L, F)$ where $L \sub G'$ is the Hamiltonian path and $F$ is the set of all remaining edges. It is easy to check 
  that this fusion network has the same target linear map as $G$.
\end{proof}
\begin{proposition}\label{prop:pattern-implementability}
  The protocol in (\ref{architecture-1}) has settings $\lambda, \alpha, \sigma, \rho, u$ that implement any runnable linear XY-fusion pattern, with probability arbitrarily close to $1$.
\end{proposition}
\begin{proof}
  Fix a linear XY-fusion pattern. Let $f_i$ be the number of fusions applied to qubit $i$. The total number of fusion operations is $f = \frac{1}{2} \sum_{i = 1}^n f_i$.
  For each qubit $i$, we emit $kf_i + 1$ photons entangled as a GHZ state with the atom, where $k$ is a positive integer. Between rounds we either apply a hadamard gate on the atom
  if the command $E_{i, i \plus 1}$ is present, or else we emit an additional photon to be measured in the $X$ basis.
  By setting the parameters $\rho$ and $\sigma$ we may route these photons arbitrarily in either a RUS fusion or a single-qubit measurement, following \cref{lemma:permutations}.
  We use $k$ photons for each node in a RUS fusion operation, giving us a probability of success of $(1 - \frac{1}{2^k})$ for each of the $f$ fusions.
  If the RUS fusion fails after $k$ rounds we restart the whole computation.
  Finally, we can apply any sequence of single qubit measurements and corrections on the remaining $n$ photons, following \cref{lemma-correct-measure}.
  In order to achieve a total success probability $\epsilon$ close to $1$ we just have to set $k$ an odd integer such that $(1 - \frac{1}{2^k})^f > 1 - \epsilon$.
\end{proof}

Note that the RUS $X$ and $Y$ fusions defined above may induce additional $Z$ errors on the target qubits. 
Even for $Y$-fusion, it is sufficient to set $n$ to be even (i.e. repeat an odd number of times) to ensure that the error is Pauli.
Thus, in order to correct these errors in an XY-fusion pattern, we must add $Z$ corrections on the target qubit. 
This is always possible with the factorisation given in \cref{thm:flow-pattern}, since fusion nodes precede their target qubits in the partial order.

\begin{theorem}\label{thm:universality}
    The protocol in (\ref{architecture-1}) has settings $\lambda, \alpha, \sigma, \rho, u$ that implement any given qubit unitary, with probability arbitrarily close to $1$.
\end{theorem}
\begin{proof} 
  Given any qubit unitary, we may represent it as a labeled open graph $\mathcal{G}$ with flow. By \cref{prop:hamiltonian-path}, there is a linear XY fusion pattern $\mathcal{F}$
  with flow and the same target linear map. This gives rise to a runnable linear XY fusion pattern and the result follows by \cref{prop:pattern-implementability}.
\end{proof}

The architecture above can be extended in several different ways. At the level of resource state generation, we may use multiple linear-GHZ emitters instead of one.
This would allow us to parallelise the computation, reducing photon delays and atom coherence time, although it comes at the cost of flower fidelities as photons emitted from different sources have higher distinguishability.
The architecture could also be extended with a module implementing local Clifford unitaries on the emitted photons. 
This would allow us to use any graph LC-equivalent to a linear cluster as resource state, as proposed in~\cite{zilkCompilerUniversalPhotonic2022}.
The architecture above assumes that corrections need not be applied on fusion nodes, which is justified for XY-fusion networks by \cref{thm:flow-pattern}. 
In order to implement more general fusion networks, we may use an additional $2$ to $1$ router to allow corrections after Type I fusion, as follows.
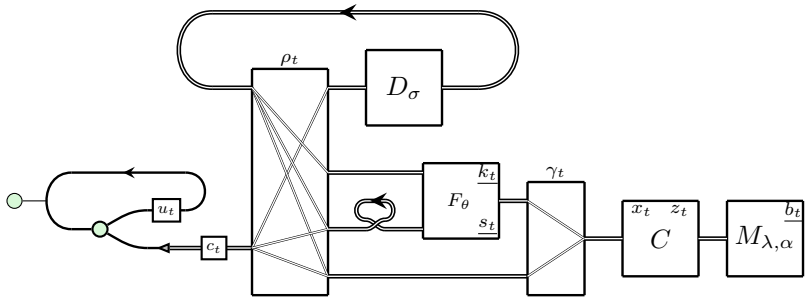
\[
  \input{figures/architecture-3.tikz}
\]

\section{Conclusion}\label{sec:conclusion}

A graphical framework for photonic quantum computing was presented aiming at bringing together linear optics, MBQC, and dataflow programming.

We used a combination of ZX diagrams and linear optical circuits to analyse the action of fusion measurements and their induced errors.
We characterised all fusion measurements whose failure outcome is a projector in the $\XYm$ plane (green failure)
and whose success outcome induces correctable Pauli errors on their input qubits; see \cref{thm-characterisation}.
Building on this, we gave a general definition of fusion networks and developed a notion of XY-flow for 
fusion networks using exclusively X and Y fusions, \cref{def:flow}, allowing us to correct undesired measurement outcomes.
The resulting patterns can be factorized such that all fusions appear before single-qubit measurements; see \cref{thm:flow-pattern}.
Moreover, we showed that any decomposition of an open graph with Pauli flow as an XY-fusion network
is guaranteed to have XY-flow; see \cref{thm:flow-simplified-graph}.
An interesting avenue for future work is the definition of flow for fusion networks with more general types of fusion, 
such as fusions in arbitrary planes and angles and $n$-ary versions of these measurements.

We presented a dataflow language describing streams of linear optical circuits and ZX diagrams and 
enabling the analysis of optical protocols involving measurements, routers, delays, switches, time-delayed emitters, and classical feedforward.
This graphical calculus allows us to inductively prove the correctness of new repeat-until-success protocols that
boost the probability of success of arbitrary fusion measurements with green failure; see \cref{thm:RUS}.
We also used our framework to give a constructive proof of universality for a simple optical architecture based on a single quantum emitter; see \cref{thm:universality}.
We believe that the potential applications of this calculus extend well beyond fusion-based quantum computing, 
towards the analysis of quantum communication protocols and distributed algorithms.

As a further development of the proposed framework, each step in the outlined compilation process can be optimised.
From open graphs to fusion networks, the number of fusions required could be minimised under different constraints on the allowed fusions and resource graphs.
From fusion networks to optical protocols, the time difference between photon emission and measurement, as well as the number of optical components that the photon needs to traverse, 
should be minimised to reduce the probability of photon loss.

\section*{Acknowledgements}

We would like to thank Will Simmons for reviewing an earlier version of this manuscript.
We greatly benefited from discussions with Ross Duncan, Pierre-Emmanuel Emeriau, Paul Hilaire, Dan Mills, Razin A.\@ Shaikh, and Richie Yeung.

\bibliographystyle{quantum}
\bibliography{preamble/references}

\onecolumn
\appendix
\allowdisplaybreaks

\section{QPath calculus}\label{sec:qpath}

We review the axioms of the QPath calculus~\cite{defeliceQuantumLinearOptics2023}, which are used in \cref{sec:mixedfusion}.
Diagrams in \textbf{QPath} are generated by:
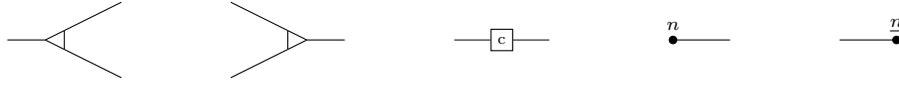
\[
  \begin{tikzpicture}
	\begin{pgfonlayer}{nodelayer}
		\node [style=none] (0) at (-0.5, 0) {};
		\node [style=none] (1) at (0, 0.25) {};
		\node [style=none] (2) at (0, -0.25) {};
		\node [style=none] (3) at (-1.5, 0) {};
		\node [style=none] (4) at (1.5, 1) {};
		\node [style=none] (5) at (1.5, -1) {};
	\end{pgfonlayer}
	\begin{pgfonlayer}{edgelayer}
		\draw (0.center) to (1.center);
		\draw (1.center) to (2.center);
		\draw (2.center) to (0.center);
		\draw (3.center) to (0.center);
		\draw (1.center) to (4.center);
		\draw (2.center) to (5.center);
	\end{pgfonlayer}
\end{tikzpicture}

  \qquad \qquad
  \begin{tikzpicture}
	\begin{pgfonlayer}{nodelayer}
		\node [style=none] (6) at (0.5, 0) {};
		\node [style=none] (7) at (0, 0.25) {};
		\node [style=none] (8) at (0, -0.25) {};
		\node [style=none] (9) at (-1.5, 1) {};
		\node [style=none] (10) at (-1.5, -1) {};
		\node [style=none] (11) at (1.5, 0) {};
	\end{pgfonlayer}
	\begin{pgfonlayer}{edgelayer}
		\draw (6.center) to (7.center);
		\draw (7.center) to (8.center);
		\draw (8.center) to (6.center);
		\draw (9.center) to (7.center);
		\draw (6.center) to (11.center);
		\draw (8.center) to (10.center);
	\end{pgfonlayer}
\end{tikzpicture}

  \qquad \qquad
  \begin{tikzpicture}
	\begin{pgfonlayer}{nodelayer}
		\node [style=none] (0) at (-1.25, 0) {};
		\node [style=none] (1) at (1.25, 0) {};
		\node [style=phase] (2) at (0, 0) {c};
	\end{pgfonlayer}
	\begin{pgfonlayer}{edgelayer}
		\draw [style=LOWire] (1.center) to (2);
		\draw [style=LOWire] (2) to (0.center);
	\end{pgfonlayer}
\end{tikzpicture}

  \qquad \qquad
  \begin{tikzpicture}
	\begin{pgfonlayer}{nodelayer}
		\node [style=black node] (0) at (-0.75, 0) {};
		\node [style=none] (1) at (0.75, 0) {};
		\node [style=sLabel] (2) at (-0.75, 0.375) {$n$};
	\end{pgfonlayer}
	\begin{pgfonlayer}{edgelayer}
		\draw (1.center) to (0);
	\end{pgfonlayer}
\end{tikzpicture}

  \qquad \qquad
  \begin{tikzpicture}
	\begin{pgfonlayer}{nodelayer}
		\node [style=black node] (0) at (0.75, 0) {};
		\node [style=none] (1) at (-0.75, 0) {};
		\node [style=sLabel] (2) at (0.75, 0.375) {$\cvar n$};
	\end{pgfonlayer}
	\begin{pgfonlayer}{edgelayer}
		\draw (1.center) to (0);
	\end{pgfonlayer}
\end{tikzpicture}

\]
for all $n \in \N$ and $c \in \C$.
Comparing with the graphical language \textbf{LO} defined in \cref{sec:LO}, the difference is that instead of the beamsplitter generator of \textbf{LO}, \textbf{QPath} has the generator
\begin{equation}
    \scalebox{0.7}{\begin{tikzpicture}
	\begin{pgfonlayer}{nodelayer}
		\node [style=none] (0) at (-0.5, 0) {};
		\node [style=none] (1) at (0, 0.25) {};
		\node [style=none] (2) at (0, -0.25) {};
		\node [style=none] (3) at (-1.5, 0) {};
		\node [style=none] (4) at (1.5, 1) {};
		\node [style=none] (5) at (1.5, -1) {};
	\end{pgfonlayer}
	\begin{pgfonlayer}{edgelayer}
		\draw (0.center) to (1.center);
		\draw (1.center) to (2.center);
		\draw (2.center) to (0.center);
		\draw (3.center) to (0.center);
		\draw (1.center) to (4.center);
		\draw (2.center) to (5.center);
	\end{pgfonlayer}
\end{tikzpicture}
} \qquad \overset{\interp{\cdot}}{\longmapsto} \qquad W\,:\, \ket{n} \quad \mapsto \quad \sum_{k=0}^n \binom{n}{k}^{\frac{1}{2}} \ket{k}\ket{n - k}
\end{equation}
along with its horizontal reflection which has the adjoint interpretation.

This enables a more descriptive diagram for the beam splitter,
\begin{equation}
    \input{RUS/beamsplitter-qpath.tikz}
\end{equation}

The \textbf{QPath} calculus admits the following graphical rewrite rules.
Additionally, all rules hold under transposition of the linear maps, which in \textbf{QPath} is represented by horizontal reflection of the diagrams.
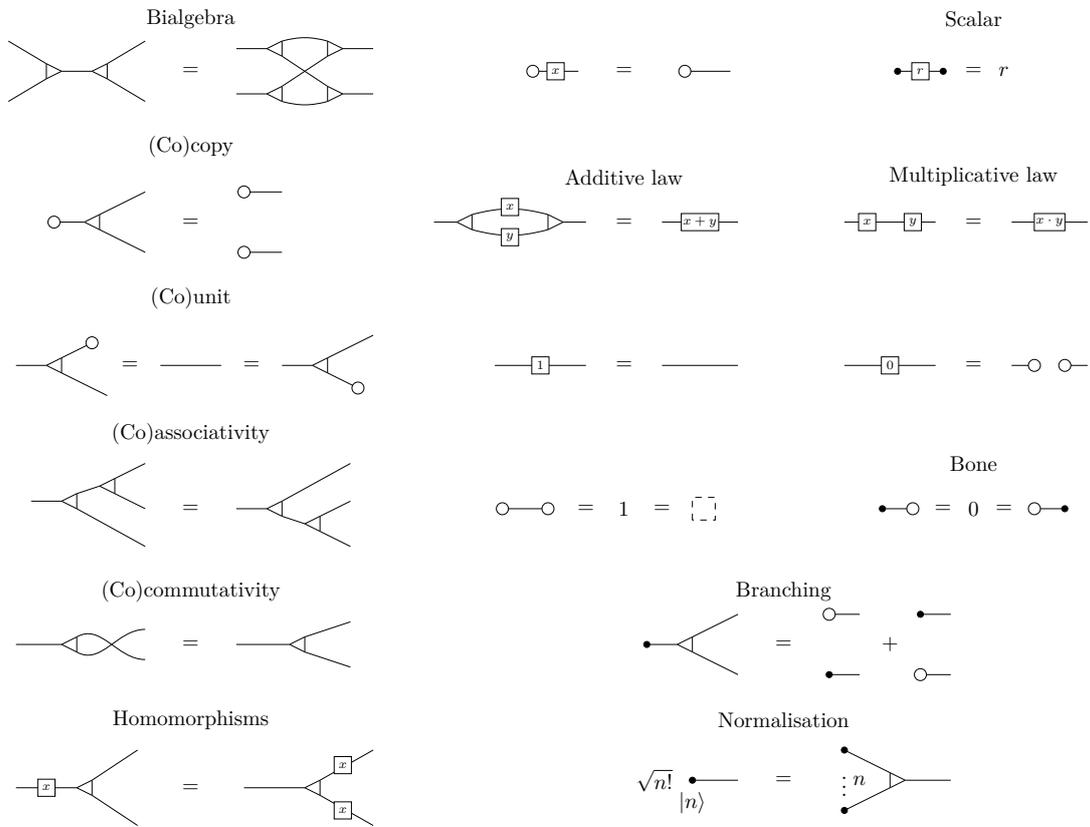
\begin{equation*}
    \scalebox{0.8}{\begin{tikzpicture}
	\begin{pgfonlayer}{nodelayer}
		\node [style=none] (215) at (-11.5, -6.5) {};
		\node [style=none] (12) at (-11.25, -1.5) {};
		\node [style=none] (13) at (-11.75, -1.25) {};
		\node [style=none] (14) at (-11.75, -1.75) {};
		\node [style=none] (15) at (-13, -0.5) {};
		\node [style=none] (16) at (-13, -2.5) {};
		\node [style=none] (17) at (-10.25, -1.5) {};
		\node [style=none] (18) at (-10.25, -1.5) {};
		\node [style=none] (19) at (-9.75, -1.25) {};
		\node [style=none] (20) at (-9.75, -1.75) {};
		\node [style=none] (21) at (-8.5, -0.5) {};
		\node [style=none] (22) at (-8.5, -2.5) {};
		\node [style=none] (23) at (-4.5, -0.75) {};
		\node [style=none] (24) at (-4, -0.5) {};
		\node [style=none] (25) at (-4, -1) {};
		\node [style=none] (26) at (-5.5, -0.75) {};
		\node [style=none] (32) at (-2, -0.75) {};
		\node [style=none] (33) at (-2.5, -0.5) {};
		\node [style=none] (34) at (-2.5, -1) {};
		\node [style=none] (35) at (-4.5, -2.25) {};
		\node [style=none] (36) at (-4, -2) {};
		\node [style=none] (37) at (-4, -2.5) {};
		\node [style=none] (38) at (-5.5, -2.25) {};
		\node [style=none] (39) at (-2, -2.25) {};
		\node [style=none] (40) at (-2.5, -2) {};
		\node [style=none] (41) at (-2.5, -2.5) {};
		\node [style=none] (42) at (-1, -0.75) {};
		\node [style=none] (43) at (-1, -2.25) {};
		\node [style=none] (44) at (-7, -1.5) {$=$};
		\node [style=none] (57) at (1.75, -6.5) {};
		\node [style=none] (58) at (2.25, -6.25) {};
		\node [style=none] (59) at (2.25, -6.75) {};
		\node [style=none] (60) at (1, -6.5) {};
		\node [style=none] (63) at (5.25, -6.5) {};
		\node [style=none] (64) at (4.75, -6.25) {};
		\node [style=none] (65) at (4.75, -6.75) {};
		\node [style=none] (68) at (6, -6.5) {};
		\node [style=none] (70) at (7.25, -6.5) {$=$};
		\node [style=none] (110) at (8.5, -6.5) {};
		\node [style=none] (111) at (11, -6.5) {};
		\node [style=none] (131) at (18.75, -6.5) {$=$};
		\node [style=none] (133) at (20, -6.5) {};
		\node [style=none] (134) at (22.5, -6.5) {};
		\node [style=none] (136) at (14.5, -6.5) {};
		\node [style=none] (137) at (17.5, -6.5) {};
		\node [style=none] (138) at (-12.75, -25.25) {};
		\node [style=none] (139) at (-10.75, -25.25) {};
		\node [style=none] (140) at (-10.75, -25.25) {};
		\node [style=none] (141) at (-10.25, -25) {};
		\node [style=none] (142) at (-10.25, -25.5) {};
		\node [style=none] (143) at (-8.75, -24) {};
		\node [style=none] (144) at (-8.75, -26.5) {};
		\node [style=none] (146) at (-7, -25.25) {$=$};
		\node [style=none] (147) at (-5.25, -25.25) {};
		\node [style=none] (148) at (-3.25, -25.25) {};
		\node [style=none] (149) at (-3.25, -25.25) {};
		\node [style=none] (150) at (-2.75, -25) {};
		\node [style=none] (151) at (-2.75, -25.5) {};
		\node [style=none] (152) at (-1, -24) {};
		\node [style=none] (153) at (-1, -26.5) {};
		\node [style=none] (175) at (-12.75, -11.25) {};
		\node [style=none] (176) at (-11.75, -11.25) {};
		\node [style=none] (177) at (-11.75, -11.25) {};
		\node [style=none] (178) at (-11.25, -11) {};
		\node [style=none] (179) at (-11.25, -11.5) {};
		\node [style=none] (180) at (-10.25, -10.5) {};
		\node [style=none] (181) at (-9.75, -12.25) {};
		\node [style=none] (182) at (-9, -11.25) {$=$};
		\node [style=none] (190) at (-5, -11.25) {$=$};
		\node [style=wn] (198) at (-10.25, -10.5) {};
		\node [style=wn] (199) at (-11.5, -6.5) {};
		\node [style=wn] (200) at (-5.25, -5.5) {};
		\node [style=none] (202) at (-8, -11.25) {};
		\node [style=none] (203) at (-6, -11.25) {};
		\node [style=none] (204) at (-4, -11.25) {};
		\node [style=none] (205) at (-3, -11.25) {};
		\node [style=none] (206) at (-3, -11.25) {};
		\node [style=none] (207) at (-2.5, -11) {};
		\node [style=none] (208) at (-2.5, -11.5) {};
		\node [style=none] (209) at (-1, -10.25) {};
		\node [style=none] (210) at (-1.5, -12) {};
		\node [style=wn] (211) at (-1.5, -12) {};
		\node [style=none] (214) at (-9, -7.25) {};
		\node [style=none] (216) at (-10.5, -6.5) {};
		\node [style=none] (217) at (-10.5, -6.5) {};
		\node [style=none] (218) at (-10, -6.25) {};
		\node [style=none] (219) at (-10, -6.75) {};
		\node [style=none] (220) at (-8.5, -5.5) {};
		\node [style=none] (221) at (-8.5, -7.5) {};
		\node [style=wn] (222) at (-5.25, -7.5) {};
		\node [style=none] (223) at (-7, -6.5) {$=$};
		\node [style=none] (224) at (-4, -5.5) {};
		\node [style=none] (225) at (-4, -7.5) {};
		\node [style=none] (239) at (4.25, -1.5) {};
		\node [style=none] (240) at (5.75, -1.5) {};
		\node [style=wn] (241) at (4.25, -1.5) {};
		\node [style=none] (242) at (7.25, -1.5) {$=$};
		\node [style=none] (243) at (9.25, -1.5) {};
		\node [style=none] (244) at (10.75, -1.5) {};
		\node [style=wn] (245) at (9.25, -1.5) {};
		\node [style=none] (265) at (-7, 0.25) {Bialgebra};
		\node [style=none] (266) at (-7, -9) {(Co)unit};
		\node [style=none] (267) at (7.25, -5) {Additive law};
		\node [style=none] (268) at (18.75, -5) {Multiplicative law};
		\node [style=none] (269) at (-7, -23) {Homomorphisms};
		\node [style=none] (270) at (-12.25, -15.75) {};
		\node [style=none] (271) at (-11.25, -15.75) {};
		\node [style=none] (272) at (-11.25, -15.75) {};
		\node [style=none] (273) at (-10.75, -15.5) {};
		\node [style=none] (274) at (-10.75, -16) {};
		\node [style=none] (275) at (-10, -15.25) {};
		\node [style=none] (276) at (-8.5, -17.25) {};
		\node [style=none] (278) at (-10, -15.25) {};
		\node [style=none] (279) at (-10, -15.25) {};
		\node [style=none] (280) at (-9.5, -15) {};
		\node [style=none] (281) at (-9.5, -15.5) {};
		\node [style=none] (282) at (-8.5, -14.5) {};
		\node [style=none] (283) at (-8.5, -16) {};
		\node [style=none] (284) at (-5.5, -16) {};
		\node [style=none] (285) at (-4.5, -16) {};
		\node [style=none] (286) at (-4.5, -16) {};
		\node [style=none] (287) at (-4, -15.75) {};
		\node [style=none] (288) at (-4, -16.25) {};
		\node [style=none] (289) at (-1.75, -14.5) {};
		\node [style=none] (290) at (-3.25, -16.5) {};
		\node [style=none] (292) at (-3.25, -16.5) {};
		\node [style=none] (293) at (-2.75, -16.25) {};
		\node [style=none] (294) at (-2.75, -16.75) {};
		\node [style=none] (295) at (-1.75, -15.75) {};
		\node [style=none] (296) at (-1.75, -17.25) {};
		\node [style=none] (297) at (-7, -16) {$=$};
		\node [style=none] (323) at (-12.75, -20.5) {};
		\node [style=none] (324) at (-11.25, -20.5) {};
		\node [style=none] (325) at (-11.25, -20.5) {};
		\node [style=none] (326) at (-10.75, -20.25) {};
		\node [style=none] (327) at (-10.75, -20.75) {};
		\node [style=none] (328) at (-8.5, -21) {};
		\node [style=none] (329) at (-8.5, -20) {};
		\node [style=none] (330) at (-5.5, -20.5) {};
		\node [style=none] (331) at (-3.75, -20.5) {};
		\node [style=none] (332) at (-3.75, -20.5) {};
		\node [style=none] (333) at (-3.25, -20.25) {};
		\node [style=none] (334) at (-3.25, -20.75) {};
		\node [style=none] (335) at (-1.75, -19.75) {};
		\node [style=none] (336) at (-1.75, -21.25) {};
		\node [style=none] (337) at (-7, -20.5) {$=$};
		\node [style=none] (377) at (-7, -4) {(Co)copy};
		\node [style=none] (378) at (-7, -13.5) {(Co)associativity};
		\node [style=none] (379) at (-7, -18.75) {(Co)commutativity};
		\node [style=none] (380) at (3.25, -16) {};
		\node [style=none] (381) at (4.75, -16) {};
		\node [style=wn] (382) at (3.25, -16) {};
		\node [style=wn] (383) at (4.75, -16) {};
		\node [style=none] (384) at (6, -16) {$=$};
		\node [style=none] (414) at (14.5, -11.25) {};
		\node [style=none] (415) at (17.5, -11.25) {};
		\node [style=none] (417) at (18.75, -11.25) {=};
		\node [style=none] (418) at (20, -11.25) {};
		\node [style=none] (419) at (20.75, -11.25) {};
		\node [style=wn] (420) at (20.75, -11.25) {};
		\node [style=none] (421) at (21.75, -11.25) {};
		\node [style=none] (422) at (22.5, -11.25) {};
		\node [style=wn] (423) at (21.75, -11.25) {};
		\node [style=none] (424) at (3, -11.25) {};
		\node [style=none] (425) at (6, -11.25) {};
		\node [style=none] (427) at (8.5, -11.25) {};
		\node [style=none] (428) at (11, -11.25) {};
		\node [style=none] (429) at (7.25, -11.25) {$=$};
		\node [style=none] (430) at (9, -20.5) {};
		\node [style=none] (431) at (9.5, -20.25) {};
		\node [style=none] (432) at (9.5, -20.75) {};
		\node [style=none] (433) at (11, -19.5) {};
		\node [style=none] (434) at (11, -21.5) {};
		\node [style=none] (435) at (8, -20.5) {};
		\node [style=black node] (436) at (8, -20.5) {};
		\node [style=none] (437) at (15, -19.5) {};
		\node [style=none] (438) at (14, -19.5) {};
		\node [style=none] (439) at (15, -21.5) {};
		\node [style=none] (440) at (14, -21.5) {};
		\node [style=black node] (441) at (14, -21.5) {};
		\node [style=none] (442) at (12.5, -20.5) {$=$};
		\node [style=wn] (443) at (14, -19.5) {};
		\node [style=none] (444) at (18, -21.5) {};
		\node [style=none] (445) at (17, -21.5) {};
		\node [style=none] (446) at (18, -19.5) {};
		\node [style=none] (447) at (17, -19.5) {};
		\node [style=black node] (448) at (17, -19.5) {};
		\node [style=wn] (449) at (17, -21.5) {};
		\node [style=none] (450) at (16, -20.5) {$+$};
		\node [style=none] (473) at (12.5, -18.75) {Branching};
		\node [style=none] (474) at (16.75, -16) {};
		\node [style=none] (475) at (15.75, -16) {};
		\node [style=black node] (476) at (15.75, -16) {};
		\node [style=wn] (477) at (16.75, -16) {};
		\node [style=none] (478) at (17.75, -16) {$=$};
		\node [style=none] (479) at (18.75, -16) {$0$};
		\node [style=none] (480) at (17.75, -1.5) {};
		\node [style=none] (481) at (16.25, -1.5) {};
		\node [style=black node] (482) at (16.25, -1.5) {};
		\node [style=black node] (483) at (17.75, -1.5) {};
		\node [style=none] (484) at (18.75, -1.5) {$=$};
		\node [style=none] (485) at (19.75, -1.5) {$r$};
		\node [style=none] (488) at (18.75, 0.25) {Scalar};
		\node [style=none] (489) at (21.75, -16) {};
		\node [style=none] (490) at (20.75, -16) {};
		\node [style=black node] (491) at (21.75, -16) {};
		\node [style=wn] (492) at (20.75, -16) {};
		\node [style=none] (493) at (19.75, -16) {$=$};
		\node [style=none] (494) at (12.5, -23) {Normalisation};
		\node [style=none] (495) at (18.75, -14.5) {Bone};
		\node [style=none] (505) at (14.5, -25) {$\vdots$};
		\node [style=none] (506) at (15, -25) {$n$};
		\node [style=none] (510) at (16.5, -25) {};
		\node [style=none] (511) at (16, -24.75) {};
		\node [style=none] (512) at (16, -25.25) {};
		\node [style=none] (513) at (18, -25) {};
		\node [style=black node] (514) at (14.5, -24) {};
		\node [style=none] (515) at (14.5, -24) {};
		\node [style=none] (516) at (14.5, -26) {};
		\node [style=black node] (518) at (14.5, -26) {};
		\node [style=none] (519) at (12.5, -25) {=};
		\node [style=black node] (522) at (9.5, -25) {};
		\node [style=none] (523) at (11, -25) {};
		\node [style=none] (524) at (9.5, -25.75) {$\ket{n}$};
		\node [style=none] (526) at (8.25, -25) {$\sqrt{n!}$};
		\node [style=phase] (530) at (-11.75, -25.25) {$x$};
		\node [style=phase] (531) at (-2, -24.5) {$x$};
		\node [style=phase] (532) at (-2, -26) {$x$};
		\node [style=phase] (537) at (5, -1.5) {$x$};
		\node [style=phase] (538) at (3.5, -6) {$x$};
		\node [style=phase] (540) at (3.5, -7) {$y$};
		\node [style=phase] (541) at (9.75, -6.5) {$x + y$};
		\node [style=phase] (542) at (15.25, -6.5) {$x$};
		\node [style=phase] (543) at (16.75, -6.5) {$y$};
		\node [style=phase] (544) at (21.25, -6.5) {$x \cdot y$};
		\node [style=phase] (545) at (17, -1.5) {$r$};
		\node [style=phase] (546) at (16, -11.25) {$0$};
		\node [style=phase] (547) at (4.5, -11.25) {$1$};
		\node [style=none] (548) at (7.25, -16) {$1$};
		\node [style=none] (549) at (8.5, -16) {$=$};
		\node [style=none] (550) at (9.5, -15.625) {};
		\node [style=none] (551) at (9.5, -16.375) {};
		\node [style=none] (552) at (10.25, -16.375) {};
		\node [style=none] (553) at (10.25, -15.625) {};
	\end{pgfonlayer}
	\begin{pgfonlayer}{edgelayer}
		\draw (12.center) to (13.center);
		\draw (13.center) to (14.center);
		\draw (14.center) to (12.center);
		\draw (15.center) to (13.center);
		\draw (12.center) to (17.center);
		\draw (14.center) to (16.center);
		\draw (18.center) to (19.center);
		\draw (19.center) to (20.center);
		\draw (20.center) to (18.center);
		\draw (19.center) to (21.center);
		\draw (20.center) to (22.center);
		\draw (23.center) to (24.center);
		\draw (24.center) to (25.center);
		\draw (25.center) to (23.center);
		\draw (26.center) to (23.center);
		\draw (32.center) to (33.center);
		\draw (33.center) to (34.center);
		\draw (34.center) to (32.center);
		\draw (35.center) to (36.center);
		\draw (36.center) to (37.center);
		\draw (37.center) to (35.center);
		\draw (38.center) to (35.center);
		\draw (39.center) to (40.center);
		\draw (40.center) to (41.center);
		\draw (41.center) to (39.center);
		\draw [bend left=15] (24.center) to (33.center);
		\draw (25.center) to (40.center);
		\draw (34.center) to (36.center);
		\draw [bend right=15] (37.center) to (41.center);
		\draw (39.center) to (43.center);
		\draw (32.center) to (42.center);
		\draw (57.center) to (58.center);
		\draw (58.center) to (59.center);
		\draw (59.center) to (57.center);
		\draw (60.center) to (57.center);
		\draw (63.center) to (64.center);
		\draw (64.center) to (65.center);
		\draw (65.center) to (63.center);
		\draw (63.center) to (68.center);
		\draw [bend left=15] (58.center) to (64.center);
		\draw [bend right=15] (59.center) to (65.center);
		\draw (110.center) to (111.center);
		\draw (133.center) to (134.center);
		\draw (136.center) to (137.center);
		\draw (138.center) to (139.center);
		\draw (140.center) to (141.center);
		\draw (141.center) to (142.center);
		\draw (142.center) to (140.center);
		\draw (141.center) to (143.center);
		\draw (142.center) to (144.center);
		\draw (147.center) to (148.center);
		\draw (149.center) to (150.center);
		\draw (150.center) to (151.center);
		\draw (151.center) to (149.center);
		\draw (150.center) to (152.center);
		\draw (151.center) to (153.center);
		\draw (175.center) to (176.center);
		\draw (177.center) to (178.center);
		\draw (178.center) to (179.center);
		\draw (179.center) to (177.center);
		\draw (178.center) to (180.center);
		\draw (179.center) to (181.center);
		\draw (202.center) to (203.center);
		\draw (204.center) to (205.center);
		\draw (206.center) to (207.center);
		\draw (207.center) to (208.center);
		\draw (208.center) to (206.center);
		\draw (207.center) to (209.center);
		\draw (208.center) to (210.center);
		\draw (215.center) to (216.center);
		\draw (217.center) to (218.center);
		\draw (218.center) to (219.center);
		\draw (219.center) to (217.center);
		\draw (218.center) to (220.center);
		\draw (219.center) to (221.center);
		\draw (200) to (224.center);
		\draw (222) to (225.center);
		\draw (239.center) to (240.center);
		\draw (243.center) to (244.center);
		\draw (270.center) to (271.center);
		\draw (272.center) to (273.center);
		\draw (273.center) to (274.center);
		\draw (274.center) to (272.center);
		\draw (273.center) to (275.center);
		\draw (274.center) to (276.center);
		\draw (279.center) to (280.center);
		\draw (280.center) to (281.center);
		\draw (281.center) to (279.center);
		\draw (280.center) to (282.center);
		\draw (281.center) to (283.center);
		\draw (284.center) to (285.center);
		\draw (286.center) to (287.center);
		\draw (287.center) to (288.center);
		\draw (288.center) to (286.center);
		\draw (287.center) to (289.center);
		\draw (288.center) to (290.center);
		\draw (292.center) to (293.center);
		\draw (293.center) to (294.center);
		\draw (294.center) to (292.center);
		\draw (293.center) to (295.center);
		\draw (294.center) to (296.center);
		\draw (323.center) to (324.center);
		\draw (325.center) to (326.center);
		\draw (326.center) to (327.center);
		\draw (327.center) to (325.center);
		\draw [in=180, out=30] (326.center) to (328.center);
		\draw [in=-180, out=-30] (327.center) to (329.center);
		\draw (330.center) to (331.center);
		\draw (332.center) to (333.center);
		\draw (333.center) to (334.center);
		\draw (334.center) to (332.center);
		\draw (333.center) to (335.center);
		\draw (334.center) to (336.center);
		\draw (380.center) to (381.center);
		\draw (414.center) to (415.center);
		\draw (418.center) to (419.center);
		\draw (421.center) to (422.center);
		\draw (424.center) to (425.center);
		\draw (427.center) to (428.center);
		\draw (430.center) to (431.center);
		\draw (431.center) to (432.center);
		\draw (432.center) to (430.center);
		\draw (431.center) to (433.center);
		\draw (432.center) to (434.center);
		\draw (435.center) to (430.center);
		\draw (438.center) to (437.center);
		\draw (440.center) to (439.center);
		\draw (445.center) to (444.center);
		\draw (447.center) to (446.center);
		\draw (475.center) to (474.center);
		\draw (481.center) to (480.center);
		\draw (490.center) to (489.center);
		\draw (510.center) to (511.center);
		\draw (511.center) to (512.center);
		\draw (512.center) to (510.center);
		\draw (513.center) to (510.center);
		\draw (516.center) to (512.center);
		\draw (511.center) to (515.center);
		\draw (523.center) to (522);
		\draw [style=dashedE] (551.center)
			 to (550.center)
			 to (553.center)
			 to (552.center)
			 to cycle;
	\end{pgfonlayer}
\end{tikzpicture}
}
\end{equation*}

\section{Kraus decomposition of Type-I fusion}\label{sec:mixedfusion}

To prove \cref{prop:type-I}, we derive diagrams in \textbf{QPath} for each measurement outcome of the linear optical circuit implementing Type I fusion:
\begin{equation*}
    D^{\cvar{a},\cvar{b}}
    \quad\coloneqq\quad
    \input{RUS/rus0b.tikz}
\end{equation*}
Because the input to $D^{\cvar{a},\cvar{b}}$ is two dual rail qubits and hence at most two photons, and no photons are created in this process, we can restrict our attention to only measurement outcomes observing at most two photons.

From to the definitions of the dual-rail encoding \cref{eq:dual-rail-encoding}, we can define the follwing decomposition:
\begin{equation}
  \label{eq:dual-rail-decomposition}
  \input{RUS/dualrail-decomp.tikz}
\end{equation}
Using this, we can now prove a set of lemmas for the different cases of $D^{\cvar{a},\cvar{b}}$.
First, $D^{\cvar{a}=0,\cvar{b}=0}$ evaluates to:
\begin{equation}
    \input{RUS/d00.tikz}
\end{equation}

\noindent Continuing on, we compute $D^{\cvar{a}=1,\cvar{b}=0}$:
\begin{equation}
    \input{RUS/d10.tikz}
\end{equation}
Similarly, $D^{\cvar{a}=0,\cvar{b}=1}$ computes to:
\begin{align}
    \input{RUS/d01-0.tikz}\quad
    &\input{RUS/d01-1.tikz}\\
    &\input{RUS/d01-2.tikz}\\
    &\input{RUS/d01-3.tikz}
\end{align}
We can show that $D^{\cvar{a}=1,\cvar{b}=1}$ evaluates to:
\begin{equation}
    \input{RUS/d11.tikz}
\end{equation}
This means that the probability of observing $D^{\cvar{a}=1,\cvar{b}=1}$ is zero which is due to the Hong-Ou-Mandel effect.

For the last two cases of $D^{\cvar{a},\cvar{b}}$, we use the following equality:
\begin{equation}
  \input{RUS/lem2.tikz}
\end{equation}
Proceeding with the calculations, we compute $D^{\cvar{a}=0,\cvar{b}=2}$ as follows:
\begin{equation}
    \input{RUS/d02.tikz}
\end{equation}
and finally, $D^{\cvar{a}=2,\cvar{b}=0}$ evaluates to:
\begin{equation}
    \input{RUS/d20.tikz}
\end{equation}

To realize $D^{\cvar{s},\cvar{k}}$, we determine the action of a classical function from measurement outcomes to the Boolean variable indicators of success ($\cvar{s}$) and correction ($\cvar{k}$):
\begin{center}
    \begin{tabular}{ll|ll} 
        \hline
        \cvar{a} & \cvar{b} & \cvar{s} & \cvar{k}\\
        \hline
        0 & 0 & 0 & 1\\
        1 & 0 & 1 & 0\\
        0 & 1 & 1 & 1\\
        2 & 0 & 0 & 0\\
        0 & 2 & 0 & 0\\
        \hline
    \end{tabular} 
\end{center}
For three of the five possible values of $(\cvar{s},\cvar{k})$, the original measurement outcomes $(\cvar{a},\cvar{b})$ can be identified.
For the case $(\cvar{s}=0,\cvar{k}=1)$, it happens that $D^{\cvar{a}=2,\cvar{b}=0} = D^{\cvar{a}=0,\cvar{b}=2}$ up to a global phase which is thereafter eliminated upon invoking the CPM construction.
Because of this, in mixed-state quantum mechanics we have
\begin{equation}
    D^{\cvar{s}=0,\cvar{k}=1} = D^{\cvar{a}=2,\cvar{b}=0} + D^{\cvar{a}=0,\cvar{b}=2} = 2 \,\, D^{\cvar{a}=2,\cvar{b}=0}
\end{equation}
Therefore, $D^{\cvar{s},\cvar{k}}$ is a non-destructive measurement that is a sum of four terms, one for each possible value of $(\cvar{s},\cvar{k})$.
By performing a mixed sum over all possible measurement outcomes of $D^{\cvar{a},\cvar{b}}$, and quotienting by $\cvar{s}$ and $\cvar{k}$, we obtain the proposition.

\typeOneFusionProp*
\begin{proof}
  \[
    \input{RUS/mixed-fusion-thm.tikz}
  \]
\end{proof}

\section{Proof of characterisation}\label{app:fusion}

In section, we prove \cref{thm-characterisation}, characterizing all fusion measurements with green failure and Pauli error.

\characterization*

We are interested in fusions with green failure satisfying the following equation:
\begin{equation}
    \input{ZX/PauliFusionCommutation.tikz}
\end{equation}
for some $w,x,y,z \in \{0,1\}$.

First note that in order for the $\cvar{j} \pi$ error to commute to the inputs as a Pauli error, 
$\phi$ must satisfy the commutation relation
\begin{equation}
    \input{ZX/PauliFusionCase34.tikz}
\end{equation}
for some $\cvar{p}, \cvar{q} \in \{0,1\}$. That is, $\phi$ must be a Clifford phase. 
Therefore, necessarily $\phi = v \frac{\pi}{2}$ for some $v \in \{0,1,2,3\}$.
We thus consider two cases: (i) $v$ is odd and (ii) $v$ is even.

If $v$ is even, we can take $\phi = 0$; the $\phi = \pi$ case is redundant because of the random error $\cvar{k} \pi$.
In this case, the $\cvar{j} \pi$ error induces $Z$ errors on both input qubits. 
In order to correct $\cvar{k} \pi$ error, one of $\alpha_1$ or $\alpha_2$ must be Clifford. Without loss of generality, 
we assume $\alpha_1$ is a Clifford phase, i.e. $\alpha_1 = d \frac{\pi}{2}$, and the result follows with $\lambda = \XZm$.
\[ 
  \input{ZX/PauliFusionCase1.tikz}
\]

If $v$ is odd, we can take $\phi = \frac{\pi}{2}$; the $\phi = -\frac{\pi}{2}$ case is redundant because of the random error $\cvar{k}$.
Then, the $\cvar{j} \pi$ error induces a $Y$ error when commuting with $\phi$. Since $Y = XZ$, this induces an $Z$ error 
on both input qubits and an $X$ error which merges with the $\cvar{k} \pi$ error. In order to correct the resulting 
$(\cvar{k} \oplus \cvar{j}) \pi$ error, either $\alpha_1$ or $\alpha_2$ must be a an integer multiple of $\frac{\pi}{2}$. 
Without loss of generality we can set $\alpha_1 = d \frac{\pi}{2}$, and the result follows with $\lambda = \YZm$.
\[ 
  \input{ZX/PauliFusionCase2.tikz}
\]

The third case, where $\lambda = \XYm$ is obtained when $\beta = \frac{\pi}{2} + \alpha$ for some angle $\alpha$ and 
$\phi = \frac{\pi}{2}$. This is technically an instance of the case where $v$ is odd, but we distinguish 
it from the other cases because the errors are different, as they are induced by a single-qubit measurement in the $\XYm$ plane. 
\[ 
  \input{ZX/PauliFusionCase3.tikz}
\]

\section{Proof of flow preserving rewrites}\label{app:flow}

In this appendix, we show that the target open graph and the simplified target graph of an XY-fusion network are equivalent, and the existence of flow on one implies that of the other.
To prove this, we use rewrites that preserves the existence of Pauli flow~\cite{simmonsRelatingMeasurementPatterns2021,mcelvanneyCompleteFlowpreservingRewrite2023, mcelvanneyFlowpreservingZXcalculusRewrite2023}.

We use an alternative notation to simplify the diagrams, and replace a Hadamard
between two spiders by a blue dashed edge, as illustrated below.
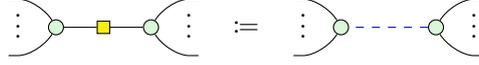
\[
    \begin{tikzpicture}
	\begin{pgfonlayer}{nodelayer}
		\node [style=gn] (1) at (-5, 0) {};
		\node [style=hbox] (2) at (-3.75, 0) {};
		\node [style=none] (5) at (-6, 0.75) {};
		\node [style=none] (6) at (-6, -0.75) {};
		\node [style=none] (7) at (-6.25, 0.75) {};
		\node [style=none] (8) at (-6.25, -0.75) {};
		\node [style=none] (9) at (-6, 0.25) {$\vdots$};
		\node [style=gn] (10) at (-2.5, 0) {};
		\node [style=none] (11) at (-1.5, 0.75) {};
		\node [style=none] (12) at (-1.5, -0.75) {};
		\node [style=none] (13) at (-1.25, 0.75) {};
		\node [style=none] (14) at (-1.25, -0.75) {};
		\node [style=none] (15) at (-1.5, 0.25) {$\vdots$};
		\node [style=none] (16) at (0, 0) {$\coloneqq$};
		\node [style=gn] (17) at (2.5, 0) {};
		\node [style=none] (19) at (1.5, 0.75) {};
		\node [style=none] (20) at (1.5, -0.75) {};
		\node [style=none] (21) at (1.25, 0.75) {};
		\node [style=none] (22) at (1.25, -0.75) {};
		\node [style=none] (23) at (1.5, 0.25) {$\vdots$};
		\node [style=gn] (24) at (5, 0) {};
		\node [style=none] (25) at (6, 0.75) {};
		\node [style=none] (26) at (6, -0.75) {};
		\node [style=none] (27) at (6.25, 0.75) {};
		\node [style=none] (28) at (6.25, -0.75) {};
		\node [style=none] (29) at (6, 0.25) {$\vdots$};
	\end{pgfonlayer}
	\begin{pgfonlayer}{edgelayer}
		\draw (2) to (1);
		\draw [in=-120, out=0, looseness=0.75] (6.center) to (1);
		\draw [in=0, out=120, looseness=0.75] (1) to (5.center);
		\draw (7.center) to (5.center);
		\draw (8.center) to (6.center);
		\draw [in=-60, out=180, looseness=0.75] (12.center) to (10);
		\draw [in=180, out=60, looseness=0.75] (10) to (11.center);
		\draw (13.center) to (11.center);
		\draw (14.center) to (12.center);
		\draw (10) to (2);
		\draw [in=-120, out=0, looseness=0.75] (20.center) to (17);
		\draw [in=0, out=120, looseness=0.75] (17) to (19.center);
		\draw (21.center) to (19.center);
		\draw (22.center) to (20.center);
		\draw [in=-60, out=180, looseness=0.75] (26.center) to (24);
		\draw [in=180, out=60, looseness=0.75] (24) to (25.center);
		\draw (27.center) to (25.center);
		\draw (28.center) to (26.center);
		\draw [style=hadamard edge] (24) to (17);
	\end{pgfonlayer}
\end{tikzpicture}

\]
Both the blue edge notation and the Hadamard box can always be translated back
into spiders when necessary.
We refer to the blue edge as a Hadamard edge.

\begin{lemma}[Copy]
  \label{lem:copy}
  Copying preserves the existence of Pauli flow~\cite[Lemma D.6]{simmonsRelatingMeasurementPatterns2021}.
  Graphically, this corresponds to the copy rule of the ZX calculus~\cite[Lemmas 2.7 and 2.8 ]{mcelvanneyFlowpreservingZXcalculusRewrite2023}:
  \[
    \input{fusion/CopyFlowPreserving.tikz}
  \]
\end{lemma}

\begin{lemma}[Local Complementation]
  \label{lem:lc}
  Local complementation about a vertex $u$ preserves the existence of Pauli flow~\cite[Lemma D.12]{simmonsRelatingMeasurementPatterns2021}.
  In the ZX calculus, this rule is also called local complementation, and is given as follows~\cite[Lemma 2.10]{mcelvanneyFlowpreservingZXcalculusRewrite2023}:
  \[
    \input{fusion/LCFlowPreserving.tikz}
  \]
\end{lemma}

\begin{lemma}[Pivot]
  \label{lem:pivot}
  Pivoting about an edge $(u, v)$ preserves the existence of Pauli flow~\cite[Lemma D.21]{simmonsRelatingMeasurementPatterns2021}.
  In the ZX calculus, this rule is also called pivoting, and is given as the following rewrite rule~\cite[Lemma 2.11]{mcelvanneyFlowpreservingZXcalculusRewrite2023}:
  \[
    \input{fusion/PivotFlowPreserving.tikz}
  \]
\end{lemma}

\begin{lemma}[State Change]
  \label{lem:state-change}
  \[
    \normalfont
    \input{ZX/state-change.tikz}
  \]
\end{lemma}

\xFusionFlowPreserving*
\begin{proof}
  \[
    \input{fusion/XFusionFlowPreservingProof.tikz}
  \]
  Since each of the rewrites preserves the existence of Pauli flow, the additional errors appearing above can always be corrected, but we have given them here for completeness.
\end{proof}

\yFusionFlowPreserving*
\begin{proof}
  \[
    \input{fusion/YFusionFlowPreservingProof.tikz}
  \]
\end{proof}

\section{Proofs of repeat-until-success}\label{app:RUS}

We prove the results of \cref{subsec:RUS}.

\RUS*
\begin{proof}
    We prove this inductively. The $n = 1$ case follows from \cref{prop:fusion-stochastic}:
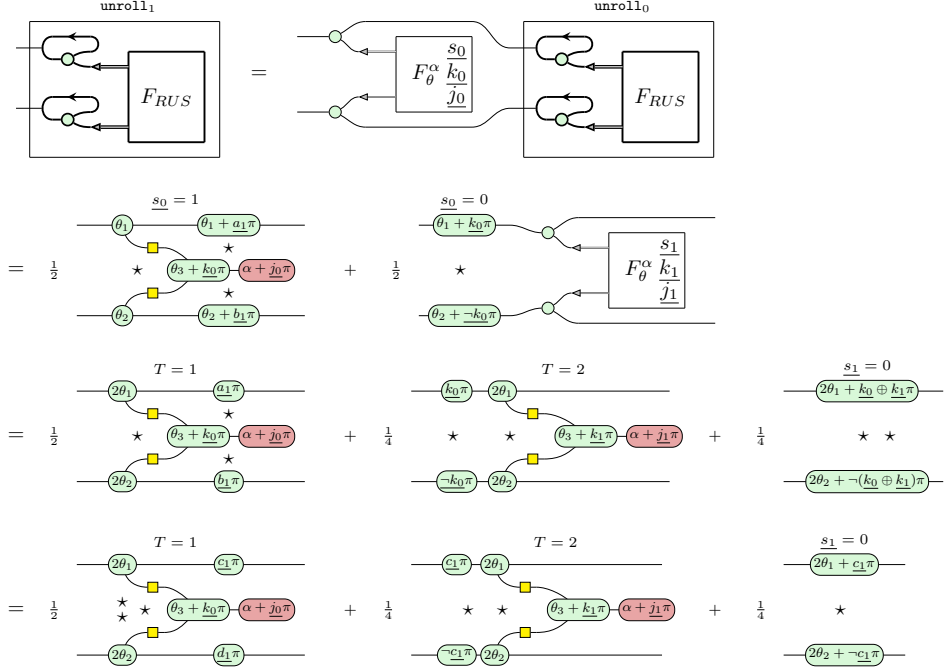
    \[ 
        \scalebox{.8}{\input{repeat-proof-basecase.tikz}}
    \]
    Then, given the statement for $n$, we show it for $n + 1$:
    \allowdisplaybreaks
    \begin{align*}
        &\scalebox{.8}{\input{repeat-prooff-1.tikz}}\\
        &\scalebox{.8}{\input{repeat-prooff-2.tikz}}\\
        &\scalebox{.8}{\input{repeat-prooff-3.tikz}}\\
        &\scalebox{.8}{\input{repeat-prooff-4.tikz}}\\
        &\scalebox{.8}{\input{repeat-prooff-5.tikz}}\\
        &\scalebox{.8}{\input{repeat-prooff-5+.tikz}}
    \end{align*}
\end{proof}

\XRUS*
\begin{proof}
  This follows from \cref{thm:RUS} and the following derivation: 
  \[
    \scalebox{.8}{\input{repeat-X-proof.tikz}}
  \]
\end{proof}

\YRUS*
\begin{proof}
  Follows from \cref{thm:RUS} and:
  \[
    \scalebox{.8}{\input{repeat-Y-proof.tikz}}
  \]
  where $x_T = k_T \oplus j_T$.
\end{proof}

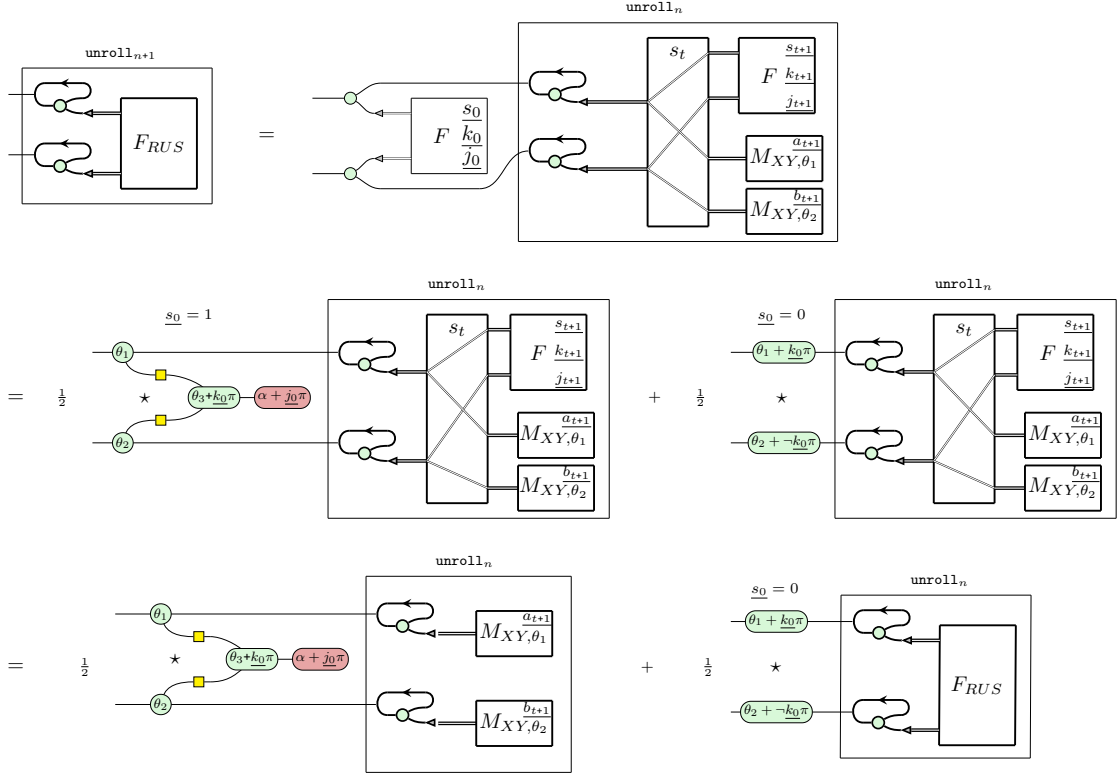
\end{document}